%% file: main2.tex
\pdfoutput=1

\newif\iflongversion

\longversiontrue 

\iflongversion
\documentclass[12pt,draftcls,onecolumn]{IEEEtran}
\else 
\documentclass{IEEEtran}
\fi

\usepackage{}              
\usepackage{graphicx}
\usepackage{amsmath,amsthm,amsfonts,amssymb}
\usepackage{floatrow}
\usepackage{caption}
\usepackage{placeins}
\usepackage{makecell}
\usepackage{subcaption}
\usepackage{multirow}
\usepackage{array}
\usepackage{mathtools}
\usepackage{color}
\usepackage{bbm}
\usepackage{wasysym}
\usepackage[dvipsnames]{xcolor}
\usepackage{tikz}
\usepackage{enumerate}
\usepackage{epstopdf} 
\usepackage{setspace} 
\usepackage{cite}

\usetikzlibrary{arrows,positioning,shapes,backgrounds,fit}  
\usetikzlibrary{positioning}

\usepackage{blkarray}

\DeclareMathOperator*{\argmax}{\arg\max}   

\newtheorem{theorem}{Theorem}
\newtheorem{definition}{Definition}

\newtheorem{assumption}{Assumption}
\newtheorem{lemma}{Lemma}

\newtheorem{remark}{Remark}

\newcommand{\red}[1]{\textcolor{black}{#1}}
\newcommand{\blue}[1]{\textcolor{black}{#1}}

\begin{document}
	\title{A Unified Approach to Dynamic Decision Problems with Asymmetric Information - \\Part II:	Strategic Agents}   
	\author{Hamidreza Tavafoghi, Yi Ouyang, and Demosthenis Teneketzis}         
	
	\date{November 23, 2018}     
	\maketitle
	{\let\thefootnote\relax\footnote{ A preliminary version of this paper will appear in
			the Proceeding of the 57th IEEE Conference on Decision and Control
			(CDC), Miami Beach, FL, December 2018 \cite{CDC18}.\\ H. Tavafoghi is with the Department of Mechanical Engineering at the University of California, Berkeley (e-mail:
			tavaf@berkeley.edu). Y. Ouyang is with Preferred Networks America, Inc. (e-mail: ouyangyi@preferred-america.com). D. Teneketzis is with the Department of Electrical Engineering and Computer Science at the University of Michigan, Ann Arbor (e-mail:
			teneket@umich.edu)\\
			This work was supported in part by the NSF grants CNS-1238962, CCF-1111061, ARO-MURI grant W911NF-13-1-0421, and ARO grant W911NF-17-1-0232.\\
		}}
		\begin{abstract}
			We study a general class of dynamic games with asymmetric information where agents' beliefs are strategy dependent, \textit{i.e. signaling} occurs. We show that the notion of \textit{sufficient information}, introduced in the companion paper \cite{team}, can be used to effectively compress the agents' information in a mutually consistent manner that is sufficient for decision-making purposes. We present instances of dynamic games with asymmetric information where we can characterize a time-invariant information state for each agent. Based on the notion of sufficient information, we define a class of equilibria for dynamic games called Sufficient Information Based Perfect Bayesian Equilibrium (SIB-PBE). Utilizing the notion of SIB-PBE, we provide a sequential decomposition of dynamic games with asymmetric  information over time; this decomposition leads to a dynamic program that determines SIB-PBE of dynamic games. Furthermore, we provide conditions under which we can guarantee the existence of SIB-PBE. 
		\end{abstract}
		\vspace*{-5pt}

		\input{Introduction}
		\input{Model}
		\input{Equilibrium}

		\vspace*{-3pt}

		\input{PBE}
		\input{CIB-PBE}

		\input{CIBPBE}
		\input{Discussion}

		\vspace*{-5pt}
		\input{Existence}
		\vspace*{-7pt}
		\input{Conclusion}
		\vspace*{-9pt}
		\bibliographystyle{ieeetr}
		\bibliography{thesis-bib}

		\input{Appendix}
	\end{document}

%% file: Introduction.tex
\section{Introduction}\label{sec:introductionG}


We study a general class of stochastic dynamic games with asymmetric information. We consider a setting where the underlying system has Markovian dynamics controlled by the agents' joint actions at every time. The instantaneous utility of each agent depends on the agents' joint actions and the system state. At every time, each agent makes a private noisy observation that depends on the current system state and the agents' past actions. Therefore, at every time agents have asymmetric and imperfect information about the history of the game. 
Furthermore, the information that an agent possesses about the history of the game at each time instant depends on the other agents' past actions and strategies; this phenomenon is known as \textit{signaling} among the agents. Moreover, at each time, each agent's strategy depend on his information about the current system state and the other agents' strategies. Therefore, the agents' decisions and information are coupled and interdependent over time.

There are three main challenges in the study of dynamic games with asymmetric information. First, since the agents' decisions and information are interdependent and coupled over time, we need to determine the agents' strategies simultaneously for all times. Second, the agents' strategy domains grow over time as they acquire more information. Third, in contrast to dynamic teams where agents coordinate their strategies, in dynamic games each agent's strategy is his private information as he chooses it individually so as to maximize his utility. Therefore, in dynamic games each agent needs to form a belief about other agents' strategies as well as about the game history. 

In this paper, we propose a general approach for the study of a dynamic games with asymmetric information that addresses the stated-above challenges. We build our approach based on notion of sufficient information, introduced in the companion paper \cite{team}, and define a class of sufficient information based assessments, where strategic agents compress their information in a mutually consistent manner that is sufficient for decision-making purposes. Accordingly, we propose the notion of Sufficient Information Based Perfect Bayesian Equilibrium (SIB-PBE) for dynamic games that characterizes a set of equilibrium outcomes. 
Using the notion of SIB-PBE, we provide a sequential decomposition of the game over time, and formulate a dynamic program that enables us to compute the set of SIB-PBEs via backward induction. We discover specific instances of dynamic games where we can determine a set of information states for the agents that have time-invariant domain. We determine conditions that guarantee the existence of SIB-PBEs. We discuss the relation between the class of SIB-PBE and PBE in dynamic games, and argue that the class of SIB-PBE provides a \textit{simpler} and \textit{more robust} set of equilibria than PBE that are consistent with agents' rationality. 

	\red{The notion of SIB-PBE we introduce in this paper provides a generalization/extension of Markov Perfect Equilibrium (MPE) to dynamic games with asymmetric information. The authors in \cite{maskin2001markov} introduce the notion of Markov Perfect Equilibrium that characterizes a subset of Subgame Perfect Equilibria (SPE) for dynamic games with symmetric information and provide a sequential decomposition of the game over time.}  Moreover, our results, along with those in the companion paper \cite{team}, provide a unified approach to the study of dynamic decision problems with asymmetric information and strategic or nonstrategic agents.   

		\vspace*{-3pt}

\subsection{Related Literature}

Dynamic games with asymmetric information have been investigated extensively in literature in the context of repeated games; see \cite{zamir1992repeated,forges1992repeated,aumann1995repeated,mailath2006repeated} and the references therein. The key feature of these games is the absence of a dynamic system. Moreover, the works on repeated games study primarily their asymptotic properties when the horizon is infinite and agents are sufficiently patient (\textit{i.e.} the discount factor is close one). In repeated games, agents play a stage (static) game repeatedly over time. As a result, 
the decision making problem that each agent faces is very simple. The main objective of this strand of literature is to explore situations where agents can form self-enforcing punishment/reward mechanisms so as to create additional equilibria that improve upon the payoffs agents can get by simply playing an equilibrium of the stage game over time. Recent works (see \cite{horner2011recursive,escobar2013efficiency,sugaya2012}) adopt approaches similar to those used in repeated games to study infinite horizon dynamic games with asymmetric information when there is an underlying dynamic Markovian system. Under certain conditions on the system dynamics and information structure, the authors of  \cite{horner2011recursive,escobar2013efficiency,sugaya2012} characterize a set of asymptotic equilibria when the agents are sufficiently patient.

The problem we study in this paper is different from the ones in \cite{zamir1992repeated,forges1992repeated,aumann1995repeated,mailath2006repeated,horner2011recursive,escobar2013efficiency,sugaya2012} in two aspects. First, we consider a class of dynamic games where the underlying system has a general Markovian dynamics and a general information structure, and we do not restrict attention to asymptotic behaviors when the horizon is infinite and the agents are sufficiently patient. Second,  we study situations where the decision problem that each agent faces, in the absence of strategic interactions with other agents, is a Partially Observed Markov Decision Process (POMDP), which is a complex problem to solve by itself. Therefore, reaching (and computing)  a set of equilibrium strategies, which take into account the strategic interactions among the agents, is a very challenging task. As a result, it is not very plausible for the agents to seek reaching an equilibria that is generated by the formation of self-enforcing punishment/reward mechanisms similar to those used in infinitely repeated games (see Section \ref{sec:discussion} for more discussion). We believe that our results provide new insight into the behavior of strategic agents in complex and dynamic environments, and complement the existing results in the repeated games literature with simple and (mostly) static environments.

The works in \cite{renault2006value,cardaliaguet2015markov,gensbittel2015value,li2017solving} consider dynamic zero-sum games with asymmetric information. The authors of \cite{cardaliaguet2015markov,renault2006value} study zero-sum games with Markovian dynamics and lack of information on one side (\textit{i.e.} one informed player and one uninformed player). The authors of \cite{gensbittel2015value,li2017solving} study zero-sum games with Markovian dynamics \red{and} lack of information on both sides. The problem that we study in this paper is different from the ones in \cite{renault2006value,cardaliaguet2015markov,gensbittel2015value,li2017solving} in three aspects. First, we study a general class of dynamic games that include dynamic zero-sum games with asymmetric information as a special case. Second, we consider a general Markovian dynamics for the underlying system whereas the authors of \cite{cardaliaguet2015markov,renault2006value,gensbittel2015value,li2017solving} consider a specific Markovian dynamics where each agent observes perfectly a local state that evolves independently of the other local states conditioned on the agents' observable actions. Third, we consider a general information structure that allows us to capture scenarios with unobservable actions and imperfect observations that are not captured in \cite{cardaliaguet2015markov,renault2006value,gensbittel2015value,li2017solving}.   

The problems investigated in \cite{nayyar2014common,gupta2014common,ouyang2015CDC,ouyang2016TAC,vasal2016signaling,sinha2016structured} are the most closely related to our problem. The authors of \cite{nayyar2014common,gupta2014common} study a class of dynamic games where the agents' common information based belief (defined in \cite{nayyar2014common}) is independent of their strategies; that is, there is no signaling among them. This property allows them to apply ideas from the common information approach developed in \cite{nayyar2011optimal,nayyar2013decentralized}, and define an equivalent dynamic game with symmetric information among the fictitious agents. Consequently, they characterize a class of equilibria for dynamic games called \textit{Common Information based Markov Perfect Equilibrium}. Our results are different from those in \cite{nayyar2014common,gupta2014common} in two aspects. First, we consider a general class of dynamic games where the agents' CIB beliefs are strategy-dependent, thus,  signaling is present. Second, the proposed approach in  \cite{nayyar2014common,gupta2014common} requires the agents to keep track of all of their private information over time. We propose an approach to effectively compress the agents' private information, and consequently, reduce the number of variables on which the agents need to form CIB beliefs.

The authors of \cite{ouyang2015CDC,ouyang2016TAC,vasal2016signaling,sinha2016structured} study a class of dynamic games with asymmetric information where signaling occurs.  When the horizon in finite, the authors of \cite{ouyang2015CDC,ouyang2016TAC} introduce the notion of Common Information Based Perfect Bayesian Equilibrium, and provide a sequential decomposition of the game over time. The authors of \cite{vasal2016signaling,sinha2016structured} extend the results of \cite{ouyang2015CDC,ouyang2016TAC} to finite horizon Linear-Quadratic-Gaussian (LQG) dynamic games and infinite horizon dynamic games, respectively.  
The class of dynamic games studied in \cite{ouyang2015CDC,ouyang2016TAC,vasal2016signaling,sinha2016structured} satisfies the following assumptions: (i) agents' actions are observable (ii) each agent has a perfect observation of his own local states/type (iii) conditioned on the agents' actions, the evolution of the local states are independent.

We relax assumptions (i)-(iii) of \cite{ouyang2015CDC,ouyang2016TAC,vasal2016signaling,sinha2016structured}, and study a general class of dynamic games with asymmetric information, hidden actions, imperfect observations, and controlled and coupled dynamics. As a result, each agent needs to form a belief about the other agents' past actions and private (imperfect) observations. Moreover, in contrast to \cite{ouyang2015CDC,ouyang2016TAC,vasal2016signaling,sinha2016structured}, an agent's, say agent $i$'s, belief about the system state and the other agents' private information is his own private information and is different from the CIB belief. In this paper, we extend the methodology developed in \cite{ouyang2015CDC,ouyang2016TAC} for dynamic games, and generalize the notion of CIB-PBE. Furthermore, we propose an approach to effectively compress the agents' private information and obtain the results of \cite{ouyang2015CDC,ouyang2016TAC,vasal2016signaling,sinha2016structured} as special cases.

		\vspace*{-7pt}

\subsection{Contribution}\vspace*{-1pt}
We develop a general methodology for the study and analysis of dynamic games with asymmetric information, where the information structure is non-classical; that is, signaling occurs. We propose an approach to characterize a set of information states that effectively compress the agents' private and common information in a mutually consistent manner. We characterize a subclass of Perfect Bayesian Equilibria, called SIB-PBE, and provide a sequential decomposition of these games over time. This decomposition provides a backward induction algorithm to determine the set of SIB-PBEs. We discover special instances of dynamic games where we can identify a set of information states with time-invariant domain. We provide conditions that guarantee the existence of SIB-PBEs in dynamic games with asymmetric information. We show that the methodology developed in this paper generalizes the existing results on dynamic games with non-classical information structure. 

		\vspace*{-7pt}
		
		\subsection{Organization}\vspace*{-1pt}
The rest of the paper is organized as follows. In Section \ref{sec:model}, we describe our model. In Section \ref{sec:equilibrium}, we discuss the main issues that arise in the study of dynamic games with asymmetric information. We provide the formal definition of Perfect Bayesian Equilibrium in Section \ref{sec:PBE}. In Section \ref{sec:CIB}, we describe the sufficient information approach to dynamic games with asymmetric information and introduce the notion of Sufficient Information Based (SIB) assessment and SIB-PBE. In Section \ref{sec:CIB-PBE}, we present our main results and provide a sequential decomposition of dynamic games over time. 
We discuss our results in Section \ref{sec:discussion}, and compare the notion of SIB-PBE with other equilibrium concepts..
In Section \ref{sec:existence}, we determine conditions that guarantee the existence of SIB-PBE in dynamic games with asymmetric information. We conclude in Section \ref{sec:conclusionG}. 
The proofs of all the theorems and lemmas appear in the Appendix. 
\blue{\begin{remark}
	Section \ref{sec:notation} on notation and Section \ref{subsec:privatecompress} on the definition of sufficient private information are similar to the ones appearing in the companion paper \cite{companion}; moreover, the model presented in Section \ref{sec:model} with strategic agents is similar to that of the companion paper\cite{companion} with non-strategic agents.		 
	 All these sections are included in this paper for ease of reading and to make the paper self-contained. 
	\end{remark} } 

		\vspace*{-4pt}
		
\subsection{Notation}\label{sec:notation}\vspace*{-1pt}
Random variables are denoted by upper case letters, their realizations by the corresponding lower case letters.
In general, subscripts are used as time index while superscripts are used to index agents.
For $t_1\hspace*{-2pt}\leq \hspace*{-2pt}t_2$, $X_{t_1:t_2}$ (resp. $f_{t_1:t_2}(\cdot)$) is the short hand notation for the random variables $(X_{t_1},\hspace*{-1pt}X_{t_1+1},...,\hspace*{-1pt}X_{t_2})$ (resp.  functions $(f_{t_1}(\cdot),\dots,\hspace*{-1pt}f_{t_2}(\cdot))$).
When we consider a sequence of random variables (resp. functions) for all time, we drop the subscript and use $X$ to denote $X_{1:T}$ (resp. $f(\cdot)$ to denote $f_{1:T}(\cdot)$).
For random variables $X^1_t,\dots,\hspace*{-1pt}X^N_t$ (resp. functions $f^1_t(\cdot),\dots,\hspace*{-1pt}f^N_t(\cdot)$), we use $X_t\hspace*{-2pt}:=\hspace*{-2pt}(X^1_t,\dots,\hspace*{-1pt}X^N_t)$ (resp. $f_t(\cdot)\hspace*{-2pt}:=\hspace*{-2pt}(f^1_t(\cdot),\dots,\hspace*{-1pt}f^N_t(\cdot))$) to denote the vector of the set of random variables (resp. functions) at $t$, and $X^{-n}_t\hspace*{-2pt}:=\hspace*{-2pt}(X^1_t,\dots,\hspace*{-1pt}X^{n-1}_t,\hspace*{-1pt}X^{n+1}_t,\dots,\hspace*{-1pt}X^N_t)$ (resp. $f^{-n}_t(\cdot)\hspace*{-2pt}:=\hspace*{-2pt}(f^1_t(\cdot),\dots, \hspace*{-1pt}f^{n-1}_t(\cdot),\hspace*{-1pt}f^{n+1}_t(\cdot),\dots,\hspace*{-1pt}f^N_t(\cdot))$) to denote all random variables (resp. functions) at $t$ except that of the agent indexed by $n$.
$\mathbb{P}\{\cdot\}$ and $\mathbb{E}\{\cdot\}$ denote the probability and expectation of an event and a random variable, respectively.
For a set $\mathcal{X}$, $\Delta(\mathcal{X})$ denotes the set of all beliefs/distributions on $\mathcal{X}$.
For random variables $X,Y$ with realizations $x,\hspace*{-1pt}y$, $\mathbb{P}\{x|y\} \hspace*{-2pt}:=\hspace*{-2pt} \mathbb{P}\{X\hspace*{-2pt}=\hspace*{-2pt}x|Y\hspace*{-2pt}=\hspace*{-2pt}y\}$ and $\mathbb{E}\{X|y\} \hspace*{-2pt}:= \hspace*{-2pt}\mathbb{E}\{X|Y\hspace*{-2pt}=\hspace*{-2pt}y\}$.
For a strategy $g$ and a belief (probability distribution) $\pi$, we use $\mathbb{P}^g_{\pi}\{\cdot\}$ (resp. $\mathbb{E}^g_{\pi}\{\cdot\}$) to indicate that the probability (resp. expectation) depends on the choice of $g$ and $\pi$. We use $\mathbf{1}_{\{X=x\}}$ to denote the indicator function for event $X\hspace*{-2pt}=\hspace*{-2pt}x$. For sets $A$ and $B$ we use $A\backslash B$ to denote all elements in set $A$ that are not in set $B$. 

		\vspace*{-4pt}

%% file: Model.tex
\vspace*{-2pt}
\section{Model}\label{sec:model}\vspace*{-1pt}

\textit{1) System dynamics:} There are $N$ strategic agents who live in a dynamic Markovian world over horizon $\mathcal{T}\hspace*{-2pt}:=\hspace*{-2pt}\{1,2,...,T\}$, $T\hspace*{-2pt}<\hspace*{-2pt}\infty$. Let $X_t\hspace*{-2pt}\in\hspace*{-2pt}\mathcal{X}_t$ denote the state of the world at $t\hspace*{-2pt}\in\hspace*{-2pt}\mathcal{T}$. At time $t$, each agent, indexed by $i\hspace*{-2pt}\in\hspace*{-2pt} \mathcal{N}\hspace*{-2pt}:=\hspace*{-2pt}\{1,2,...,N\}$, chooses an action $a^i_t\hspace*{-2pt}\in\hspace*{-2pt}\mathcal{A}^i_t$, where  $\mathcal{A}^i_t$ denotes the set of available actions to him at $t$. Given the collective action profile $A_t\hspace*{-2pt}:=\hspace*{-2pt}(A_t^1,...,A_t^N)$, the state of the world evolves according to the following stochastic dynamic equation,\vspace*{-2pt}
\begin{align}
X_{t+1}=f_t(X_t,A_t,W_t^x), \label{eq:systemdynamic1} \vspace*{-2pt}
\end{align} 
where $W_{1:T-1}^x$ is a sequence of independent random variables. The initial state $X_1$ is a random variable that has a probability distribution $\eta\in\Delta(\mathcal{X}_1)$ with full support.

At every time $t\in\mathcal{T}$, before taking an action, agent $i$ receives a noisy private observation $Y_t^i\in\mathcal{Y}_t^i$ of the current state of the world $X_t$ and the action profile $A_{t-1}$, given by\vspace*{-2pt}
\begin{align}
Y_t^i=O_t^i(X_t,A_{t-1},W_t^i), \label{eq:systemdynamic2}\vspace*{-2pt}
\end{align} 
where $W_{1:T}^i$, $i\in\mathcal{N}$, are sequences of independent random variables. Moreover, at every $t\in\mathcal{T}$, all agents receive a common observation $Z_t\in\mathcal{Z}_t$ of the current state of the world $X_t$ and the action profile $A_{t-1}$, given by\vspace*{-2pt}
\begin{align}
Z_t=O_t^c(X_t,A_{t-1},W_t^c), \label{eq:systemdynamic3}\vspace*{-3pt}
\end{align} 
where $W_{1:T}^c$, is a sequence of independent random variables. We note that the agents' actions $A_{t-1}$ is commonly observable at $t$ if $A_{t-1}\subseteq Z_t$.
We assume that the random variables $X_1$, $W_{1:T-1}^x$, $W_{1:T}^c$, and $W_{1:T}^i$, $i\in\mathcal{N}$ are mutually independent. 

\vspace{3pt}

\textit{2) Information structure:} Let $H_t\in\mathcal{H}_t$ denote the aggregate information of all agents at time $t$. Assuming that agents have perfect recall, we have $H_t=\{Z_{1:t},Y_{1:t}^{1:N},A_{1:t-1}^{1:N}\}$, \textit{i.e.} $H_t$ denotes the set of all agents' past observations and actions. The set of all possible realizations of the agents' aggregate information is given by $\mathcal{H}_t:=\prod_{\tau\leq t}\mathcal{Z}_\tau\times\prod_{i\in\mathcal{N}}\prod_{\tau\leq t}\mathcal{Y}_\tau^i\times \prod_{i\in\mathcal{N}}\prod_{\tau< t}\mathcal{A}_\tau^i$. 

At time $t\hspace*{-2pt}\in\hspace*{-2pt}\mathcal{T}$, the aggregate information $H_t$ is not fully known to all agents. 
Let $C_t\hspace*{-2pt}:=\hspace*{-2pt}\{Z_{1:t}\}\hspace*{-2pt}\in\hspace*{-2pt}\mathcal{C}_t$ denote the agents' common information  about $H_t$ and $P_t^i\hspace*{-2pt}:=\hspace*{-2pt}\{Y_{1:t}^i,A_{1:t-1}^i\}\backslash C_t\hspace*{-2pt}\in\hspace*{-2pt}\mathcal{P}_t^i$ denote agent $i$'s private information about $H_t$, where $\mathcal{P}_t^i$ and $\mathcal{C}_t$ denote the set of all possible realizations of agent $i$'s private and common information at time $t$, respectively. 
In Section \ref{sec:model:special}, we discuss several instances of information structures  that can be captured as special cases of our model.

\vspace{3pt}

\textit{3) Strategies and Utilities:} Let $H_t^i:=\{C_t,P_t^i\}\in \mathcal{H}_t^i$ denote the information available to agent $i$ at $t$, where $\mathcal{H}_t^i$ denote the set of all possible realizations of agent $i$'s information at $t$. Agent $i$'s \textit{behavioral strategy} $g_t^i$, $t\in\mathcal{T}$, is defined as a sequence of mappings $g_t^i:\mathcal{H}_t^i\rightarrow \Delta (\mathcal{A}_t^i)$, $t\in\mathcal{T}$, that determine agent $i$'s action $A_t^i$ for every realization $h_t^i\in\mathcal{H}_t^i$ of the history  at $t\in\mathcal{T}$. 

Agent $i$'s instantaneous utility at $t$ depends on the system state $X_t$ and the collective action profile $A_t$, and is given by $u_t^i\hspace*{-1pt}(\hspace*{-1pt}X_t,\hspace*{-1pt}A_t\hspace*{-1pt})$. Agent $i$ chooses his strategy $g_{1\hspace*{-1pt}:T}^i$ so as to maximize his total (expected) utility over horizon $\mathcal{T}$, given by,\vspace*{-2pt}
\begin{align}
U^i(X_{1:T},A_{1:T})=\sum_{t\in\mathcal{T}}u_t^i(X_t,A_t). \label{eq:totalutility}\vspace*{-2pt}
\end{align}

To avoid measure-theoretic technical difficulties and for clarity and convenience of exposition, we assume that all the random variables take values in finite sets.
\begin{assumption}\label{assump:finite}(Finite game)
	The sets $\mathcal{X}_t$, $\mathcal{Z}_t$, $\mathcal{Y}_t^i$, $\mathcal{A}_t^i$, $\mathcal{N}$, and $\mathcal{T}$ are finite.
\end{assumption}

Moreover, we assume that given any sequence of actions $a_{1:t-1}$ up to time $t-1$, every possible realization $x_{t}\hspace*{-2pt}\in\hspace*{-2pt}\mathcal{X}_{t}$ of the system state at $t$ has a strictly positive probability of realization. 

\begin{assumption}(Strictly positive transition matrix)\label{assump:positivetrans}
	For all $t\hspace*{-2pt}\in\hspace*{-2pt}\mathcal{T}$,  $x_{t}\hspace*{-2pt}\in\hspace*{-2pt}\mathcal{X}_{t}$ and $a_{1:t-1}\hspace*{-2pt}\in\hspace*{-2pt}\mathcal{A}_{1:t-1}$, we have $\mathbb{P}\{x_{t}|a_{1:t-1}\}\hspace*{-2pt}>\hspace*{-2pt}0$.
\end{assumption}

Furthermore, we assume that for any sequence of actions $\{a_{1:T}\}$, all possible realizations of private observations $\{y_{1:T}^{1:N}\}$ have positive probability. That is, no agent can infer perfectly another agent's action based only on his private observations.

\begin{assumption}\label{assump:positiveprob} (Imperfect private monitoring)
	For all $t\hspace*{-2pt}\in\hspace*{-2pt}\mathcal{T}$, $y_{1:t}\hspace*{-2pt}\in\hspace*{-2pt}\mathcal{Y}_{1:t}$, and $a_{1:t-1}\hspace*{-2pt}\in\hspace*{-2pt}\mathcal{A}_{1:t-1}$, we have $\mathbb{P}\{y_{1:t}|a_{1:t-1}\}\hspace*{-2pt}>\hspace*{-2pt}0$.
\end{assumption} 

\begin{remark}
	We can relax Assumptions  \ref{assump:positivetrans} and \ref{assump:positiveprob} under certain conditions and obtain results similar to those appearing in this paper; for instance, when agents actions are observable we can relax Assumptions \ref{assump:positivetrans} and \ref{assump:positiveprob}. Broadly, the crucial assumption that underlies our results is that every deviation that can be detected by agent $i$ at
	any time $t$ must be also detectable by all agents at the same time $t$ based only on the common information $C_t$. Due to space limitation we do not include the discussion of Assumptions \ref{assump:positivetrans} and \ref{assump:positiveprob} and the extension of our results when we relax them; we refer an interested reader to \cite{arXiv}. 
\end{remark}


		\vspace*{-7pt}
		
		\subsection{Special Cases}\label{sec:model:special}
We discuss several instances of dynamic games with asymmetric information that are special cases of the general model described above.  

\vspace{3pt}

\textit{1) Nested information structure:} Consider a two-player game with one informed player and one uninformed player and a general Markovian dynamics. At every time $t\hspace*{-2pt}\in\hspace*{-2pt} \mathcal{T}$, 
the informed player makes a private perfect observation of the state $X_t$, \textit{i.e.} $Y_t^1\hspace*{-2pt}=\hspace*{-2pt}X_t$. The uninformed player does not have any observation of the state $X_t$. Both the informed and uninformed players observe each others' actions, \textit{i.e.} $Z_t\hspace*{-2pt}=\hspace*{-2pt}\{A_{t-1}\}$. Therefore, we have $P_t^1=\{X_{1:t}\}$, $P_t^2=\emptyset$, and $C_t\hspace*{-2pt}=\hspace*{-2pt}\{A_{1:t-1}^1\hspace*{-1pt},\hspace*{-1pt}A_{1:t-1}^2\}$ for all $t\hspace*{-2pt}\in\hspace*{-2pt}\mathcal{T}$. The above nested information structure corresponds to dynamic games considered in \cite{renault2006value,renault2012value,li2017efficient}, where in \cite{renault2012value,li2017efficient} the state $X_t$ is static.

\vspace{3pt}

\textit{2) Independent dynamics with observable actions:}
Consider an $N$-player game where the state $X_t\hspace*{-2pt}:=\hspace*{-2pt}(X_t^0\hspace*{-1pt},\hspace*{-1pt}X_t^1\hspace*{-1pt},\hspace*{-1pt}X_t^2\hspace*{-1pt},\hspace*{-1pt}...,\hspace*{-1pt}X_t^N)$ has $N$ components. The agents' actions $A_t$ are observable by all agents, \textit{i.e.} $A_{t-1}\hspace*{-2pt}\subset\hspace*{-2pt} Z_t$ for all $t\hspace*{-2pt}\in\hspace*{-2pt}\mathcal{T}$. At every time $t\in\mathcal{T}$, agent $i$ makes a perfect observation of its local state $X_t^i$ as well as a global state $X_t^0$. Moreover, at time $t$ all agents make a common imperfect observation of state $X_t^i$ given by $Z_t^i\hspace*{-2pt}=\hspace*{-2pt}O_t^c(X_t^i\hspace*{-1pt},\hspace*{-1pt}A_{t-1}\hspace*{-1pt},\hspace*{-1pt}W_t^{c,i})$, $i\hspace*{-2pt}\in\hspace*{-2pt}\mathcal{N}$. Conditioned on the agents' collective action $A_t$, each $X_t^i$ evolves independently over time as $X_{t+1}^i\hspace*{-2pt}=\hspace*{-2pt}f_t(X_t^i\hspace*{-1pt},\hspace*{-1pt}A_{t-1}\hspace*{-1pt},\hspace*{-1pt}W_t^{x,i})$ for all $i\hspace*{-2pt}\in\hspace*{-2pt}\mathcal{N}$ and $t\hspace*{-2pt}\in\hspace*{-2pt}\mathcal{T}$. We assume that $X_1$, $W_t^c$, $t\hspace*{-2pt}\in\hspace*{-2pt}\mathcal{T}$, and $W_t^{x,i}$, $i\hspace*{-2pt}\in\hspace*{-2pt}\mathcal{N}$, $t\hspace*{-2pt}\in\hspace*{-2pt}\mathcal{T}$ are mutually independent. Therefore, we have $P_t^i\hspace*{-2pt}=\hspace*{-2pt}\{X_{1:t}^i\}$ and $C_t\hspace*{-2pt}=\hspace*{-2pt}\{X_{1:t}^0\hspace*{-1pt},\hspace*{-1pt}Z_{1:t}^{1:N}\hspace*{-1pt},\hspace*{-1pt}A_{1:t-1}\}$. The above environment includes the dynamic game considered in \cite{ouyang2016TAC,ouyang2015CDC} as special cases.   

%


\vspace{5pt}

\textit{3) Perfectly controlled dynamics with hidden actions:} Consider a $N$-player game where the state $X_t\hspace*{-2pt}:=\hspace*{-2pt}(X_t^1\hspace*{-1pt},\hspace*{-1pt}X_t^2\hspace*{-1pt},\hspace*{-1pt}...,\hspace*{-1pt}X_t^N)$ has $N$ components. Agent $i$, $i\hspace*{-2pt}\in\hspace*{-2pt}\mathcal{N}$, perfectly controls $X_t^i$, \textit{i.e.} $X_{t+1}^i=A_t^i$. Agent $i$'s actions $A_t^i$, $t\hspace*{-2pt}\in\hspace*{-2pt}\mathcal{T}$, is not observable by  all other agents $-i$. Every agent $i$, $i\hspace*{-2pt}\in\hspace*{-2pt}\mathcal{N}$, makes a noisy private observation $Y_i^t(X_t,W_t^i)$ of the system state at $t\hspace*{-2pt}\in\hspace*{-2pt}\mathcal{T}$. Therefore, we have $P_t^i\hspace*{-2pt}:=\hspace*{-2pt}\{A_{1:t},Y_{1:t}^i\}$, $C_t\hspace*{-2pt}=\hspace*{-2pt}\emptyset$. 

\vspace{-5pt} 

%% file: Equilibrium.tex
\section{Appraisals and Assessments}
\label{sec:equilibrium}
In this section we 
\red{provide an overview} of the notions of appraisals, assessments, and an equilibrium solution concept for dynamic games with asymmetric information. We argue that an equilibrium solution concept must consist of a pair of a strategy profile and a belief system (to be defined below), and discuss the importance of off-equilibrium path beliefs in dynamic games.\footnote{We refer the interested reader to the papers by Battigalli \cite{battigalli1996strategic}, Myerson and Remy \cite{myerson2015open}, and Watson \cite{watson2016perfect} for more discussion.}

In a dynamic game with asymmetric information 	
 agents have private information about the evolution of the game, and they do not observe the complete history of the game given by $\{\hspace*{-1pt}H_t,\hspace*{-1pt}X_{t}\hspace*{-1pt}\}$. Therefore, at every time $t\hspace*{-2pt}\in\hspace*{-2pt}\mathcal{T}\hspace*{-1pt}$, each agent, say agent $i\hspace*{-2pt}\in\hspace*{-2pt}\mathcal{N}\hspace*{-1pt}$, needs to form (i) an appraisal about the current state of the system $X_{t}$ and the other agents' information $H_t^{-i}$ (appraisal about the history), and (ii) an appraisal about how other agents will play in the future, so as to evaluate the performance of his strategy choices (appraisal about the future). 
Given the other agents' strategies $g^{-i}$, agent $i$ can utilize his own information $H_t^i$ at $t\hspace*{-2pt}\in\hspace*{-2pt}\mathcal{T}\hspace*{-1pt}$,  along with (i) other agents' past strategies $g_{1:t-1}^{-i}$  and (ii) other agents' future strategies $g_{t:T}^{-i}$ to form these appraisals about the history and future of the game, respectively.
%


In contrast to dynamic teams where agents have a common objective and coordinate their strategies, in dynamic games each agent has his own objective and chooses his strategy $g^i$ so as to maximize his objective. Thus, unlike dynamic teams, in dynamic games strategy $g^i$ is agent $i$'s private information and not known to other agents. Therefore, in dynamic games, each agent needs to form a prediction about the other agents' strategies. We denote this prediction by $g^{*1:N}_{1:T}$ to distinguish it from the strategy profile $g_{1:T}^{1:N}$ that is actually being played by the agents. Following Nash's idea, we assume that agents share a common prediction $g^*$ about the actual strategy $g$. We would like to emphasize that the prediction $g^*$ does not necessarily coincide with the actual strategy $g$. As we point out later, one requirement of an equilibrium is that for every agent $i\hspace*{-2pt}\in\hspace*{-2pt}\mathcal{N}$, the prediction $g^{*i}$ must be an optimal strategy for him given the other agents prediction strategies $g^{*-i}$. 

Since an agent's actual strategy, say agent $i$'s strategy $g^i$, is his own private information, it is possible that $g^i$ is different from the prediction $g^{*i}$. Below we discuss the implication of an agent's deviation from the prediction strategy profile $g^*$. For that matter, we first consider an agent who may want to deviate from  $g^*$, and then we consider an agent who faces such a deviation and his response.

In dynamic games, when agent $i\hspace*{-2pt}\in\hspace*{-2pt}\mathcal{N}$ chooses his strategy $g^i$, he needs to know how other agents will play for any choice of $g^i$ which can be different from the prediction $g^{*i}$. Therefore, the prediction $g^{*}$ has to be defined at all possible information realizations (\textit{i.e.} information sets) of every agent, those that have positive probability under $g^*$ as well as those that have zero probability under $g^*$.\footnote{This is not an issue in dynamic teams since agents coordinate  in advance their choice of strategy profile $g$, and no agent has an incentive to (privately) deviate from it. Hence, the agents' strategy profile $g$ is needed to be defined \red{only} on information sets of positive probability under $g$.}
 Using the prediction $g^*$, any agent, say agent $i$, can form an appraisal about the future of the game for any strategy choice $g^i$, and evaluate the performance of $g^i$.

By the same rationale, when agent $i$ chooses  $g^i$ he needs to determine his strategy for all of his information sets, even those that have zero probability under  $g^{*-i}$. This is because it is possible that some agent $j\hspace*{-2pt}\in\hspace*{-2pt}\mathcal{N}$ may deviate from $g^{*j}$ and play a strategy $g^j$ that is different from the prediction $g^{*j}$. Agent $i$ must foresee these possible deviations by other agents and determine his response to these deviations.

To determine his optimal strategy $g^i$ at any information set, agent $i$ needs to first form an appraisal about the history of the game at $t$ as well as an appraisal about the future of the game using the strategy prediction $g^{*-i}$. For an information set $h_t^i$ that is \textit{compatible} with the prediction $g^{*-i}$ given his strategy $g^i$ at $t\hspace*{-2pt}\in\hspace*{-2pt}\mathcal{T}$ (\textit{i.e.} $h_t^i$ has positive probability of being realized under $g^*$), agent $i$ can use Bayes' rule to derive the appraisal about the history of the game at $t$. However, 
for an information set $h_t^i$ that has zero probability under the prediction $g^{*-i}$ given $g^i$, agent $i$ cannot anymore rely on the prediction $g^*$ and use Bayes' rule to form his appraisal about the history of the game at $t$. The realization of history $h_t^i$ tells agent $i$ that his original prediction $g^{*-i}_{1:t-1}$ is not (completely) correct, thus, he needs to revise his original prediction $g^{*-i}_{1:t-1}$ and to form a revised appraisal about the history of the game at $t$. Therefore, agent $i$ must determine how to form/revise his appraisal about the history of the game for every realization $h_t^i\hspace*{-2pt}\in\hspace*{-2pt}\mathcal{H}_t^i$, $t\hspace*{-2pt}\in\hspace*{-2pt}\mathcal{T}$, that has zero probability under $g^{*-i}$. We note that upon reaching an information set of measure zero, agent $i$ only revises his prediction $g^{*-i}_{1:t-1}$ about other agents' past strategies, but does not change his prediction $g^{*-i}_{t:T}$ about their future strategies. This is because at equilibrium, the prediction $g^{*-i}_{t:T}$ specifies a set of strategies for other agents that are optimal in the continuation game that takes place after the realization of the information set $h_t^i$ of zero probability under $g^*_{1:t-1}$.%
\footnote{In dynamic teams, agents only need to determine their optimal strategy $g$ for information sets that have positive probability under $g$. As a result, a collective choice of strategy $g$ is optimal at every information set with positive probability if and only if it maximizes the (expected) utility of the team from $t=1$ up to $T$. However, in dynamic games agents need to determine their strategies for all information sets irrespective of whether they have zero or positive probability under $g^*$. Therefore, if a choice of strategy $g^i$ maximizes agent $i$'s  (expected) utility from $t=1$ to $T=1$, it does not imply that it is also optimal at all information sets that have zero probability under $\{g^{*-i},g^i\}$. Consequently, a choice of agent $i$'s strategy must be optimal for all continuation games that follow after a realization of an information set $h_t^i$ irrespective of whether it has zero or positive probability.}

We describe \red{below how one} can formalize the above issues we need to consider in the study of dynamic games with asymmetric information. Following the game theory literature \cite{fudenberg1991game}, agents' appraisals about the history and future of the game can be captured by an \textit{assessment} that all agents commonly hold about the game. We define an \textit{assessment} as a pair of mappings $(g^*\hspace*{-2pt},\hspace*{-1pt}\mu)$, where $g^*\hspace*{-2pt}\hspace*{-2pt}:=\hspace*{-2pt}\hspace*{-2pt}\{g_t^{*i}\hspace*{-1pt},\hspace*{-1pt}i\hspace*{-2pt}\in\hspace*{-2pt}\mathcal{N}\hspace*{-1pt},t\hspace*{-2pt}\in\hspace*{-2pt}\mathcal{T}\}$, $g^{*i}_t\hspace*{-2pt}:\hspace*{-2pt}\mathcal{H}_t^i\hspace*{-1pt}\rightarrow\hspace*{-2pt}\Delta(\mathcal{A}_t^i)$ denotes a \textit{prediction} about agent $i$'s strategy at $t$, and $\mu\hspace*{-2pt}:=\hspace*{-2pt}\{\hspace*{-1pt}\mu_t^{i},\hspace*{-1pt}i\hspace*{-2pt}\in\hspace*{-2pt}\mathcal{N}\hspace*{-1pt},\hspace*{-1pt}t\hspace*{-2pt}\in\hspace*{-2pt}\mathcal{T}\}$, where $\mu^{i}_t\hspace*{-2pt}:\hspace*{-2pt}\mathcal{H}_t^i\hspace*{-1pt}\rightarrow\hspace*{-2pt}\Delta(\mathcal{X}_t\hspace*{-2pt}\times\hspace*{-2pt}\mathcal{H}_t^{-i})$ denotes agent $i$'s \textit{belief} about the system state $X_t$ and agents $-i$'s information $H_t^{-i}$ given his information $H_t^i$. The collection of mappings $\mu\hspace*{-2pt}:=\hspace*{-2pt}\{\mu^{i}_t, \hspace*{-1pt}i\hspace*{-2pt}\in\hspace*{-2pt}\mathcal{N}\hspace*{-1pt},\hspace*{-1pt} t\hspace*{-2pt}\in\hspace*{-2pt}\mathcal{T}\}$ is called a \textit{belief system}. For every $i\hspace*{-2pt}\in\hspace*{-2pt}\mathcal{N}$, $t\hspace*{-2pt}\in\hspace*{-2pt}\mathcal{T}$, and $h_t^i\hspace*{-2pt}\in\hspace*{-2pt}\mathcal{H}_t^i$, $\mu^{i}_t(h_t^i)$ denotes agent $i$'s belief about the history $\{\hspace*{-1pt}X_t,\hspace*{-1pt}H_t^{-i}\}$ of the game, and $g_{t:T}^{*-i}$ denotes agent $i$'s prediction about all other agents' continuation strategy from $t$ onward.  We note that $\mu^{i}_t(h_t^i)$  determines agent $i$'s appraisal about the history of the game when $h_t^i$ has either positive or zero probability under $g^*$. Therefore, using an assessment  $(g^*\hspace*{-1pt},\hspace*{-1pt}\mu)$ each agent can fully construct appraisals  about the history and future of the game \red{at any $t\hspace*{-2pt}\in\hspace*{-2pt}\mathcal{T}$.}

Using the definition of an assessment, we can extend the idea of Nash equilibrium to dynamic games with asymmetric information. An equilibrium of the dynamic game is defined as a common assessment $(g^*\hspace*{-2pt},\mu)$ among the agents that satisfies the following conditions under the assumption that the agents are rational. (i) Agent $i\hspace*{-2pt}\in\hspace*{-2pt}\mathcal{N}$ chooses his strategy $g^i_{1:T}$ so as to maximize his total expected utility (\ref{eq:totalutility}) in all continuation games given the assessment $(g^*\hspace*{-1pt},\hspace*{-1pt}\mu)$ about the game. Therefore, the prediction $g^{*i}_{1:T}$ that other agents hold about agent $i$'s strategy must be a maximizer of agent $i$'s total expected utility under the assessment $(g^*\hspace*{-1pt},\hspace*{-1pt}\mu)$. (ii) For all $t\hspace*{-2pt}\in\hspace*{-2pt}\mathcal{T}$, agent $i$'s, $i\hspace*{-2pt}\in\hspace*{-2pt}\mathcal{N}$, belief $\mu^{i}_t(h_t^i)$ at information set $h_t^i\hspace*{-2pt}\in\hspace*{-2pt}\mathcal{H}_t^i$ that has positive probability \red{of realization} under $g^*$, must be equal to the probability distribution of $\{\hspace*{-1pt}X_t,\hspace*{-1pt}H_t^{-i}\}$ \red{conditioned on} the realization $h_t^i$ (determined via Bayes' rule) assuming that agents $-i$ play according to $g^{*-i}_{1:t}$. 
When $h_t^i$ has zero probability under the assessment $g^*$, the belief $\mu_t^i(h_t^i)$ cannot be determined via Bayes' rule and must be revised.   The revised belief must satisfy a certain set of \textit{``reasonable''} conditions so as to be compatible with agent $i$'s rationality. Various sets of conditions have been proposed in the literature (see \cite{fudenberg1991game,osborne1994course}) to capture the notion of ''reasonable'' beliefs  that are compatible with the agents' rationality. Different sets of conditions for off-equilibrium beliefs $\mu^{i}_t(h_t^i)$  result in the different equilibrium concepts that are proposed for dynamic games with asymmetric information.

In this paper, we consider Perfect Bayesian Equilibrium (PBE) as the equilibrium solution concept.
In the next section we provide the formal definition of PBE. 

%% file: PBE.tex
\section{Perfect Bayesian Equilibrium}\label{sec:PBE}

The formal definition of Perfect Bayesian Equilibrium (PBE) for dynamic games in extensive form can be found in \cite{osborne1994course}. In this paper we use a state space representation for dynamic games instead of an extensive game form representation, therefore, we need to adapt the definition of PBE to this representation. \red{A PBE is defined} as an assessment $(g^*\hspace*{-2pt},\hspace*{-1pt}\mu)$ that satisfies the \textit{sequential rationality} and \textit{consistency} conditions. The sequential rationality condition requires that for all $i\hspace*{-2pt}\in \hspace*{-2pt}\mathcal{N}$, the prediction $g^{*i}$ is optimal for agent $i$ given the assessment $(g^*\hspace*{-2pt},\hspace*{-1pt}\mu)$.
The consistency condition requires that for all $i\hspace*{-2pt}\in\hspace*{-2pt} \mathcal{N}$, $t\hspace*{-2pt}\in\hspace*{-2pt}\mathcal{T}$, and $h_t^i\hspace*{-2pt}\in\hspace*{-2pt}\mathcal{H}_t^i$, agent $i$'s belief $\mu(h_t^i)$ must be compatible with prediction $g^*$. We formally define these conditions below.

\vspace*{2pt}

Let $\mathbb{P}^{(g^{*-i}_{t:T},g^{*i}_{t:T})}_{\mu^{i}_t}\hspace*{-2pt}\{\cdot|h_t^i\}$ denote the probability measure induced by the stochastic process that starts at time $t$ with initial condition $\{\hspace*{-1pt}X_t,\hspace*{-1pt}P_t^{-i},\hspace*{-1pt}p^i_t,\hspace*{-1pt}c_t\hspace*{-1pt}\}$, 
where \red{random variables} $\{\hspace*{-1pt}X_t\hspace*{-1pt},\hspace*{-2pt}P_t^{-i}\}$ \red{are} distributed according to probability distribution $\mu^{i}_t(h_t^i)$, assuming that agents $i$ and $-i$ take actions according to strategies $g_{t:T}^{*i}$ and  $g_{t:T}^{*-i}$, respectively. In the sequel, to save some notation, we write $\mathbb{P}^{g^*}_\mu\hspace*{-2pt}\hspace*{-1pt}\{\cdot\}$ instead of  $\mathbb{P}^{(g^{*-i}_{t:T},g^{*i}_{t:T})}_{\mu^{i}_t}\hspace*{-2pt}\{\cdot\}$ whenever there is no confusion. 
	
\begin{definition}[Sequential rationality] We say that an assessment $(g^*\hspace*{-2pt},\hspace*{-1pt}\mu)$ is sequentially rational if $\forall i\hspace*{-2pt}\in\hspace*{-2pt}\mathcal{N}$, $t\hspace*{-2pt}\in\hspace*{-2pt}\mathcal{T}$, and $h_t^i\hspace*{-2pt}\in\hspace*{-2pt}\mathcal{H}_t^i$, the strategy prediction $g^{*i}_{t:T}$ is a solution to 
	\begin{align}
	\sup_{g^{i}_{t:T}} \mathbb{E}^{(g^{*-i}_{t:T},g^{i}_{t:T})}_{\mu_{t}^{i}}\left\{\sum_{\tau=t}^{T}u^i_t(X_t,A_t)|h_t^i\right\} \label{SeqR}
	\end{align} 	
\end{definition}


The sequential rationality condition (\ref{SeqR}) requires that, given the assessment $(g^*\hspace*{-2pt},\hspace*{-1pt}\mu)$,  the prediction strategy $g^{i*}_t$ for agent $i$ is an optimal strategy for him for all continuation games after history realization $h_t^i\in H^i$, irrespective of whether $h_t^i$ has positive or zero probability  under $(g^*\hspace*{-2pt},\hspace*{-1pt}\mu)$. That is, the common prediction $g^{*i}$ about agent $i$'s strategy must be an optimal strategy choice for him since it is common knowledge that he is a rational agent. We note that the sequential rationality condition defined above is more restrictive than the optimality condition for Bayesian Nash Equilibrium (BNE) which only requires  (\ref{SeqR}) to hold at $t\hspace*{-2pt}=\hspace*{-2pt}1$. By the sequential rationality condition, we require the optimality of prediction $g^*$ even along off-equilibrium paths, and thus, we rule out the possibility of \textit{non-credible threats} (see \cite{fudenberg1991game} for more discussion). 

The sequential rationality condition results in a set of constraints that the strategy prediction $g^*$ must satisfy given a belief system $\mu$. As we argued in Section \ref{sec:equilibrium}, the belief system $\mu$ must be also compatible with the strategy prediction $g^*$. The following \textit{consistency} condition captures such compatibility between the belief system $\mu$ and the prediction $g^*$.

\begin{definition}[Consistency] We say that an assessment $(g^*\hspace*{-2pt},\hspace*{-1pt}\mu)$ is consistent if
	\begin{enumerate}[i)]
		\item for all $i\hspace*{-2pt}\in\hspace*{-2pt}\mathcal{N}$\hspace*{-1pt}, \hspace*{-1pt}$t\hspace*{-2pt}\in\hspace*{-2pt}\mathcal{T}\backslash\{1\}$, $h_{t-1}^i\hspace*{-2pt}\in\hspace*{-2pt}\mathcal{H}_{t-1}^i$, and $h_t^i\hspace*{-2pt}\in\hspace*{-2pt}\mathcal{H}_t^i$ such that $\mathbb{P}_{\mu}^{g^*}\hspace*{-2pt}\{\hspace*{-1pt}h_t^i|h_{t-1}^i\hspace*{-1pt}\}\hspace*{-2pt}>\hspace*{-2pt}0$, the belief $\mu^{i}_t(h_t^i)$ must satisfy Bayes' rule, \textit{i.e.} \vspace*{-2pt}
		\begin{align}
		\mu^{i}_t(h_t^i)(x_{1:t},p_t^{-i})=\frac{\mathbb{P}^{g^*}\{h_t^i,x_{t},p_t^{-i}|h_{t-1}^i\}}{\mathbb{P}^{g^*}\{h_t^i|h_{t-1}^i\}}. \label{eq:consistency1} \vspace*{-2pt}
		\end{align}
		\item for all $i\hspace*{-2pt}\in\hspace*{-2pt}\mathcal{N}$, $t\hspace*{-2pt}\in\hspace*{-2pt}\mathcal{T}\backslash\{1\}$, $h_{t-1}^i\hspace*{-2pt}\in\hspace*{-2pt}\mathcal{H}_{t-1}^i$, and $h_t^i\hspace*{-2pt}\in\hspace*{-2pt}\mathcal{H}_t^i$ such that $\mathbb{P}_{\mu}^{g^*}\hspace*{-2pt}\{\hspace*{-1pt}h_t^i|h_{t-1}^i\hspace*{-1pt}\}\hspace*{-2pt}=\hspace*{-2pt}0$, we have \vspace*{-2pt}
		\begin{align*}
		\mu^{i}_t(h_t^i)(x_{1:t},p_t^{-i})> 0 \vspace*{-2pt}
		\end{align*}
		only if there exists an open loop strategy ${(A_{1:t-1}^{-i}\hspace*{-2pt}=\hspace*{-2pt}\hat{a}_{1:t-1}^{-i}\hspace*{-1pt},\hspace*{-1pt}A_{1:t-1}^i\hspace*{-2pt}=\hspace*{-2pt}a_{1:t-1}^i)}$ such that 
		\begin{align}
		\mathbb{P}^{(A_{1:t-1}^{-i}=\hat{a}_{1:t-1}^{-i},A_{1:t-1}^i=a_{1:t-1}^i)}\{x_t,p_t^{-i}\}> 0.
		\label{eq:consistency2}
		\end{align}
	\end{enumerate}
	\label{def:consistency}
\end{definition}

		\vspace*{-3pt}

The above consistency condition places a restriction on the belief system $\mu$ so that it is compatible with the strategy prediction $g^*$. For information sets along equilibrium paths, \textit{i.e.} $\mathbb{P}_{\mu^{i}_1}^{g^*}\hspace*{-1pt}\{\hspace*{-1pt}h_{t}^i\}\hspace*{-2pt}>\hspace*{-2pt}0$, \red{belief} $\mu_t^{i}(h_t^i)$ must be updated according to (\ref{eq:consistency1}) via Bayes' rule since  agent $i$'s observations are consistent with the prediction $g^*$. For information sets along off-equilibrium paths, \textit{i.e.} $\mathbb{P}_{\mu^{i}_1}^{g^*}\hspace*{-1pt}\{\hspace*{-1pt}h_{t}^i\}\hspace*{-2pt}=\hspace*{-2pt}0$,  agent $i$ needs to revise his belief about the strategy of agents $-i$ as the realization of $h_t^i$ indicates that some agent has deviated from prediction $g^{*-i}_{1:t}$. As pointed out before, the revised belief $\mu^{i}\hspace*{-1pt}(h_t^i)$ must be ``reasonable''. Definition \ref{def:consistency} provides a set of such ``reasonable'' conditions captured by (\ref{eq:consistency1}) and (\ref{eq:consistency2}) that we discuss further below. 

First, consider an information set $h_t^i$ along an off-equilibrium path such that $\mathbb{P}_{\mu^{i}_{t-1}}^{g^*}\hspace*{-3pt}\{\hspace*{-1pt}h_{t}^i|h_{t-1}^i\hspace*{-1pt}\}\hspace*{-2pt}>\hspace*{-2pt}0$. That is, conditioned on reaching information set $h_{t-1}^i$ at $t\hspace*{-2pt}-\hspace*{-2pt}1$, $h_t^i$ has a positive probability under the prediction strategy $g^*$. Since $\mathbb{P}_{\mu^{i}_1}^{g^*}\hspace*{-1pt}\{h_{t}^i\}\hspace*{-2pt}=\hspace*{-2pt}\mathbb{P}_{\mu^{i}_{t-1}}^{g^*}\hspace*{-2pt}\{h_{t}^i|h_{t-1}^i\}\mathbb{P}_{\mu^{i}_1}^{g^*}\hspace*{-1pt}\{h_{t-1}^i\}$ and $\mathbb{P}_{\mu^{i}_1}^{g^*}\hspace*{-1pt}\{h_{t}^i\}\hspace*{-2pt}=\hspace*{-2pt}0$,  we have $\mathbb{P}_{\mu^{i}_1}^{g^*}\hspace*{-1pt}\{h_{t-1}^i\}\hspace*{-2pt}=\hspace*{-2pt}0$. Therefore, $h_{t-1}^i$ is also an information set along an off-equilibrium path, and $\mu^{i}(h_{t-1}^i)$ is a revised belief that agent $i$ holds at $t\hspace*{-2pt}-\hspace*{-2pt}1$. Note that if the assessment $(g^*\hspace*{-2pt},\hspace*{-1pt}\mu)$ satisfies the sequential rationality condition, $g^*$ is a best response for all agents in all continuation games that follow the realization of every information set of  positive or zero probability. Moreover, since $\mathbb{P}_{\mu^{i}_{t-1}}^{g^*_{t-1}}\hspace*{-1pt}\{h_{t}^i|h_{t-1}^i\}\hspace*{-2pt}>\hspace*{-2pt}0$, the realization of $h_t^i$ conditioned on reaching $h_{t-1}^i$ is consistent with the strategy prediction $g^*_{t-1}$. Therefore, agent $i$ does not have any reason to further revise his belief about agents $-i$'s strategy beyond the revision that results in $\mu^{i}_{t-1}\hspace*{-1pt}(h_{t-1}^i)$. Thus, agent $i$ determines his belief $\mu^{i}_t\hspace*{-1pt}(h_t^i)$ by utilizing his belief $\mu^{i}_{t-1}\hspace*{-1pt}(h_{t-1}^i)$ at $t\hspace*{-2pt}-\hspace*{-2pt}1$ and updating it via Bayes' rule assuming that agents $-i$' play according to $g^{*-i}_{t-1}$ (see part (i), eq. (\ref{eq:consistency1})).

Next, consider an information set $h_t^i$ along an off-equilibrium path such that $\mathbb{P}_{\mu^{i}_t}^{g^*}\hspace*{-1pt}\{\hspace*{-1pt}h_{t}^i|h_{t-1}^i\}\hspace*{-2pt}=\hspace*{-2pt}0$. That is, conditioned on reaching information set $h_{t-1}^i$ at $t\hspace*{-2pt}-\hspace*{-2pt}1$, $h_t^i$ has a zero probability of realization under the prediction $g^*$.  In this case, the realization of $h_t^i$ indicates that agents $-i$ have deviated from prediction $g_{1:t-1}^{*-i}$, and this deviation has not been detected by agent $i$ before. Therefore, agent $i$ needs to form a new belief \red{about} agents $-i$'s private information $P_t^{-i}$ and the state $X_t$ by revising $\mu_t^i(h_t^{i})$. Part (ii) of the consistency condition concerns such belief revisions and requires that the support of agent $i$'s revised belief $\mu_t^i(h_t^{i})$ includes only the states and private information that are feasible under the system and information dynamics (\ref{eq:systemdynamic1}) and (\ref{eq:systemdynamic2}), that is, they are reachable under some open-loop control strategy ${(A_{1:t-1}^{-i}\hspace*{-2pt}=\hspace*{-2pt}\hat{a}_{1:t-1}^{-i}\hspace*{-1pt},\hspace*{-1pt}A_{1:t-1}^i\hspace*{-2pt}=\hspace*{-2pt}a_{1:t-1}^i)}$. We note that since we are using a state representation of the dynamic game, we need to impose such a requirement, whereas in the equivalent extensive form representation of the game such a requirement is satisfied by the construction of the game-tree.        

\begin{remark}\label{remark:off}
	Under Assumptions \ref{assump:positivetrans} and \ref{assump:positiveprob}, we have $\mathbb{P}_{\mu_1}^{(\hspace*{-1pt}A_{1\hspace*{-1pt}:t-1}=\hat{a}_{1\hspace*{-1pt}:t-1}\hspace*{-1pt})}\{\hspace*{-1pt}x_{1:t},\hspace*{-1pt}p_t^{-i}\}\hspace*{-2pt}>\hspace*{-2pt}0$ for all $(\hspace*{-1pt}A_{1\hspace*{-1pt}:t-1}\hspace*{-2pt}=\hspace*{-2pt}\hat{a}_{1\hspace*{-1pt}:t-1}\hspace*{-1pt})$. Therefore part (ii) of the consistency conditions is trivially satisfied. In the rest of the paper, we ignore part (ii) and only consider part (i) of the definition of consistency. In \cite{arXiv}, we discuss the case when we relax Assumptions \ref{assump:positivetrans} and \ref{assump:positiveprob}.
\end{remark}

We can now provide the formal definition of PBE for the dynamic game of Section \ref{sec:model}.
\begin{definition}
	An assessment $(g^*\hspace*{-2pt},\hspace*{-1pt}\mu)$ is called a PBE if it satisfies the sequential rationality and consistency conditions.
\end{definition} 

The definition of Perfect Bayesian equilibrium provides a general formalization of outcomes that are \textit{rationalizable} (\textit{i.e.} consistent with agents' rationality) under some strategy profile and belief system. However, \red{as we argue further in Section \ref{sec:discussion}}, there are computational and philosophical reasons that motivate us to define a sub class of PBEs that provide a \textit{simpler} and more \textit{tractable} approach to characterizing the outcomes of dynamic games with asymmetric information.

There are two major challenges in computing a PBE $(g^*\hspace*{-2pt},\hspace*{-1pt}\mu)$. First,  there is an inter-temporal coupling between the agents' strategy prediction $g^*$ and belief system $\mu$. According to the consistency requirement, the belief system $\mu$ has to satisfy a set of conditions given a strategy prediction $g^*$. On the other hand, by sequential rationality, a strategy prediction $g^*$ must satisfy a set of optimality condition given belief system $\mu$. Therefore, there is a circular dependency between a prediction strategy $g^*$ and a belief system $\mu$ over time. For instance, by sequential rationality, agent $i$'s strategy $g^{i*}_t$ at time $t$ depends on the agents' future strategies $g^*_{t:T}$ and on the agents' past strategies $g^*_{1:t-1}$ indirectly through the consistency condition for $\mu^{i}_t$.  As a result, one needs to determine the strategy prediction $g^*$ and belief system $\mu$ simultaneously for the whole time horizon so as to satisfy the sequential rationality and consistency conditions, and thus, cannot sequentially decompose the computation of PBE over time. Second, the agents' information $h_t^i$, $i\hspace*{-2pt}\in\hspace*{-2pt}\mathcal{N}$, has a growing domain over time. Hence, the agents' strategies have growing domains over time, and this feature further complicates the computation of PBEs of dynamic games with asymmetric information.    

The definition of PBE requires an agent to keep track of all observations he acquires over time and to form beliefs about the private information of all other agents. As we show next, agents do not need to keep track of all of their past observations to reach an equilibrium. They can take into account fewer variables for decision making and ignore part of their information that is not \textit{relevant} to the continuation game. As we argue in Section \ref{subsec:PBE}, the class of simpler strategies proposed in this paper characterize a more plausible prediction about the outcome of the interaction among agents when the underlying system is highly dynamic and there exists considerable information asymmetry among them. 

%% file: CIB-PBE.tex
\vspace*{-5pt}

\section{The Sufficient Information Approach}\label{sec:CIB}
\red{We} characterize a class of PBEs that utilize strategy choices that are simpler than general behavioral strategies as they require agents to keep track of only a compressed version of their information over time. We proceed as follows. 
In Section \ref{subsec:privatecompress} we provide sufficient conditions for the subset of private information an agent needs to keep track of over time for decision making purposes.
 In Section \ref{subsec:CIBassessment}, we introduce the sufficient information based belief as a compressed version of the agents' common information that is sufficient for decision-making purposes. Based on these compressions of the agents' private and common information, we introduce the notion of sufficient information based assessments and Sufficient Information Based-Perfect Bayesian Equilibrium (SIB-PBE) in Sections \ref{subsec:CIBassessment} and \ref{subsec:CIBPBE}, respectively.
  
		\vspace*{-7pt}

\subsection{Sufficient Private Information}\label{subsec:privatecompress}
The key ideas for compressing an agent's private information appear in Definitions \ref{def:sufficient} below; We refer an interested reader to the companion paper \cite{team} for discussion on the rationale behind Definition \ref{def:sufficient}. 

\begin{definition}[Sufficient private information]
	We say $S_t^i=\zeta_t^i(P_t^i,C_t;g_{1:t-1}^*)$, $i\in\mathcal{N}$, $t\in\mathcal{T}$, is \textit{sufficient private information} for the agents if, 
	\begin{enumerate}[(i)]
		\item it can be updated recursively as 
		\begin{gather}
		S_t^{i}=\phi_t^i(S_{t-1}^{i},H_t^i\backslash H_{t-1}^i;g_{1:t-1}^*)  \text{ if } t\in\mathcal{T}\backslash\{1\}, \label{eq:sufficientupdate}
		\end{gather}
		\item for any strategy profile $g^*$ and for all realizations $\{c_t,p_t,p_{t+1},z_{t+1},a_t\}\in\mathcal{C}_t\times\mathcal{P}_t\times\mathcal{P}_{t+1}\times\mathcal{Z}_{t+1}$ of positive probability,
		\begin{align}
		\hspace*{-26pt}\mathbb{P}^{g^*_{1:t}}\hspace*{-1pt}\left\{\hspace*{-2pt}s_{t+1}\hspace*{-1pt},\hspace*{-1pt}z_{t+1}\hspace*{-1pt}\Big|p_t\hspace*{-1pt},\hspace*{-1pt}c_t\hspace*{-1pt},\hspace*{-1pt}a_t\hspace*{-2pt}\right\}\hspace*{-3pt}=\hspace*{-2pt}\mathbb{P}^{g^*_{1:t}}\hspace*{-1pt}\left\{\hspace*{-2pt}s_{t+1}\hspace*{-1pt},\hspace*{-1pt}z_{t+1}\hspace*{-1pt}\Big|s_t\hspace*{-1pt},\hspace*{-1pt}c_t\hspace*{-1pt},\hspace*{-1pt}a_t\hspace*{-2pt}\right\}\hspace*{-1pt},\hspace*{-4pt}\label{eq:sufficientdynamic}
		\end{align}
		\hspace*{-4pt}where $s_{\tau}^{1:N}\hspace*{-3pt}=\hspace*{-2pt}\zeta_{\tau}^{1:N}\hspace*{-2pt}(p_{\tau}^{1:N}\hspace*{-1pt},\hspace*{-1pt}c_{\tau};\hspace*{-1pt}g_{1\hspace*{-1pt}:\tau-1}^*\hspace*{-1pt})$ for $\tau\in\mathcal{T}$;
		\item 
		for every strategy profile \red{$\tilde{g}^*$ of the form} $\tilde{g}^*\hspace*{-2pt}:=\hspace*{-2pt}\{\hspace*{-1pt}\tilde{g}^{*i}_t\hspace*{-1pt}:\hspace*{-1pt}\mathcal{S}_t^i\times \mathcal{C}_t\rightarrow \Delta(\mathcal{A}_t^i), i\hspace*{-2pt}\in\hspace*{-2pt}\mathcal{N}\hspace*{-1pt},\hspace*{-1pt} t\hspace*{-2pt}\in\hspace*{-2pt}\mathcal{T}\}$ and $a_t\hspace*{-2pt}\in\hspace*{-2pt}\mathcal{A}_t$, $t\hspace*{-2pt}\in\hspace*{-2pt}\mathcal{T}$;
		\begin{align} 		
		\hspace*{-26pt}\mathbb{E}^{\tilde{g^*\hspace*{-2pt}}_{1:t-1}^{}}\hspace*{-2.5pt}\left\{\hspace*{-2pt}u_t^i(\hspace*{-1pt}X_t\hspace*{-1pt},\hspace*{-1pt}A_t\hspace*{-1pt})\hspace*{-1pt}\Big|c_t\hspace*{-1pt},\hspace*{-1pt}p_t^i\hspace*{-1pt},\hspace*{-1pt}a_t\hspace*{-2pt}\right\}\hspace*{-3pt}=\hspace*{-2pt}\mathbb{E}^{\tilde{g^*\hspace*{-2pt}}_{1:t-1}^{}}\hspace*{-2.5pt}\left\{\hspace*{-2pt}u_t^i(\hspace*{-1pt}X_t\hspace*{-1pt},\hspace*{-1pt}A_t\hspace*{-1pt})\hspace*{-1pt}\Big|c_t\hspace*{-1pt},\hspace*{-1pt}s_t^{i}\hspace*{-1pt},\hspace*{-1pt}a_t\hspace*{-2pt}\right\}\hspace*{-3pt},\hspace*{-5pt}\label{eq:payoff-relevant2}
		\end{align} 
		for all realizations $\{\hspace*{-1pt}c_{t}\hspace*{-1pt},\hspace*{-1pt}p_{t}^i\}\hspace*{-3pt}\in\hspace*{-2pt}\mathcal{C}_{t}\hspace*{-1pt}\times\hspace*{-1pt}\mathcal{P}_{t}^i$ of positive probability where $s_{\tau}^{1:N}\hspace*{-3pt}=\hspace*{-2pt}\zeta_{\tau}^{1:N}\hspace*{-2pt}(p_{\tau}^{1:N}\hspace*{-1pt},\hspace*{-1pt}c_{\tau};\hspace*{-1pt}\tilde{g}_{1\hspace*{-1pt}:\tau-1}^*\hspace*{-1pt})$ for $\tau\in\mathcal{T}$;\vspace{5pt}
		
		\item given an arbitrary strategy profile \red{$\tilde{g}^*$ of the form} $\tilde{g^*}\hspace*{-1pt}:=\hspace*{-1pt}\{\tilde{g}^{i}_t:\mathcal{S}_t^i\hspace*{-1pt}\times \hspace*{-1pt}\mathcal{C}_t\rightarrow \Delta(\mathcal{A}_t^i), i\hspace*{-2pt}\in\hspace*{-2pt}\mathcal{N}, t\hspace*{-2pt}\in\hspace*{-2pt}\mathcal{T}\}$, $i\hspace*{-2pt}\in\hspace*{-2pt}\mathcal{N}$, and $t\hspace*{-2pt}\in\hspace*{-2pt}\mathcal{T}$,
		\begin{align}
		\hspace*{-25pt}\mathbb{P}^{\tilde{g}^{*}_{1:t-1}}\hspace*{-2pt}\left\{\hspace*{-2pt}s_t^{-i}\hspace*{-1pt}\Big|p_t^i\hspace*{-1pt},\hspace*{-1pt}c_t\hspace*{-2pt}\right\}\hspace*{-3pt}=\hspace*{-2pt}\mathbb{P}^{\tilde{g}^{*}_{1:t-1}}\hspace*{-2pt}\left\{\hspace*{-1pt}s_t^{-i}\hspace*{-1pt}\Big|s_t^i\hspace*{-1pt},\hspace*{-1pt}c_t\hspace*{-2pt}\right\}\hspace*{-1pt},\hspace*{-4pt}\label{eq:sufficientinfo}
		\end{align}
		for all realizations $\{c_{t}\hspace*{-1pt},\hspace*{-1pt}p_{t}^i\}\hspace*{-2pt}\in\hspace*{-2pt}\mathcal{C}_{t}\hspace*{-2pt}\times\hspace*{-2pt}\mathcal{P}_{t}^i$ of positive probability where $s_{\tau}^{1:N}\hspace*{-3pt}=\hspace*{-2pt}\zeta_{\tau}^{1:N}\hspace*{-2pt}(p_{\tau}^{1:N}\hspace*{-1pt},\hspace*{-1pt}c_{\tau};\hspace*{-1pt}\tilde{g}_{1\hspace*{-1pt}:\tau-1}^*\hspace*{-1pt})$ for $\tau\in\mathcal{T}$.		
	\end{enumerate}
	\label{def:sufficient}
\end{definition}

We note that the conditions of Definition \ref{def:sufficient} is written in terms of strategy prediction profile $g^*$ for dynamic games. This is because, as we discussed before, the agents' actual strategy profile $g$ is their private information. Therefore, each agent $i$, $i\in\mathcal{N}$, evaluates the sufficiency of a compression of his private information using the strategy prediction he holds about other agents. 

\vspace*{-5pt}

\subsection{Sufficient Common Information}\label{subsec:commoncompress}

Based on the characterization of sufficient private information, we present a statistic (compressed version) of the common information $C_t$ that agents need to keep track of over time for decision making purposes.

Consider the sufficient private information $S_t^{1:N}\hspace*{-1pt}$, $t\hspace*{-2pt}\in\hspace*{-2pt}\mathcal{T}$. Define $\mathcal{S}_t^i$ to be the set of all possible realizations of $S_t^i$, and $\mathcal{S}_t\hspace*{-2pt}:=\hspace*{-2pt}\prod_{i=1}^N\hspace*{-1pt}\mathcal{S}_t^i$.     
Let $\gamma_t\hspace*{-2pt}:\hspace*{-2pt}\mathcal{C}_t\hspace*{-2pt}\rightarrow \hspace*{-2pt}\Delta(\hspace*{-1pt}\mathcal{X}_t\hspace*{-2pt}\times\hspace*{-2pt}\mathcal{S}_t\hspace*{-1pt})$ denote a mapping that determines a conditional probability distribution over the system state $X_t$ and the agents' sufficient private information $S_t$ given the common information $C_t$ at time $t$. We call  the collection of mappings $\gamma:=\{\gamma_t\hspace*{-1pt},\hspace*{-1pt}t\hspace*{-2pt}\in\hspace*{-2pt}\mathcal{T}\}$ a \textit{Sufficient Information Based belief system} (SIB belief system). Note that $\gamma_t$ is only a function of the common information $C_t$, and thus, it is computable by all agents. Let $\Pi_t^\gamma\hspace*{-2pt}:=\hspace*{-2pt}\gamma_t\hspace*{-1pt}(C_t)$ denote the (random) sufficient information based belief that agents hold under belief system $\gamma$ at $t$. We can interpret $\Pi_t^\gamma$ as the common belief that each agent holds about the system state $X_t$ and all the agents' (including himself) sufficient private information $S_t$ at time $t$. We call the SIB belief $\Pi_t$ a sufficient common information for the agents.   
In the rest of the paper, we write $\Pi_t$ and drop the superscript $\gamma$ whenever such a simplification in notation is clear. Moreover, we use the terms sufficient common information and SIB belief interchangeably.

\vspace{-6pt}
\subsection{Special Cases:}
We consider the special classes described in Section \ref{sec:model} and identify the sufficient information $S_{1:T}^{1:N}$ and SIB belief for each of them.

\textit{1) Nested information structure:} The uninformed agent (agent $2$) has no private information, $P_t^2\hspace*{-2pt}=\hspace*{-2pt}\emptyset$. Thus, $S_t^2\hspace*{-2pt}=\hspace*{-2pt}\emptyset$. For  the informed agent (agent $1$) consider $P_t^{1,pr}\hspace*{-2pt}=\hspace*{-2pt}X_t$.  Consequently, we can set $S_t^1\hspace*{-2pt}=\hspace*{-2pt}X_t$. Note that $P_t^2\hspace*{-2pt}=\hspace*{-2pt}\emptyset$, thus, the uninformed agent's belief about $P_t^1$ is the same as SIB belief $\Pi_t\hspace*{-2pt}=\hspace*{-2pt}\mathbb{P}\{X_t|A_{1:t-1}^1,\hspace*{-1pt}A_{1:t-1}^2\}$. 

\textit{2) Independent dynamics with observable actions:}  Consider $S_t^{i}\hspace*{-2pt}=\hspace*{-2pt}X_t^i$. Note that $X_t^j$, $j\hspace*{-2pt}\in\hspace*{-2pt}\mathcal{N}$ have independent dynamics given the collective action $A_t$ that is commonly observable by all agents. Therefore, agent $i$'s belief about $X_j$, $j\neq i$, is the same as SIB belief, $\Pi_t\hspace*{-2pt}=\hspace*{-2pt}\mathbb{P}^g\{X_t^j|C_t\}\hspace*{-2pt}=\hspace*{-2pt}\textit{P}^g\{X_t^j|P_t^i,C_t\}$. 

%

\textit{3) Perfectly controlled dynamics with hidden actions:} Since agent $i$, $i\in\mathcal{N}$, perfectly controls $X_t^i$ over time $t\hspace*{-2pt}\in\hspace*{-2pt}\mathcal{T}$, we set $S_t^i\hspace*{-2pt}=\hspace*{-2pt}\{A_{t-1}^i,Y_t^i\}$ and $\Pi_t\hspace*{-2pt}=\hspace*{-2pt}\emptyset$. 

		\vspace*{-5pt}

\subsection{Sufficient Information based Assessment} \label{subsec:CIBassessment}\vspace*{-2pt}
As we discussed in Section \ref{sec:PBE}, to form a prediction about the game we need to determine an assessment about the game that is sequentially rational and consistent. 
We show below that using the sufficient private information $S_t^{1:N}$ and sufficient common information (the SIB belief) $\Pi_t$, we can form a sufficient information based assessment about the game. We prove that such a sufficient information based assessment is rich enough to capture a subset of PBE. 

 
Consider a class of strategies that utilize the information given by $(\Pi_t,\hspace*{-1pt}S_t^i)$ for agent $i\hspace*{-2pt}\in\hspace*{-2pt}\mathcal{N}$ at time $t$. We call the mapping $\sigma^i_t\hspace*{-2pt}:\hspace*{-2pt}\Delta(\hspace*{-1pt}X_t\hspace*{-2pt}\times\hspace*{-1pt}\mathcal{S}_t\hspace*{-1pt})\hspace*{-1pt}\times\hspace*{-1pt} \mathcal{S}_t^i\hspace*{-1pt}\rightarrow\hspace*{-2pt} \Delta(\hspace*{-1pt}\mathcal{A}_t^i)$ a \textit{Sufficient Information Based (SIB) strategy} for agent $i$ at time $t$. A SIB strategy $\sigma^i_t$ determines a probability distribution for agent $i$'s action $A_t^i$ at time $t$ given his information $(\Pi_t,\hspace*{-1pt}S_t^i)$. A SIB strategy is a behavioral strategy where agents only use the SIB belief $\Pi_t\hspace*{-2pt}=\hspace*{-2pt}\gamma_t\hspace*{-1pt}(\hspace*{-1pt}C_t)$ (instead of complete common information $C_t$), and the sufficient private information $S_t^i=\zeta_t^i(\hspace*{-1pt}P_t^i,\hspace*{-1pt}C_t\hspace*{-1pt};g^*_{1:t-1})$ (instead of complete private information $P_t^i$). A collection of SIB strategies $\{\hspace*{-1pt}\sigma_{1:T}^1\hspace*{-1pt},\hspace*{-1pt}...,\hspace*{-1pt}\sigma_{1:T}^N\hspace*{-1pt}\}$ is called a \textit{SIB strategy profile} $\sigma$.  The set of SIB strategies is a subset of behavioral strategies, defined in Section \ref{sec:model}, as we can define,
\begin{align*}
g^{(\sigma,\gamma),i}_t(h_t^i):=\sigma^i_t(\pi_t^\gamma,s_t^i). 
\end{align*} 

In Section \ref{sec:PBE}, we defined a consistency condition between strategy prediction $g^*$ and a belief system $\mu$. Below, we provide an analogous consistency condition between a SIB strategy prediction $\sigma^*$ and a SIB belief system $\gamma$.

\begin{definition}
	A pair $(\sigma^*\hspace*{-2pt},\hspace*{-1pt}\gamma)$ of a SIB strategy prediction profile $\sigma^*$  and belief system $\gamma$ satisfies the consistency condition if
	\begin{enumerate}[(i)]
		\item for all $t\hspace*{-2pt}\in\hspace*{-2pt}\mathcal{T}\backslash\{1\}$\footnote{For $t\hspace*{-2pt}=\hspace*{-2pt}1$, $\Pi_1$ is given by the conditional probability at $t\hspace*{-2pt}=\hspace*{-2pt}1$ as  $\Pi_1\hspace*{-1pt}(x_1\hspace*{-1pt},\hspace*{-2pt}s_1\hspace*{-1pt})\hspace*{-2pt}:=\hspace*{-1pt}\frac{\mathbb{P}\{s_1|x_1,z_1\}}{\sum_{\hat{x}_1\in\mathcal{X}_1}\mathbb{P}\{z_1|\hat{x}_1\}\eta(\hat{x}_1)}$.}, $z_{t}\hspace*{-2pt}\in\hspace*{-2pt}\mathcal{Z}_{t}$, $\pi_{t-1}\hspace*{-2pt}=\hspace*{-2pt}\gamma_{t-1}\hspace*{-1pt}(\hspace*{-1pt}c_{t-1}\hspace*{-1pt})$,  and $\pi_t\hspace*{-2pt}=\hspace*{-2pt}\gamma_t\hspace*{-1pt}(\hspace*{-1pt}\{\hspace*{-1pt}c_{t-1}\hspace*{-1pt},\hspace*{-1pt}z_t\hspace*{-1pt}\}\hspace*{-1pt})$ such that $\mathbb{P}^{\sigma^*_t}_{\pi_{t-1}}\hspace*{-1pt}\{\hspace*{-1pt}z_{t}\hspace*{-1pt}\}\hspace*{-2pt}>\hspace*{-2pt}0$, $\pi_t$ must satisfy Bayes' rule, \textit{i.e.},\vspace*{-2pt}
		\begin{align}
		\hspace{-15pt}\pi_{t}(x_{t},s_{t})=\frac{\mathbb{P}_{\pi_{t-1}}^{\sigma^*_t}\{x_{t},s_{t},z_{t}\}}{\mathbb{P}_{\pi_{t-1}}^{\sigma^*_t}\{z_{t}\}}, \quad\forall x_{t}\in\mathcal{X}_{t},\forall s_{t}\in\mathcal{S}_{t},\label{eq:CIBconsistency-on}
		\end{align}
		\item for all $t\hspace*{-2pt}\in\hspace*{-2pt}\mathcal{T}\backslash\{1\}$, $c_{t-1}\hspace*{-2pt}\in\hspace*{-2pt}\mathcal{C}_{t-1}$, $\pi_{t-1}\hspace*{-2pt}=\hspace*{-2pt}\gamma_{t-1}\hspace*{-1pt}(\hspace*{-1pt}c_{t-1}\hspace*{-1pt})$, $z_{t}\hspace*{-2pt}\in\hspace*{-2pt}\mathcal{Z}_{t}$, and $\pi_t\hspace*{-2pt}=\hspace*{-2pt}\gamma_t\hspace*{-1pt}(\hspace*{-1pt}\{\hspace*{-1pt}c_{t-1}\hspace*{-1pt},\hspace*{-1pt}z_t\hspace*{-1pt}\}\hspace*{-1pt})$ such that $\mathbb{P}^{\sigma^*_t}_{\pi_{t-1}}\hspace*{-1pt}\{\hspace*{-1pt}z_{t}\hspace*{-1pt}\}\hspace*{-2pt}=\hspace*{-2pt}0$, we have\vspace*{-2pt}
		\begin{align*}
		\pi_{t}(x_{t},s_{t})>0,\quad\forall x_{t}\in\mathcal{X}_{t},\forall s_{t}\in\mathcal{S}_{t},\vspace*{-2pt}
		\end{align*}
		only if there exists an open-loop strategy $(\hspace*{-1pt}A_{1\hspace*{-1pt}:t-1}=a_{1\hspace*{-1pt}:t-1}\hspace*{-1pt})$ such that $\mathbb{P}_{\pi_1}^{(\hspace*{-1pt}A_{1\hspace*{-1pt}:t-1}=a_{1\hspace*{-1pt}:t-1}\hspace*{-1pt})}\hspace*{-1pt}\{\hspace*{-1pt}c_{t-1}\hspace*{-1pt},\hspace*{-1pt}z_t\hspace*{-1pt}\}\hspace*{-2pt}>\hspace*{-2pt}0$, and
				\begin{align}
				\mathbb{P}_{\pi_1}^{(A_{1:t-1}=a_{1:t-1})}\{x_{t},s_{t}\}> 0,
				\label{eq:CIBconsistency-off}
				\end{align}
		\item for all $t\hspace*{-2pt}\in\hspace*{-2pt}\mathcal{T}\backslash\{1\}$, $c_{t-1}\hspace*{-2pt}\in\hspace*{-2pt}\mathcal{C}_{t-1}$, $\pi_{t-1}\hspace*{-2pt}=\hspace*{-2pt}\gamma_{t-1}\hspace*{-1pt}(\hspace*{-1pt}c_{t-1}\hspace*{-1pt})$, $z_{t}\hspace*{-2pt}\in\hspace*{-2pt}\mathcal{Z}_{t}$, and $\pi_t\hspace*{-2pt}=\hspace*{-2pt}\gamma_t\hspace*{-1pt}(\hspace*{-1pt}\{\hspace*{-1pt}c_{t-1}\hspace*{-1pt},\hspace*{-1pt}z_t\hspace*{-1pt}\}\hspace*{-1pt})$ such that $\mathbb{P}^{\sigma^*_t}_{\pi_{t-1}}\hspace*{-1pt}\{\hspace*{-1pt}z_{t}\hspace*{-1pt}\}\hspace*{-2pt}=\hspace*{-2pt}0$, we have
		\begin{align*}
		\sum_{x_{t}\in\mathcal{X}_{t}}\pi_{t}(x_{t},s_{t})>0, \quad\forall s_{t}\in\mathcal{S}_{t}\vspace*{-2pt}
		\end{align*}
		if there exists an open-loop strategy $\hspace*{-1pt}(\hspace*{-1pt}A_{1\hspace*{-1pt}:t-1}\hspace*{-2pt}=\hspace*{-2pt}a_{1\hspace*{-1pt}:t-1}\hspace*{-1pt})$ such that $\mathbb{P}_{\pi_1}^{(\hspace*{-1pt}A_{1\hspace*{-1pt}:t-1}\hspace*{-1pt}=a_{1\hspace*{-1pt}:t-1}\hspace*{-1pt})}\hspace*{-1pt}\{\hspace*{-1pt}c_{t-1}\hspace*{-1pt},\hspace*{-1pt}z_t\hspace*{-1pt}\}\hspace*{-2pt}>\hspace*{-2pt}0$, and
		\begin{align}
		\sum_{x_{t}\in\mathcal{X}_{t}}\mathbb{P}_{\pi_1}^{(A_{1:t-1}=a_{1:t-1})}\{x_{t},s_{t}\}> 0.
		\label{eq:CIBconsistency-positive}\vspace*{-2pt}
		\end{align}
	\end{enumerate}
	\label{def:consistencycommon}
\end{definition} 

Parts (i) and (ii) of Definition \ref{def:consistencycommon} follow from rationales similar to their analogues  in Definition \ref{def:consistency}, and require a SIB belief system to satisfy a sets of constraints with respect to a SIB strategy profile that are similar to those for an assessment $(g^*\hspace*{-2pt},\hspace*{-1pt}\mu)$. Definition \ref{def:consistencycommon} requires an additional condition described by part (iii). By (\ref{eq:CIBconsistency-positive}), we require a SIB belief system $\gamma$ consistent with the SIB strategy profile $\sigma^*$ to assign a positive probability to every realization $s_t$ of the agents' sufficient private information $S_{t}$ that is ``plausible'' given the common information realization $c_{t}\hspace*{-2pt}=\hspace*{-2pt}\{\hspace*{-1pt}c_{t-1}\hspace*{-1pt},\hspace*{-1pt}z_t\hspace*{-1pt}\}$; plausibility of $s_{t}$ given $c_t$ means that there exists an open-loop strategy profile $(\hspace*{-1pt}A_{1\hspace*{-1pt}:t-1}\hspace*{-2pt}=\hspace*{-2pt}a_{1\hspace*{-1pt}:t-1}\hspace*{-1pt})$ consistent with the realization $c_t$ that leads to the realization of $s_{t}$ with positive probability. Therefore, part (iii) ensures that there exists no possible incompatibility between the SIB belief $\Pi_{t}$ and the agents' sufficient private information $S_{t+1}$. As we show later (Section \ref{subsec:CIBPBE}), such a compatibility condition allows each agent to refine the SIB belief $\Pi_t$ using his own private sufficient information $S_t^i$, and to form his private belief about the game. 

\begin{remark}\label{remark:CIBoff}
	Assumptions \ref{assump:positivetrans} and \ref{assump:positiveprob} imply that (\ref{eq:CIBconsistency-off}) holds for all ${(A_{1\hspace*{-1pt}:t-1}=a_{1\hspace*{-1pt}:t-1})}$ such that  $\mathbb{P}^{(\hspace*{-1pt}A_{1\hspace*{-1pt}:t-1}\hspace*{-1pt}=a_{1\hspace*{-1pt}:t-1}\hspace*{-1pt})}\hspace*{-1pt}$ $\{\hspace*{-1pt}c_{t}\hspace*{-1pt}\}\hspace*{-2pt}>\hspace*{-2pt}0$. Therefore, in the rest of the paper, we ignore part (ii) of the consistency condition for SIB belief systems. Moreover, under Assumptions \ref{assump:positivetrans} and \ref{assump:positiveprob}, condition (\ref{eq:CIBconsistency-positive}) is always satisfied. Therefore, condition (iii) is equivalent to having $\sum_{x_t\hspace*{-1pt}\in\mathcal{X}_t}\hspace*{-2pt}\pi_{t}\hspace*{-1pt}(\hspace*{-1pt}x_t,\hspace*{-1pt}p_t\hspace*{-1pt})\hspace*{-2pt}>\hspace*{-2pt}0$ whenever $\mathbb{P}^{\sigma^*}_{\pi_{t-1}}\hspace*{-2pt}\{\hspace*{-1pt}z_{t}\hspace*{-1pt}\}\hspace*{-2pt}=\hspace*{-2pt}0$.  
\end{remark}

Given a SIB strategy profile prediction $\sigma^*$, a consistent SIB belief must satisfy (\ref{eq:CIBconsistency-on}), which determines the SIB belief $\Pi_{t}$ at $t$ in terms of the SIB belief $\Pi_{t-1}$ at $t-1$ and the new common information $Z_{t}$ at $t$. We define a \textit{SIB belief update rule} as a mapping $\psi_t\hspace*{-1pt}:\hspace*{-1pt}\Delta(\hspace*{-1pt}\mathcal{X}_{t-1}\hspace*{-1pt}\times\hspace*{-1pt} \mathcal{S}_{t-1}\hspace*{-1pt})\hspace*{-1pt}\times \hspace*{-1pt}\mathcal{Z}_{t}\hspace*{-1pt}\rightarrow\hspace*{-2pt} \Delta(\hspace*{-1pt}\mathcal{X}_t\hspace*{-1pt}\times\hspace*{-1pt} \mathcal{S}_t\hspace*{-1pt}), t\hspace*{-2pt}\in\hspace*{-2pt}\mathcal{T}$  that determines recursively the SIB belief
\begin{align}
\Pi_t^\psi:=\psi_t(\Pi_{t-1}^\psi,Z_{t}),\label{eq:CIBupdaterule}
\end{align}
as a function of new common observation $Z_{t}$ at $t$ and the SIB belief $\Pi_{t-1}^{\psi}$ at $t\hspace*{-2pt}-\hspace*{-2pt}1$.\footnote{Upon reaching an information set of measure zero (parts (ii) and (iii) of Definition \ref{def:consistencycommon}), the revised SIB belief could be a function of $C_t\hspace*{-2pt}=\hspace*{-2pt}\{\hspace*{-1pt}C_{t-1}\hspace*{-1pt},\hspace*{-1pt}Z_{t}\hspace*{-1pt}\}$, \red{rather than only} $\Pi_{t-1}\hspace*{-1pt}(\hspace*{-1pt}C_{t-1}\hspace*{-1pt})$ and $Z_{t}$. Therefore, the set of SIB belief systems that \red{can be} generated from SIB update rules is a subset of all consistent SIB belief systems given by Definition \ref{def:consistencycommon}. However, we argue that upon reaching an information set of measure zero, it is more plausible to revise the SIB belief only as a function of relevant information $\Pi_{t-1}\hspace*{-1pt}(\hspace*{-1pt}C_{t-1}\hspace*{-1pt})$ and $Z_{t}$; $C_t$ is irrelevant given $\Pi_{t-1}\hspace*{-1pt}(\hspace*{-1pt}C_{t-1}\hspace*{-1pt})$ and $Z_{t}$.}
 The superscript $\psi$ in $\Pi_t^{\psi}$ indicates that the SIB belief $\Pi_t^\psi$ is generated using the SIB update rule $\psi$. Let $\gamma^{\psi}$ denote the common belief system that is equivalent to the SIB update rule $\psi$. We call a SIB belief update rule $\psi$ consistent with a SIB strategy profile $\sigma^*$ if the equivalent SIB belief system $\gamma^{\psi}$ is consistent with $\sigma^*$ (Definition \ref{def:consistencycommon}).

Define a SIB assessment $(\sigma^*\hspace*{-2pt},\hspace*{-1pt}\gamma)$ as a pair of SIB strategy profile $\sigma^*$ and a SIB belief system $\gamma$. Below, we show that a consistent SIB assessment $(\sigma^*\hspace*{-2pt},\hspace*{-1pt}\gamma)$ is equivalent to a consistent assessment $(g^*\hspace*{-2pt},\hspace*{-1pt}\mu)$ as defined in Section \ref{sec:PBE} (Definition \ref{def:consistency}).

\begin{lemma}\label{lemma:equivalence}
	For any given SIB assessment $(\sigma^*\hspace*{-2pt},\hspace*{-1pt}\gamma)$, there exists an equivalent assessment $(g^*\hspace*{-2pt},\hspace*{-1pt}\mu)$ of a behavioral strategy prediction $g^*$ and belief system $\mu$ such that:
	\begin{enumerate}[i)]
		\item the behavioral strategy $g^*$ is defined by
		\begin{align}
		g^{*i}_t(h_t^i):=\sigma^{*i}_t(\pi_t^\gamma,s_t^i); 
		\label{eq:CIBstrategy}
		\end{align}
		\item the belief system $\mu$ is consistent with $g^*$ and satisfies
		\begin{align}
			\mathbb{P}^{g^*}\left\{s_t^{-i}|h_t^i\right\} = \mathbb{P}\left\{s_t^{-i}|\pi_t,s_t^i\right\},\label{eq:lemma1} 
		\end{align}	
		for all $i\hspace*{-1pt}\in\hspace*{-1pt}\mathcal{N}$, $t\hspace*{-1pt}\in\hspace*{-1pt}\mathcal{T}$,$h_t^i\hspace*{-1pt}\in\hspace*{-1pt}\mathcal{H}_t^i$, and $s_t^{-i}\hspace*{-1pt}\in\hspace*{-1pt}\mathcal{S}_t^{-i}$.
		\end{enumerate}
\end{lemma}

Lemma \ref{lemma:equivalence} shows that the set of consistent SIB assessment $(\sigma^*\hspace*{-2pt},\hspace*{-1pt}\gamma)$ is equivalent to a subset of consistent assessments $(g^*\hspace*{-2pt},\hspace*{-1pt}\mu)$. 
That is, using the SIB belief system $\gamma$ and SIB strategy profile $\sigma^*$, agents can form a consistent assessment about the evolution of the game. Moreover, condition (\ref{eq:lemma1}) implies that the SIB belief $\Pi_t$ along with agent $i$'s sufficient information $S_t^i$ capture all the information in $H_t^i$ that is relevant to agent $i$'s belief about $S_t^{-i}$. 

\vspace*{-5pt}
\subsection{Sufficient Information based PBE}\label{subsec:CIBPBE}

Using the result of Lemma \ref{lemma:equivalence}, we can define a class of PBE, called Sufficient Information based PBE (SIB-PBE), as a set of equilibria for dynamic games with asymmetric information that can be expressed as SIB assessments.

\begin{definition}
	A SIB assessment $(\sigma^*\hspace*{-2pt},\hspace*{-1pt}\gamma)$ is called a SIB-PBE if $\gamma$ is consistent with $\sigma^*$ (Definition \ref{def:consistencycommon}), and the equivalent consistent assessment $(g^*,\mu)$, given by Lemma \ref{lemma:equivalence}, is a PBE.   
\end{definition} 

In the sequel, we also call a consistent pair $(\sigma^*\hspace*{-2pt},\hspace*{-1pt}\psi)$ of a SIB strategy prediction profile $\sigma^*$ and SIB belief update rule $\psi$ a SIB-PBE if $(\sigma^*\hspace*{-2pt},\hspace*{-1pt}\gamma^{\psi})$ is a SIB-PBE.

Throughout Section \ref{sec:CIB}, we \red{assumed} that agents play according to the strategy predictions $g^*$ (or SIB strategy predictions $\sigma^*$). However, an agent's, say agent $i\hspace*{-2pt}\in\hspace*{-2pt}\mathcal{N}$'s, actual strategy $g^i$ is his private information and could be different from $g^*$ if such a deviation is profitable for him. The proposed class of SIB assessments imposes two restrictions on agents' strategies and beliefs compared to the general class of assessment presented in Section \ref{sec:PBE}. First, it requires that each agent $i$, $i\hspace*{-2pt}\in\hspace*{-2pt}\mathcal{N}$, must play a SIB strategy $\sigma^{*i}$ instead of a general behavioral strategy $g^{*i}$. Second, it requires that each agent $i$, $i\hspace*{-2pt}\in\hspace*{-2pt}\mathcal{N}$, must form a belief about the status of the game using only the SIB belief $\Pi_t$ along with his sufficient private information $S_t^i$ (instead of a general belief $\mu_t^i$). A strategic agent $i\hspace*{-2pt}\in\hspace*{-2pt}\mathcal{N}$ does not restrict his choice of strategy to SIB strategies, and may deviate from $\sigma^{*i}$  to a non-SIB  strategy $g^i$ if it is profitable to him. Moreover, a strategic agent $i$ does not limit himself to form belief about the current status of the game only based on $\Pi_t$ and $S_t^i$, and may instead use a general belief $\mu^i$ if it enables him to improve his expected utility. In the next section, we address these strategic concerns, and show that no agent $i\hspace*{-2pt}\in\hspace*{-2pt}\mathcal{N}$ wants to deviate from $(\Pi,\hspace*{-1pt}\sigma^*\hspace*{-1pt})$ and play a non-SIB strategy  $g^i$ when all other agents are playing according to SIB assessment $(\Pi,\hspace*{-1pt}\sigma^*\hspace*{-1pt})$.  This result allows us to focus on the class of SIB assessments, and develop a methodology to sequentially decompose the dynamic game over time.

%% file: CIBPBE.tex
\section{Main Results}\label{sec:CIB-PBE}

In this section, we show that the class of SIB assessments is rich enough to capture the agents' strategic interactions.
We first show that when agents $-i$ play according to a SIB assessment $(\sigma^*\hspace*{-2pt},\hspace*{-1pt}\gamma)$, agent $i$, $i\hspace*{-2pt}\in\hspace*{-2pt}\mathcal{N}$, cannot mislead these agents by playing a strategy $g^i$ different from $\sigma^{*i}$, thus, creating dual beliefs, one belief that is based on the SIB assessment $(\sigma^*\hspace*{-2pt},\hspace*{-1pt}\gamma)$ the functional form of which is known to all agents, and another belief that is based on his private strategy $g^i$ that is only known to him (Theorem \ref{thm:beliefindependence}). Then, we show that given that agents $-i$ play SIB strategy  $\sigma^{*-i}$, agent $i$ has a response that is a SIB strategy (Theorem \ref{thm:closeness}). 

	The result of Theorems \ref{thm:beliefindependence} (resp. \ref{thm:closeness}) for agent $i\in\mathcal{N}$ assumes that all other agents $-i$ are playing according to strategy prediction $g^{*-i}$ (resp. $\sigma^{*-1}$). The same results hold for every continuation game that starts at any time $t\in\mathcal{T}$ along an off-equilibrium path; they can be proved by relabeling  time $t$ as time $1$, and using the SIB belief $\pi_t=\gamma_t(c_t)$ and the corresponding belief $\mu_t^i(h_t^i)$, defined by Lemma \ref{lemma:equivalence}, as the initial common belief for the continuation game.

Using the results of Theorems \ref{thm:beliefindependence} and \ref{thm:closeness}, we present a methodology to determine the set of SIB-PBEs of stochastic dynamic games with asymmetric information (Theorem 3). The proposed methodology leads to a sequential decomposition of stochastic dynamic games with asymmetric information. This decomposition gives rise to a dynamic program that can be utilized to compute SIB-PBEs via backward induction. We proceed by stating the following results from the companion paper \cite[Theorem 1]{team}. 


\begin{theorem}[Policy-independence belief property {\cite{team}}]\label{thm:beliefindependence} \hspace*{10pt}
	
	(i) Consider a general strategy prediction profile $g^*$.
	If agents $-i$ play according to strategy predictions $g^{*-i}$, then for every strategy $g^i$ that agent $i$ plays,
	\begin{align}
	\mathbb{P}^{g^*\hspace*{-1pt},g^{-i}}\left\{x_{t},p_t^{-i}\Big|h_t^i\right\}=\mathbb{P}^{g^{*-i}}\left\{x_{t},p_t^{-i}\Big|h_t^i\right\}.\label{eq:beliefindependence1}
	\end{align}
	
	(ii) Consider a SIB strategy prediction profile $\sigma^*$ along with the associated consistent update rule $\psi$.
	If agents $-i$ play according to SIB strategy predictions $\sigma^{*-i}$, then for every general strategy ${g}^i$ that agent $i$ plays,
	\begin{align}
	\hspace*{-6pt}\mathbb{P}^{\sigma^{*-i}\hspace*{-2pt},g^i}_{\psi}\hspace*{-3pt}\left\{\hspace*{-1pt}x_{t},p_t^{-i}\Big|h_t^i\hspace*{-1pt}\right\}\hspace*{-2pt}=\hspace*{-2pt}\mathbb{P}^{\sigma^{*-i}}_{\psi}\hspace*{-3pt}\left\{\hspace*{-1pt}x_{t},p_t^{-i}\Big|h_t^i\hspace*{-1pt}\right\}\hspace*{-2pt}.\label{eq:beliefindependence2}
	\end{align}
\end{theorem}

Part (i) of Theorem \ref{thm:beliefindependence} states 
that under perfect recall agent $i$'s belief is independent of his actual strategy $g^i$. Therefore, agent $i$ cannot mislead agents $-i$ by deviating from the SIB strategy prediction $g^{*i}$ to a behavioral strategy $g^i$ so as to create dual beliefs (described above) that he can use to his advantage. Part (ii) of Theorem \ref{thm:beliefindependence} concerns situations where all agents except agent $i$ play SIB strategies $\sigma^{*-i}$. Given SIB update rule $\psi$ that is consistent with $\sigma^*$,  it states that agent $i$'s beliefs independent of his actual strategy $g^i$. We note when agents $-i$ play SIB strategies  the consistency condition for SIB update rule $\psi$ depends on strategy prediction $\sigma^{*i}$ they have about agent $i$.

Using the result of Theorem \ref{thm:beliefindependence}, we show that agent $i\hspace*{-2pt}\in\hspace*{-2pt}\mathcal{N}$ does not gain by playing a non-SIB strategy $\tilde{g}^i$ when all other agents $-i$ are playing SIB strategies $\sigma^{*-i}$.

\begin{theorem}[Closedness of SIB strategies] Consider a consistent SIB assessment $(\sigma^*\hspace*{-2pt},\hspace*{-1pt}\gamma^\psi)$ where $\psi$ is a SIB update rule consistent with $\sigma^*$. If every agent $j\hspace*{-2pt}\in\hspace*{-2pt}\mathcal{N}$, $j\hspace*{-2pt}\neq\hspace*{-2pt} i$ plays the SIB strategy $\sigma^{*j}$, then, there exists a SIB strategy $\sigma^i$ for agent $i$ that is a best response to $\sigma^{*-i}$.
\label{thm:closeness}
\end{theorem}


The results of Theorems \ref{thm:beliefindependence} and \ref{thm:closeness} address the two restrictions (discussed above) imposed in SIB assessments on the agents' beliefs and strategies, respectively. 

Based on these results, we restrict attention to SIB assessments, and provide a sequential decomposition of dynamic games with asymmetric information. A SIB-PBE is SIB assessment that is a fixed point under the best response map for all agents. Below, we formulate a dynamic program that enables us to compute SIB-PBEs of dynamic games with asymmetric information.

Consider a dynamic program over time horizon $\mathcal{T}\hspace*{-1pt}\cup\hspace*{-1pt}\{\hspace*{-1pt}T\hspace*{-1pt}+\hspace*{-1pt}1\hspace*{-1pt}\}$ with information state $\{\hspace*{-1pt}\Pi_t,\hspace*{-1pt}S_t\hspace*{-1pt}\}$, $t\hspace*{-2pt}\in\hspace*{-2pt}\mathcal{T}$. Let $V_t\hspace*{-2pt}:=\hspace*{-2pt}\{\hspace*{-1pt}V_t^i\hspace*{-1pt}:\hspace*{-1pt}\Delta(\hspace*{-1pt}\mathcal{X}_t\hspace*{-1pt}\times\hspace*{-1pt}\mathcal{S}_t\hspace*{-1pt})\hspace*{-1pt}\times\hspace*{-1pt}\mathcal{S}_t\hspace*{-1pt}\rightarrow\hspace*{-1pt}\mathbb{R},i\hspace*{-2pt}\in\hspace*{-2pt}\mathcal{N}\}$ denote the value function that captures the continuation payoffs for all agents, for all realizations of the SIB belief $\Pi_t$ and the agents' private sufficient information $S_t$, $t\in \mathcal{T}$. Set $V_{T+1}^i=0$ for all $i\in\mathcal{N}$. For each stage $t\in\mathcal{T}$ of the dynamic program consider the following static game.

\textbf{Stage game $\mathbf{G}_t(\hspace*{-1pt}\pi_t\hspace*{-1pt},\hspace*{-2pt}V_{t+1}\hspace*{-1pt},\hspace*{-1pt}\psi_{t+1}\hspace*{-1pt})$:} Given the value function $V_{t+1}$ and SIB update rule $\psi_{t+1}$, we define the stage game $\mathbf{G}_t(\hspace*{-1pt}\pi_t\hspace*{-1pt},\hspace*{-2pt}V_{t+1}\hspace*{-1pt},\hspace*{-1pt}\psi_t\hspace*{-1pt})$ as a static game of asymmetric information among agents for every realization $\pi_t$. Each agent $i\hspace*{-2pt}\in\hspace*{-2pt}\mathcal{N}$ has private information $S_t^i$ that is distributed according to $\pi_t$, which is common knowledge among the agents. Given a realization $a_t$ of the agents' collective action profile and a realization $s_t$ of the agents' sufficient private information, agent $i$'s utility is given by
\begin{align}
&\hspace*{-8pt}\bar{U}_t^i\hspace*{-1pt}(\hspace*{-1pt}a_t\hspace*{-1pt},\hspace*{-2pt}s_t\hspace*{-1pt},\hspace*{-2pt}\pi_t\hspace*{-1pt},\hspace*{-2pt}V_{t+1}\hspace*{-1pt},\hspace*{-2pt}\psi_{t+1}\hspace*{-1pt})\hspace*{-2pt}:=\mathbb{E}_{\pi_t}\hspace*{-2pt}\left\{\hspace*{-1pt}u_t^i\hspace*{-1pt}(\hspace*{-1pt}X_t\hspace*{-1pt},\hspace*{-2pt}a_t\hspace*{-1pt})\hspace*{-2pt}+\hspace*{-2pt}V_{t+1}^i\hspace*{-1pt}\left(\hspace*{-1pt}\psi_{t+1}\hspace*{-1pt}(\hspace*{-1pt}\pi_t\hspace*{-1pt},\hspace*{-2pt}Z_{t+1}\hspace*{-1pt})\hspace*{-1pt},\hspace*{-1pt}S_{t+1}\hspace*{-1pt}\right)\hspace*{-1pt}\Big|\pi_t\hspace*{-1pt},\hspace*{-2pt}s_t\hspace*{-1pt},\hspace*{-2pt}a_t\hspace*{-1pt}\right\}\hspace*{-1pt}.\hspace*{-5pt}\label{eq:stageutility}
\end{align}  

\textbf{BNE correspondence:} We define the correspondence $\text{\textbf{BNE}}_t\hspace*{-1pt}\left(\hspace*{-1pt}V_{t+1}\hspace*{-1pt},\hspace*{-2pt}\psi_{t+1}\hspace*{-1pt}\right)$, $t\hspace*{-2pt}\in\hspace*{-2pt}\mathcal{T}$, as the correspondence mapping that characterizes the set of BNEs of the stage game $\mathbf{G}_t\hspace*{-1pt}(\hspace*{-1pt}\pi_t\hspace*{-1pt},\hspace*{-2pt}V_{t+1}\hspace*{-1pt},\hspace*{-2pt}\psi_{t+1}\hspace*{-1pt})$ for every realization of $\pi_t$; this correspondence is given by
\begin{align}
&\text{\textbf{BNE}}_t\hspace*{-1pt}\left(\hspace*{-1pt}V_{t+1}\hspace*{-1pt},\hspace*{-2pt}\psi_{t+1}\hspace*{-1pt}\right)\hspace*{-2pt}:=\hspace*{-2pt}\big\{\hspace*{-1pt}\sigma_t^*\hspace*{-3pt}:\hspace*{-2pt}\forall \pi_t\hspace*{-2pt}\in\hspace*{-2pt}\Delta(\hspace*{-1pt}\mathcal{X}_t\hspace*{-2pt}\times\hspace*{-2pt}\mathcal{S}_t\hspace*{-1pt}), \hspace*{3pt}\sigma(\hspace*{-1pt}\pi_t,\hspace*{-1pt}\cdot\hspace*{-1pt})\hspace*{-1pt} \text{ is a BNE of } \mathbf{G}_t\hspace*{-1pt}(\hspace*{-1pt}\pi_t\hspace*{-1pt},\hspace*{-2pt}V_{t+1}\hspace*{-1pt},\hspace*{-2pt}\psi_{t+1}\hspace*{-1pt})\hspace*{-1pt}\big\}\hspace*{-1pt}.\hspace*{-5pt}\label{eq:BNEcorrespondence}
\end{align}
We say $\sigma^*\hspace*{-1pt}(\hspace*{-1pt}\pi_t,\hspace*{-1pt}\cdot\hspace*{-1pt})$ is a BNE of the stage game $\mathbf{G}_t\hspace*{-1pt}(\hspace*{-1pt}\pi_t\hspace*{-1pt},\hspace*{-2pt}V_{t+1}\hspace*{-1pt},\hspace*{-2pt}\psi_t\hspace*{-1pt})$ if for all agents $i\hspace*{-2pt}\in\hspace*{-2pt}\mathcal{N}$, and for all $s_t^i\hspace*{-2pt}\in\hspace*{-2pt}\mathcal{S}_t^i$,
\begin{align}
&\hspace*{-4pt}\sigma^{\hspace*{-1pt}*i}\hspace*{-1pt}(\hspace*{-1pt}\pi_t\hspace*{-1pt},\hspace*{-2pt}s_t^i\hspace*{-1pt})\hspace*{-2pt}\in \argmax_{\alpha\in \Delta(\mathcal{A}_t^i)}\mathbb{E}_{\pi_t}\hspace*{-3pt}\left\{\hspace*{-2pt}\bar{U}_t^i\hspace*{-1pt}(\hspace*{-1pt}(\hspace*{-1pt}\alpha,\hspace*{-2pt}\sigma^{\hspace*{-1pt}*-i}\hspace*{-1pt}(\hspace*{-1pt}\pi_t\hspace*{-1pt},\hspace*{-2pt}S_t^{-i}\hspace*{-1pt})\hspace*{-1pt})\hspace*{-1pt},\hspace*{-2pt}S_t\hspace*{-1pt},\hspace*{-2pt}\pi_t\hspace*{-1pt},\hspace*{-2pt}V_{t+1}\hspace*{-1pt},\hspace*{-2pt}\psi_{t+1}\hspace*{-1pt})\hspace*{-1pt}\Big|\pi_t\hspace*{-1pt},\hspace*{-2pt}s_t^i\hspace*{-1pt}\right\}\hspace*{-2pt}.\hspace*{-5pt}
\end{align}

Below, we provide a sequential decomposition of dynamic games with asymmetric information using the stage game and the BNE correspondence defined above. 

\begin{theorem}[Sequential decomposition] \label{thm:decomposition} A pair $(\sigma^*\hspace*{-2pt},\hspace*{-1pt}\psi)$ of a SIB strategy profile $\sigma^*$ and a SIB update rule $\psi$ (equivalently, a SIB assessment $(\sigma^*\hspace*{-2pt},\hspace*{-1pt}\gamma^\psi)$) is a SIB-PBE if  $(\sigma^*,\psi)$ solves the following dynamic program:
		
	\indent for $t\hspace*{-2pt}=\hspace*{-2pt}T\hspace*{-2pt}+\hspace*{-2pt}1$:\vspace*{-2pt}
	\begin{align}
	&V_{T+1}^i(\cdot):=0		\quad\quad \forall\;i\in\mathcal{N};\label{eq:decomposition-initial}\vspace*{-3pt}
	\end{align}
	\indent for $t\in\mathcal{T}$:\vspace*{-2pt}
	\begin{align}
	&\hspace*{-10pt}\sigma^*_t\in \text{\textbf{BNE}}_t\left(V_{t+1},\psi_{t+1}\right), \label{eq:decomposition-BNE}\\
	&\hspace*{-10pt}\psi_{t+1} \text{ is consistent with } \sigma^*, \label{eq:decomposition-update}\\
	&\hspace*{-10pt}V_t^i\hspace*{-1pt}(\hspace*{-1pt}\pi_t\hspace*{-1pt},\hspace*{-1pt}s_t\hspace*{-1pt})\hspace*{-2pt}:=\hspace*{-2pt}\mathbb{E}_{\psi_{t+1}}^{\sigma^*_t}\hspace*{-3pt}\left\{\hspace*{-2pt}\bar{U}_t^i\hspace*{-1pt}(\hspace*{-1pt}(\hspace*{-1pt}\sigma^{\hspace*{-1pt}*}\hspace*{-1pt}(\hspace*{-1pt}\pi_t\hspace*{-1pt},\hspace*{-2pt}S_t\hspace*{-1pt})\hspace*{-1pt})\hspace*{-1pt},\hspace*{-2pt}S_t\hspace*{-1pt},\hspace*{-2pt}\pi_t\hspace*{-1pt},\hspace*{-2pt}V_{t+1}\hspace*{-1pt},\hspace*{-2pt}\psi_{t+1}\hspace*{-1pt})\hspace*{-1pt}\Big|\pi_t\hspace*{-1pt},\hspace*{-2pt}s_t^i\hspace*{-2pt}\right\}\hspace*{-3pt},\hspace*{10pt}\forall i\hspace*{-2pt}\in\hspace*{-2pt}\mathcal{N}\hspace*{-1pt}.\hspace*{-5pt}\label{eq:decomposition-value}\vspace*{-2pt}
	\end{align}
\end{theorem}

\vspace*{-3pt}

%% file: Discussion.tex
\vspace*{-10pt} 	
\section{Discussion}\label{sec:discussion}
%
\label{subsec:PBE}

 We showed that SIB assessments proposed in this chapter are rich enough to capture a set of PBE. However, we would like to point out that the concept of SIB-PBE does not capture all PBEs of a dynamic game in general. We elaborate on the relation between the  of SIB-PBE and PBE below. We argue that the set of SIB-PBEs are more plausible to arise as the information asymmetry among the agents increases and the underlying system is dynamic.
 
 We presented an approach to compress the agents' private and common information by providing conditions sufficient  to characterize the information that is relevant for decision making purposes. Such information compression means that the agents do not incorporate into their decision making processes their observations that are irrelevant to the continuation game. As we show in \cite{team}, this information compression is without loss of generality for dynamic team problems. However, this is not the case in dynamic games. In general, the set of SIB-PBEs of a dynamic game is a subset of all PBEs of that game. This is because in a dynamic game agents can incorporate their past irrelevant observations into their future decisions so as to create  rewards (resp. punishments) that incentivize the agents to play (resp. not play) specific actions over time. By compressing the agents' private and common information in SIB assessments, we do not capture such punishment/reward schemes that are based on past irrelevant observations. Below, we present an example where there exists a PBE that cannot be captured as a SIB-PBE. 

 Consider a two-agent repeated game with $T\hspace*{-2pt}=\hspace*{-2pt}2$ and a payoff matrix given in Table \ref{table:1}. At each stage, agent $1$ chooses from $\{\hspace*{-1pt}U\hspace*{-1pt},\hspace*{-1pt}D\hspace*{-1pt}\}$, and agent $2$ chooses from $\{\hspace*{-1pt}L\hspace*{-1pt},\hspace*{-1pt}M\hspace*{-1pt},\hspace*{-1pt}R\}$. We assume that agents' actions are observable. Therefore, the agents have no private information, and the sufficient private information and SIB belief are trivial. The stage game has two equilibria in pure strategies given by $(\hspace*{-1pt}D\hspace*{-1pt},\hspace*{-1pt}M)$ and $(U\hspace*{-1pt},\hspace*{-1pt}R)$. Using the results of Theorem \ref{thm:decomposition}, we can characterize four SIB-PBEs of the repeated game that correspond to the different combinations of the two equilibria of the stage game as follows: $(\hspace*{-1pt}D\hspace*{-1pt}D,\hspace*{-1pt}M\hspace*{-1pt}M)$, $(		UU,\hspace*{-1pt}RR)$, $(\hspace*{-1pt}DU,\hspace*{-1pt}M\hspace*{-1pt}R)$, and $(U\hspace*{-1pt}D,\hspace*{-1pt}RM)$. However, there exists an another PBE of the repeated game that cannot be captured as a SIB-PBE. Consider the following equilibrium: \textit{Play $(U\hspace*{-1pt},\hspace*{-1pt}L)$ at $t\hspace*{-2pt}=\hspace*{-2pt}1$. If agent $2$ plays $L$ at $t\hspace*{-2pt}=\hspace*{-2pt}1$ then play $(U\hspace*{-1pt},\hspace*{-1pt}R)$; otherwise, play $(\hspace*{-1pt}D\hspace*{-1pt},\hspace*{-1pt}M)$ at $t\hspace*{-2pt}=\hspace*{-2pt}2$.} Note that agent $1$'s decision at $t\hspace*{-2pt}=\hspace*{-2pt}2$ depends on agent $2$'s action at $t\hspace*{-2pt}=\hspace*{-2pt}1$, which is a payoff-irrelevant information since the two stages of the game are independent. Nevertheless, agent $1$ utilizes agent $2$'s action at $t\hspace*{-2pt}=\hspace*{-2pt}1$, and \textit{punishes} him (agent $2$) by playing $D$ at $t\hspace*{-2pt}=\hspace*{-2pt}2$ if he deviates from playing $L$ at $t\hspace*{-2pt}=\hspace*{-2pt}1$.
 \vspace*{-3pt}

\begin{table}[h!]
\vspace*{-5pt} 	
	\centering
		 \hspace*{30pt}\caption{Payoff matrix}
		 \label{table:1}
		\begin{tabular}{ c|c|c|c| } 
			\multicolumn{1}{c}{}  & \multicolumn{1}{c}{L}  & \multicolumn{1}{c}{M} & \multicolumn{1}{c}{R}  \\ 	 	\cline{2-4}
			U & (8,3) & (0,2) & (2,10) \\ \cline{2-4} 
			D &(0,1) & (1,2) & (0,0) \\ \cline{2-4}
		 \end{tabular} \vspace*{-8pt}
	 \end{table} \vspace*{-8pt} 	
We would like to point out that there are instances of dynamic games with asymmetric information, such as zero-sum dynamic games \cite{sorin2002first}, where the equilibrium payoffs for the agents are unique. In these games it is not possible to incorporate pay-off irrelevant information so as to construct additional equilibria where the agents' payoffs are different from the ones corresponding to  SIB-PBEs; this is clearly the case  for zero-sum games since the agents do not cooperate on creating punishment/reward schemes due to the zero-sum nature of the game. In Section \ref{sec:existence}, we \red{utilize this fact} and  prove the existence of SIB-PBEs in zero-sum games. \red{Furthermore,} we show that SIB-PBEs are equivalent to PBEs, in terms of the agents' equilibrium payoffs, in these games.

While it is true that in general, the set of PBEs of a dynamic game is larger than the set of SIB-PBEs of that game, in the remainder of this section, we provide three reasons on why in a highly dynamic environment with information asymmetry among agents, SIB-PBEs are more plausible to arise as an outcome of a game.

First, we argue that in the face of a highly dynamic environment, an agent with partial observations of the environment should not behave fundamentally different whether he interacts in a strategic or cooperative environment. From the single-agent decision making point of view (\textit{i.e.} control theory), SIB strategies are the natural choice of an agent for decision making purposes (See Thm. 2 in the companion paper \cite{team}).\footnote{The SIB approach proposed in this paper for dynamic games along with the SIB approach to dynamic teams proposed in the companion paper \cite{team} provide a unified approach to the study of agents' behavior in a dynamic environment with information asymmetry among them.}

Second, we argue that in a highly dynamic environment with information asymmetry among the agents, the formation of punishment/reward schemes that utilize the agents' payoff-irrelevant information requires prior complex agreements among the agents; these complex agreements are sensitive to the parameters of the model and are not very plausible to arise in practice when the decision making problem for each agent is itself a complex task.  We note that the set of  PBEs that cannot be captured as SIB-PBEs are the ones that utilize payoff-irrelevant information to create punishment/reward schemes in the continuation game as in the example above. However, such punishment/reward schemes require the agents to form a common agreement among themselves on how to utilize such payoff-irrelevant information and how to implement such punishment/reward schemes. The formation of such a common agreement among the agents is more likely in games where the underlying system is not highly dynamic  (as in repeated games \cite{mailath2006repeated}) and there is no much information asymmetry among agents. 
However, in a highly dynamic environment with information asymmetry among agents the formation of such common agreement becomes less likely for the following reasons. First, in those environments each agent's individual decision making process is described by a complex POMDP; thus, strategic agents are less likely to form a prior common agreement (that depend on the solution of the individual POMDPs) in addition to solving their individual POMDPs. Second, as the information asymmetry increases among agents, punishment/reward schemes that utilize payoff-irrelevant information require a complex agreement among the agents that is sensitive and not robust to changes in the assumptions on the information structure of the game. For instance, consider the example described above, but assume that agents observe imperfectly each others' actions at each stage (Assumption \ref{assump:positiveprob}). Let $1\hspace*{-1pt}-\hspace*{-1pt}\epsilon$, $\epsilon\hspace*{-2pt}\in\hspace*{-2pt}(0,\hspace*{-1pt}1)$, denote the probability that agents observe each others' actions perfectly, and $\epsilon$ denote the probability that their observation is different from the true action of the other agent. Then, the described non-SIB strategy profile above remains as a PBE of the game only if $\epsilon\hspace*{-2pt}\leq\hspace*{-1pt} \frac{1}{5}$. 
The author of \cite{miller2012robust} provides a general result on the robustness of above-mentioned punishment/reward schemes in repeated games; he shows that the set of equilibria that are robust to changes in information structure that affect only payoff-irrelevant signals does not include the set of equilibria that utilize punishment/reward schemes described above. 

Third, the proposed notion of SIB-PBE can be viewed as a generalization of Markov Perfect Equilibrium \cite{maskin2001markov} to dynamic games with asymmetric  information. Therefore, a similar set of rationales that support the notion of MPE also applies to the notion of SIB-PBE as follows. First, the set of SIB assessments describe the simplest form  of strategies capturing the agents' behavior that is consistent with the agents' rationality. Second, the class of SIB assessments captures the notion that ``bygones are bygones'', which also underlies the requirement of subgame perfection in equilibrium concepts for dynamic games. That is, the agents' strategies in two continuation games that only differ in the agents' information about payoff-irrelevant events must be identical. 
Third, the class of SIB assessments embodies the principle that ``minor changes in the past should have minor effects''. This implies that if there exists a small perturbation in the specifications of the game or the agents' past strategies that are irrelevant to the continuation game, the outcome of the continuation game should not change drastically. The two-step example above presents one such situation, where one equilibrium that is not SIB-PBE  disappears suddenly as  $\epsilon\hspace*{-1pt}\rightarrow \hspace*{-1pt}\frac{1}{5}$. 


%% file: Existence.tex
\vspace*{-3pt}
\section{Existence of Equilibria} \label{sec:existence}\vspace*{-2pt}

As we discussed in Section \ref{sec:discussion}, there exist PBEs that cannot be described as SIB-PBEs in general. Therefore, the standard results that guarantee the existence of a PBE for dynamic games with asymmetric information \cite[Proposition. 249.1]{osborne1994course} cannot be used to guarantee the existence of a SIB-PBE in these games. In this Section, we discuss the existence of SIB-PBEs for dynamic games with asymmetric information. 
We provide conditions that are sufficient to guarantee the existence of SIB-PBEs  (Lemmas \ref{lemma:sufficient1} and \ref{lemma:sufficient2}). Using the result of Lemma \ref{lemma:sufficient1}, we prove the existence of SIB-PBEs for zero-sum dynamic games with asymmetric information (Theorem \ref{thm:zerosum}). 
Using the result of Lemma \ref{lemma:sufficient2}, we identify instances of non-zero-sum dynamic games with asymmetric information where we can guarantee the existence of SIB-PBEs.		\vspace*{-2pt}

\begin{lemma}\label{lemma:sufficient1}
	The dynamic program given by (\ref{eq:decomposition-BNE})-(\ref{eq:decomposition-value}) has at least one solution 
	at stage $t$ if the value function $V_{t+1}$ is continuous in $\pi_{t+1}$.		\vspace*{-2pt}
\end{lemma}

We note that the condition of Lemma \ref{lemma:sufficient1} is always satisfied for $t\hspace{-2pt}=\hspace{-2pt}T$ by definition of $V_{T+1}$; see (\ref{eq:stageutility}) and (\ref{eq:decomposition-initial}). However, for $t\hspace{-2pt}<\hspace{-2pt}T$, it is not straightforward, in general, to prove the continuity of the value function $V_t$ in $\pi_t$.  Given $V_{t+1}$ is continuous in $\pi_{t+1}$, the result of \cite[Theorem 2]{milgrom1985distributional} implies that the set of equilibrium payoffs for the state game at $t$ is \textit{upper-hemicontinuous} in $\pi_t$. Therefore, if the stage game ${\mathbf{G}_t\hspace*{-1pt}(\hspace*{-1pt}\pi_t\hspace{-1pt},\hspace{-2pt}V_{t+1}\hspace{-1pt},\hspace{-2pt}\psi_{t+1}\hspace*{-1pt})}$ has a unique equilibrium payoff for every $\pi_t$, we can show that $V_t$ is continuous in $\pi_t$ for $t\hspace*{-2pt}<\hspace*{-2pt}T$. Using this approach, we prove the existence of SIB-PBEs for zero-sum games below.  
		\vspace*{-2pt}

\begin{theorem}\label{thm:zerosum}
	For every dynamic zero-sum game with asymmetric information there exists a SIB-PBE that is a solution to the dynamic program given by (\ref{eq:decomposition-BNE})-(\ref{eq:decomposition-value}). 		\vspace*{-2pt}
\end{theorem}  

For dynamic non-zero-sum games, it is harder to establish that  $V_t$ is continuous in  $\pi_t$ for $t\hspace{-2pt}<\hspace{-2pt}T$ since the set of equilibrium payoffs is not a singleton in general. 
However, we conjecture that for every dynamic game with asymmetric information described in Section \ref{sec:model}, at every stage of the corresponding dynamic program, it is possible to select a BNE for every realization of $\pi_t$ so that the resulting $V_t$ is continuous in $\pi_t$.

In addition to the results of Lemma \ref{lemma:sufficient1} and Theorem \ref{thm:zerosum}, we provide another condition below that guarantees the existence of SIB-PBEs in some instances of dynamic games with asymmetric information.		\vspace*{-2pt}

\begin{lemma}\label{lemma:sufficient2}
	A dynamic game with asymmetric information described in Section \ref{sec:model} has at least one SIB-PBE 
		if there exits sufficient information $S_{1:t}^{1:N}$ such that the SIB update rule $\psi_{1:T}$ is independent of $\sigma^*$.		\vspace*{-2pt}
\end{lemma}

The independence of SIB update rule from $\sigma^*$ is a condition that is not satisfied for all dynamic games with asymmetric information. Nevertheless, we present below special instances where this condition is satisfied. 


\textbf{1) Nested information structure with one controller \cite{renault2006value}:} Consider the nested information structure case described in Section \ref{sec:model}. Assume that the evolution of the system is controlled only by the uniformed player and is given by $X_{t+1}\hspace{-2pt}=\hspace{-2pt}f_t(X_t\hspace{-1pt},\hspace{-1pt}A_t^2\hspace{-1pt},\hspace{-1pt}W_t)$. For $S_t^1\hspace{-2pt}=\hspace{-2pt}X_t$ and $S_t^2\hspace{-2pt}=\hspace{-2pt}0$, it is easy to check that $\mathbb{P}^{\sigma^*}\hspace{-2pt}\{\hspace{-1pt}\pi_{t+1}|\pi_t\hspace{-1pt},\hspace{-1pt}a_t\}\hspace{-2pt}=\hspace{-2pt}\mathbb{P}\{\hspace{-1pt}\pi_{t+1}|\pi_t\hspace{-1pt},\hspace{-1pt}a_t\hspace{-1pt}\}$ for all 
$t\hspace{-2pt}\in\hspace{-2pt}\mathcal{T}$. 

\textbf{2) Independent dynamics with observable actions and no private valuation \cite{ouyang2016TAC}:} Consider the model with independent dynamics and observable actions described in Section \ref{sec:model}. Assume that agent $i$'s, $i\hspace{-2pt}\in\hspace{-2pt}\mathcal{N}$, instantaneous utility is given by $u_t^i(A_t\hspace{-1pt},\hspace{-1pt}X_t^{-i})$ (no private valuation); that is, agent $i$'s utility at $t$ does not depend on $X_t^i$. It is easy to verify that $S_t^i=\emptyset$ is sufficient private information for agent $i$. Hence, the condition of Lemma \ref{lemma:sufficient2} is trivially satisfied.   

\textbf{3) Delayed sharing information structure with $\mathbf{d\hspace{-2pt}=\hspace{-2pt}1}$ \cite{ba1978two,basar1978decentralized}:} Consider a delayed sharing information structure where agents' actions and observations are revealed publicly with $d$ time step delay. When delay $d\hspace{-2pt}=\hspace{-2pt}1$, we have  $P_t^i\hspace{-2pt}=\hspace{-2pt}\{\hspace{-1pt}Y_{t}^i\hspace{-1pt}\}$. Let $S_t^i\hspace{-2pt}=\hspace{-2pt}P_t^i\hspace{-2pt}=\hspace{-2pt}Y_t^i$. Then, it is easy to verify that the condition of Lemma \ref{lemma:sufficient2} is satisfied. 

\textbf{4) Uncontrolled state process with hidden actions:} Consider an $N$-player game with uncontrolled dynamics given by $X_{t+1}\hspace{-2pt}=\hspace{-2pt}f_t(X_t,W_t)$, $t\hspace{-2pt}\in\hspace{-2pt}\mathcal{T}$. At every time $t\in\mathcal{T}$, agent $i$, $i\in\mathcal{N}$, receives a noisy observation $Y_t^i\hspace{-2pt}=\hspace{-2pt}O_t^i(X_t^i,Z_t^i)$. The agents' actions are hidden. Thus, $P_t^i\hspace{-2pt}=\hspace{-2pt}\{Y_{1:t}^i,A_{1:t-1}^i\}$ and $C_t=\emptyset$.  Hence, the condition of Lemma \ref{lemma:sufficient2} is trivially satisfied.   We note that in the case where a subset of the agents' observations is revealed to all agents' with some delay, \textit{i.e.} $C_t\hspace{-2pt}\subseteq\hspace{-2pt}\{Y_{1:t}\}$, the condition of Lemma \ref{lemma:sufficient2} is also satisfied.

%% file: Conclusion.tex
\section{Conclusion} \label{sec:conclusionG}		\vspace*{-2pt}

We proposed a general approach to study dynamic games with asymmetric information. We presented a set of conditions sufficient to characterize an information state for each agent that effectively compresses his common and private information in a mutually consistent manner. Along with the results in the companion paper \cite{team}, we showed that the characterized information state provides a sufficient statistic for decision making purposes in strategic and non-strategic settings. We introduced the notion of Sufficient Information based Perfect Bayesian Equilibrium that characterizes a set of outcomes \red{in} a dynamic game. We provided a sequential decomposition of the dynamic game over time, which leads to a dynamic program for the computation of the set of SIB-PBEs of the dynamic game. We determined conditions 	that guarantee the existence of SIB-PBEs. Using these conditions, we proved the existence of SIB-PBE for dynamic zero-sum games and for special instances of dynamic non-zero-sum games. 


%% file: Appendix.tex
\appendix\label{appG}

\begin{proof}[Proof of Lemma \ref{lemma:equivalence}]
	
	For any given consistent SIB assessment $(\sigma^*\hspace*{-1pt},\hspace*{-1pt}\gamma)$, let $g^*$ denote the behavioral strategy profile constructed according to (\ref{eq:CIBstrategy}). In the following, we construct recursively a belief system $\mu$ that is consistent with $g^*$ and satisfies (\ref{eq:lemma1}).

	For $t\hspace*{-1pt}=\hspace*{-1pt}1$, we have $P_1^i\hspace*{-2pt}=\hspace*{-2pt}Y_1^i$ and $C_1\hspace*{-2pt}=\hspace*{-2pt}Z_1$. Define
	\begin{align}
	\mu_t^i(h_1^i)(x_1,p_1^{-i}):=\frac{\mathbb{P}\{y_1,z_1|x_1\}\eta(x_1)}{\sum_{\hat{x}_1\in\mathcal{X}_1}\mathbb{P}\{y_1^i,z_1|\hat{x}_1\}\eta(\hat{x}_1)}.\label{eq:proof:lemma1-1}
	\end{align} 

	For $t\hspace*{-2pt}>\hspace*{-2pt}1$, if $\mathbb{P}_{\mu_{t-1}^i}^{g^*}\hspace*{-2pt}\{h_t^i|h_{t-1}^i\}\hspace*{-2pt}>\hspace*{-2pt}0$ (\textit{i.e.} no deviation from $g^*_{t-1}$ at $t\hspace*{-2pt}-\hspace*{-2pt}1$), define $\mu_t^i$ recursively by Bayes' rule,
	\begin{align}
	\mu_t^i(h_t^i)(x_t,p_t^{-i}):=\frac{\mathbb{P}_{\mu_{t-1}^i}^{g^*}\{h_t^i,x_{t},p_t^{-i}|h_{t-1}^i\}}{\mathbb{P}_{\mu_{t-1}^i}^{g^*}\{h_t^i|h_{t-1}^i\}}.\label{eq:proof:lemma1-2}
	\end{align}

	For $t\hspace*{-2pt}>\hspace*{-2pt}1$, if $\mathbb{P}_{\mu_{t-1}^i}^{g^*}\hspace*{-2pt}\{h_t^i|h_{t-1}^i\}\hspace*{-2pt}=\hspace*{-2pt}0$ (\textit{i.e.} there is a deviation from $g^*_{t-1}$ at $t\hspace*{-2pt}-\hspace*{-2pt}1$), define $\mu_t^i$ as,
	\begin{align}
	\mu_t^i(h_t^i)(x_t,p_t^{-i}):=\frac{\left|\mathcal{P}_t\right|}{\left|\mathcal{S}_t\right|}\frac{\gamma_t(c_t)(x_t,s_t)}{\sum_{\hat{s}_t^{-i}\in\mathcal{S}_t^{-i}}\gamma(c_t)(x_t,\hat{s}_t^{-i},s_t^i)},\label{eq:proof:lemma1-3}
	\end{align}
	where $s_t^j\hspace*{-1pt}=\hspace*{-1pt}\zeta_t^j(p_t^j,c_t;g_{1:t-1}^*)$ for all $j\hspace*{-1pt}\in\hspace*{-1pt}\mathcal{N}$.
	
	At $t\hspace*{-1pt}=\hspace*{-1pt}1$, (\ref{eq:lemma1}) holds by construction from (\ref{eq:proof:lemma1-1}).
	For $t>1$, 
	\begin{align*}
	\mathbb{P}^{g^*}\{S_t^{-i}|h_t^i\}&=\mathbb{P}^{g^*}\{S_t^{-i}|p_t^i,c_t\}=\mathbb{P}^{g^*}\{S_t^{-i}|s_t^i,c_t\}=\frac{\mathbb{P}^{g^*}\{S_t^{-i},s_t^i|c_t\}}{\mathbb{P}^{g^*}\{s_t^i|c_t\}}=\frac{\pi_t(S_t^{-i},s^i_t)}{\sum_{\hat{s}_t^{-i}\in\mathcal{S}_t^{-i}}\pi_t(\hat{s}_t^{-i},s_t^i)}\\&=\mathbb{P}\{S_t^{-i}|s_t^i,\pi_t\}
	\end{align*} 
	where the second equality follows from (\ref{eq:sufficientinfo}). Therefore, (\ref{eq:lemma1}) holds for all $t\hspace*{-1pt}\in\hspace*{-1pt}\mathcal{T}$.
\end{proof}


\begin{lemma}\cite[Lemma 2]{team}
	\label{lemma:lem}
	Given a SIB strategy profile $\sigma^*$ and update rule $\psi$ consistent with $\sigma^*$,
	\begin{align}
	\mathbb{P}^{\sigma^*}_\psi\hspace*{-1pt}\{S_{t+1},\hspace*{-1pt}\Pi_{t+1}|p_t,c_t,a_t\}\hspace*{-2pt}=\hspace*{-2pt}\mathbb{P}^{\sigma^*}_\psi\hspace*{-1pt}\{S_{t+1},\hspace*{-1pt}\Pi_{t+1}|s_t,\pi_t,a_t\}.
	\end{align}
	for all $s_t,\pi_t,a_t$.
\end{lemma}

\begin{proof}[Proof of Theorem \ref{thm:closeness}] Consider a ``super dynamic system'' as the collection of the original dynamic system along with agents $-i$ who play according to SIB assessment $(\sigma^*\hspace*{-1pt},\gamma^{\psi})$. This superdynamic system captures the system agent $i$ ``sees''. We establish the claim of Theorem \ref{thm:closeness} in two steps: (i) we show that from agent $i$'s viewpoint the super dynamic system is a POMDP, and (ii) we show that $\{\Pi_t,S_t^i\}$ is an information state for agent $i$ when he faces the super dynamic system with the original utility $u_{1:T}^i(\cdot,\cdot)$. Therefore, without loss of optimality, agent $i$ can choose his optimal decision strategy (best response) from the class of strategies that are functions of the information state $\{\Pi_t,S_t^i\}$, \textit{i.e.} the class of SIB strategies.

To establish step (i), consider $\tilde{X}\hspace*{-2pt}:=\hspace*{-2pt}\{\hspace*{-1pt}X_t\hspace*{-1pt},\hspace*{-1pt}\Pi_t,\hspace*{-1pt}S_t,\Pi_{t\hspace*{-1pt}-\hspace*{-1pt}1}\hspace*{-1pt},\hspace*{-1pt}S_{t\hspace*{-1pt}-\hspace*{-1pt}1}\hspace*{-1pt}\}$ as the system state at $t$ for the super dynamic system. Agent $i$'s observation at time $t$ is given by $\tilde{Y}_t^i\hspace*{-2pt}:=\hspace*{-2pt}\{\hspace*{-1pt}Y_t^i\hspace*{-1pt},\hspace*{-1pt}Z_t\hspace*{-1pt}\}$. To show that the super dynamic system is a POMDP, we need to show that it satisfies the following properties:\\
(a) it has a controlled Markovian dynamics, that is,
\begin{align}
\hspace*{-5pt}\mathbb{P}^{\sigma^*}_\psi\{\tilde{x}_{t+1}|\tilde{x}_{1:t},a_{1:t}^i,\tilde{y}_{1:t}^i\}\hspace*{-2pt}=\hspace*{-1pt}\mathbb{P}^{\sigma^{*-i}}_\psi\{\tilde{x}_{t+1}|\tilde{x}_{t},a_{t}^i\}, \hspace*{-3pt}\quad \forall t\hspace*{-1pt}\in\hspace*{-1pt}\mathcal{T}\hspace*{-1pt},\hspace*{-3pt}\label{eq:thm2-1}
\end{align}	
(b) agent $i$'s observation $\tilde{Y}_t^i$ is a function of system state 	$\tilde{X}_t$ along with the previous action $A_{t-1}^i$, that is,
\begin{align}
\hspace*{-4pt}\mathbb{P}^{\sigma^*}_\psi\hspace*{-1pt}\{\tilde{y}_{t}^i|\tilde{x}_{1:t},\hspace*{-1pt}a_{1:t-1}^i,\hspace*{-1pt}\tilde{y}_{1:t-1}^i\}\hspace*{-2pt}=\hspace*{-1pt}\mathbb{P}^{\sigma^{*-i}}\hspace*{-1pt}\{\tilde{y}_{t}^i|\tilde{x}_{t},\hspace*{-1pt}a_{t-1}^i\}, \hspace*{-7pt}\quad \forall t\hspace*{-1pt}\in\hspace*{-1pt}\mathcal{T}\hspace*{-1pt},\hspace*{-3pt}\label{eq:thm2-2}
\end{align}	
(c) agent $i$'s instantaneous utility at $t\hspace*{-1pt}\in\hspace*{-1pt}\mathcal{T}$ can be written as a function $\tilde{u}_t(\tilde{x}_t,\hspace*{-1pt}a_t^i)$ of system state $\tilde{X}_t$ along with his action $A_t^i$, that is,
\begin{align}
\mathbb{E}^{\sigma^*}_\psi\hspace*{-2pt}\left\{u_t^i(\hspace*{-1pt}X_t,\hspace*{-1pt}A_t)|\tilde{x}_{1:t},\hspace*{-1pt}a_{1:t}^i,\hspace*{-1pt}\tilde{y}_{1:t}^i\right\}\hspace*{-2pt}=\hspace*{-1pt}\tilde{u}_t(\tilde{x}_t,\hspace*{-1pt}a_t^i), \hspace*{-3pt}\quad \forall t\hspace*{-1pt}\in\hspace*{-1pt}\mathcal{T}\hspace*{-1pt}\label{eq:thm2-3}.
\end{align}
Property (a) is true because,
\begin{gather*}
\mathbb{P}^{\sigma^*}_\psi\{\tilde{x}_{t+1}|\tilde{x}_{1:t},a_{1:t}^i,\tilde{y}_{1:t}^i\}
=
\mathbb{P}^{\sigma^*}_\psi\hspace*{-1pt}\{x_{t\hspace*{-1pt}+\hspace*{-1pt}1}\hspace*{-1pt},\hspace*{-1pt}\pi_{t\hspace*{-1pt}+\hspace*{-1pt}1}\hspace*{-1pt},\hspace*{-1pt}s_{t\hspace*{-1pt}+\hspace*{-1pt}1}\hspace*{-1pt},\hspace*{-1pt}\pi_t\hspace*{-1pt},\hspace*{-1pt}s_t|x_{1:t}\hspace*{-1pt},\hspace*{-1pt}\pi_{1:t}\hspace*{-1pt},\hspace*{-1pt}s_{1:t}\hspace*{-1pt},\hspace*{-1pt}y_{1:t}^i\hspace*{-1pt},\hspace*{-1pt}z_{1:t}\hspace*{-1pt},\hspace*{-1pt}a_{1:t}^i\hspace*{-1pt}\}\\
=\\
\sum_{a_{t}^{-\hspace*{-1pt}i}\hspace*{-2pt},z_{t\hspace*{-1pt}+\hspace*{-1pt}1}\hspace*{-1pt},y_{t\hspace*{-1pt}+\hspace*{-1pt}1}}\hspace{-15pt}\mathbb{P}^{\sigma^*}_\psi\hspace*{-2pt}\{\hspace*{-1pt}x_{t\hspace*{-1pt}+\hspace*{-1pt}1}\hspace*{-1pt},\hspace*{-2pt}\pi_{t\hspace*{-1pt}+\hspace*{-1pt}1}\hspace*{-1pt},\hspace*{-2pt}s_{t\hspace*{-1pt}+\hspace*{-1pt}1}\hspace*{-1pt},\hspace*{-2pt}a_{t}^{-i}\hspace*{-2pt},\hspace*{-2pt}z_{t\hspace*{-1pt}+\hspace*{-1pt}1}\hspace*{-1pt},\hspace*{-2pt}y_{t\hspace*{-1pt}+\hspace*{-1pt}1}\hspace*{-1pt}|x_{1\hspace*{-1pt}:t}\hspace*{-1pt},\hspace*{-2pt}\pi_{1\hspace*{-1pt}:t}\hspace*{-1pt},\hspace*{-2pt}s_{1\hspace*{-1pt}:t}\hspace*{-1pt},\hspace*{-2pt}y_{1\hspace*{-1pt}:t}^i\hspace*{-1pt},\hspace*{-2pt}z_{1\hspace*{-1pt}:t}\hspace*{-1pt},\hspace*{-2pt}a_{1\hspace*{-1pt}:t}^i\hspace*{-1pt}\}\\
\stackrel{\text{by system dynamics (\ref{eq:systemdynamic1}) and (\ref{eq:systemdynamic2})}}{=}\\
\sum_{a_{t}^{-\hspace*{-1pt}i}\hspace*{-2pt},z_{t\hspace*{-1pt}+\hspace*{-1pt}1}\hspace*{-1pt},y_{t\hspace*{-1pt}+\hspace*{-1pt}1}}\hspace{-17pt}\Big[\hspace*{-1pt}\mathbb{P}^{\sigma^*}_\psi\hspace*{-2pt}\{\hspace*{-1pt}\pi_{t\hspace*{-1pt}+\hspace*{-1pt}1}\hspace*{-1pt},\hspace*{-2pt}s_{t\hspace*{-1pt}+\hspace*{-1pt}1}\hspace*{-1pt}|x_{1\hspace*{-1pt}:t}\hspace*{-1pt},\hspace*{-2pt}\pi_{1\hspace*{-1pt}:t}\hspace*{-1pt},\hspace*{-2pt}s_{1\hspace*{-1pt}:t}\hspace*{-1pt},\hspace*{-2pt}y_{1\hspace*{-1pt}:t}^i\hspace*{-1pt},\hspace*{-2pt}z_{1\hspace*{-1pt}:t}\hspace*{-1pt},\hspace*{-2pt}a_{1\hspace*{-1pt}:t}^i\hspace*{-1pt},\hspace*{-2pt}a_t^{-i}\hspace*{-1pt},\hspace*{-2pt}x_{t\hspace*{-1pt}+\hspace*{-1pt}1}\hspace*{-1pt},\hspace*{-2pt}z_{t\hspace*{-1pt}+\hspace*{-1pt}1}\hspace*{-1pt},\hspace*{-2pt}y_{t\hspace*{-1pt}+\hspace*{-1pt}1}\hspace*{-1pt}\}\\\hspace{45pt}\mathbb{P}\{\hspace*{-1pt}z_{t\hspace*{-1pt}+\hspace*{-1pt}1}\hspace*{-1pt},\hspace*{-2pt}y_{t\hspace*{-1pt}+\hspace*{-1pt}1}\hspace*{-1pt}|x_{t\hspace*{-1pt}+\hspace*{-1pt}1}\hspace*{-1pt},\hspace*{-2pt}a_t\hspace*{-1pt}\}\mathbb{P}\{\hspace*{-1pt}x_{t\hspace*{-1pt}+\hspace*{-1pt}1}\hspace*{-1pt}|x_t\hspace*{-1pt},\hspace*{-2pt}a_{t}\}\hspace*{-1pt}\sigma^{*-i}_{t}\hspace*{-1pt}(\hspace*{-1pt}\pi_t\hspace*{-1pt},\hspace*{-1pt}s_t^{-i}\hspace*{-1pt})(\hspace*{-1pt}a_t^{-i}\hspace*{-1pt})\hspace*{-1pt}\Big]\\
\hspace*{38pt}\stackrel{\fontsize{7}{9}\selectfont\begin{array}{c}\text{Define }\hspace*{-2pt} \hat{\mathcal{Z}}\hspace*{-1pt}:\hspace*{-1pt}=\hspace*{-1pt}\{\hspace*{-1pt}z_{t\hspace*{-1pt}+\hspace*{-1pt}1}\hspace*{-1pt}:\pi_{t\hspace*{-1pt}+\hspace*{-1pt}1}\hspace*{-1pt}=\psi_{t\hspace*{-1pt}+\hspace*{-1pt}1}\hspace*{-1pt}(\hspace*{-1pt}\pi_t\hspace*{-1pt},z_{t\hspace*{-1pt}+\hspace*{-1pt}1}\hspace*{-1pt})\hspace*{-1pt}\} \text{ \& }\\ \hat{\mathcal{Y}}\hspace*{-1pt}(\hspace*{-1pt}z_{t\hspace*{-1pt}+\hspace*{-1pt}1}\hspace*{-1pt})\hspace*{-1pt}:\hspace*{-1pt}=\{\hspace*{-1pt}y_{t\hspace*{-1pt}+\hspace*{-1pt}1}\hspace*{-1pt}:s_{t\hspace*{-1pt}+\hspace*{-1pt}1}^j\hspace*{-1pt}=\phi_{t\hspace*{-1pt}+\hspace*{-1pt}1}^j\hspace*{-1pt}(\hspace*{-1pt}s_t^j\hspace*{-1pt},\hspace*{-1pt}\{\hspace*{-1pt}y_{t\hspace*{-1pt}+\hspace*{-1pt}1}^j\hspace*{-1pt},z_{t\hspace*{-1pt}+\hspace*{-1pt}1}\hspace*{-1pt},a_t^j\hspace*{-1pt}\}\hspace*{-1pt})\hspace*{-1pt},\hspace*{-1pt} \forall j\hspace*{-1pt}\}\end{array}}{=}\\
\sum_{\scriptsize\begin{array}{c}a_{t}^{-i}\hspace*{-2pt},z_{t\hspace*{-1pt}+\hspace*{-1pt}1}\hspace*{-1pt}\in\hspace*{-1pt}\hat{\mathcal{Z}}\hspace*{-1pt},\\y_{t\hspace*{-1pt}+\hspace*{-1pt}1}\hspace*{-1pt}\in\hat{\mathcal{Y}}(\hspace*{-1pt}z_{t\hspace*{-1pt}+\hspace*{-1pt}1}\hspace*{-1pt})\end{array}}\hspace{-23pt}\Big[\mathbb{P}^{}_\psi\hspace*{-2pt}\{\hspace*{-1pt}\pi_{t\hspace*{-1pt}+\hspace*{-1pt}1}\hspace*{-1pt},\hspace*{-2pt}s_{t\hspace*{-1pt}+\hspace*{-1pt}1}\hspace*{-1pt}|\hspace*{-1pt}s_t\hspace*{-1pt},\hspace*{-2pt}\pi_t\hspace*{-1pt},\hspace*{-2pt}y_{t\hspace*{-1pt}+\hspace*{-1pt}1}\hspace*{-1pt},\hspace*{-2pt}x_{t\hspace*{-1pt}+\hspace*{-1pt}1}\hspace*{-1pt},\hspace*{-2pt}z_{t\hspace*{-1pt}+\hspace*{-1pt}1}\hspace*{-1pt},\hspace*{-2pt}a_t\hspace*{-1pt}\}\hspace*{-1pt}\mathbb{P}\{\hspace*{-1pt}y_{t\hspace*{-1pt}+\hspace*{-1pt}1}\hspace*{-1pt},\hspace*{-2pt}z_{t\hspace*{-1pt}+\hspace*{-1pt}1}\hspace*{-1pt}|x_{t\hspace*{-1pt}+\hspace*{-1pt}1}\hspace*{-1pt},\hspace*{-2pt}a_t\hspace*{-1pt}\}\mathbb{P}\{\hspace*{-1pt}x_{t\hspace*{-1pt}+\hspace*{-1pt}1}|x_t\hspace*{-1pt},\hspace*{-1pt}a_{t}\hspace*{-1pt}\}\sigma^{*-i}_{t}\hspace*{-1pt}(\hspace*{-1pt}\pi_t\hspace*{-1pt},\hspace*{-1pt}s_t^{-i}\hspace*{-1pt})(\hspace*{-1pt}a_t^{-i}\hspace*{-1pt})\hspace*{-1pt}\Big]\\
=\\
\mathbb{P}^{\sigma^{*-i}}_\psi\hspace*{-1pt}\{\hspace*{-1pt}x_{t\hspace*{-1pt}+\hspace*{-1pt}1}\hspace*{-1pt},\hspace*{-1pt}\pi_{t\hspace*{-1pt}+\hspace*{-1pt}1}\hspace*{-1pt},\hspace*{-1pt}s_{t\hspace*{-1pt}+\hspace*{-1pt}1}|x_{t}\hspace*{-1pt},\hspace*{-1pt}\pi_{t}\hspace*{-1pt},\hspace*{-1pt}s_{t}\hspace*{-1pt},\hspace*{-1pt}a_{t}^i\hspace*{-1pt}\}
=
\mathbb{P}^{\sigma^{*-i}}_\psi\hspace*{-1pt}\{\hspace*{-1pt}\tilde{x}_{t\hspace*{-1pt}+\hspace*{-1pt}1}|\tilde{x}_{t}\hspace*{-1pt},\hspace*{-1pt}a_{t}^i\hspace*{-1pt}\}.
\end{gather*}

Property (b) is true because
\begin{gather*}
\mathbb{P}^{\sigma^*}_\psi\hspace*{-1pt}\{\tilde{y}_{t}^i|\tilde{x}_{1:t}\hspace*{-1pt},\hspace*{-1pt}a_{1:t\hspace*{-1pt}-\hspace*{-1pt}1}^i\hspace*{-1pt},\hspace*{-1pt}\tilde{y}_{1:t\hspace*{-1pt}-\hspace*{-1pt}1}^i\hspace*{-1pt}\}
=
\mathbb{P}^{\sigma^*}_\psi\hspace{-2pt}\{\hspace{-1pt}y_{t}^i\hspace*{-1pt},\hspace*{-1pt}z_{t}\hspace*{-1pt}|x_{1:t}\hspace*{-1pt},\hspace*{-1pt}\pi_{1:t}\hspace*{-1pt},\hspace*{-1pt}s_{1:t}\hspace*{-1pt},\hspace*{-1pt}y_{1:t\hspace*{-1pt}-\hspace*{-1pt}1}^i\hspace*{-1pt},\hspace*{-1pt}z_{1:t\hspace*{-1pt}-\hspace*{-1pt}1}\hspace*{-1pt},\hspace*{-1pt}a_{1:t\hspace*{-1pt}-\hspace*{-1pt}1}^i\hspace{-1pt}\}\hspace{-2pt}\\
=\\
\hspace{-3pt}\sum_{a_{t-1}^{-i}}\hspace{-2pt}\mathbb{P}^{\sigma^*}_\psi\hspace{-2pt}\{\hspace{-1pt}y_{t}^i\hspace*{-1pt},\hspace*{-2pt}z_{t}\hspace*{-1pt},\hspace*{-2pt}a_{t-1}^{-i}\hspace*{-1pt}|x_{1\hspace*{-1pt}:t}\hspace*{-1pt},\hspace*{-2pt}\pi_{1\hspace*{-1pt}:t}\hspace*{-1pt},\hspace*{-2pt}s_{1\hspace*{-1pt}:t}\hspace*{-1pt},\hspace*{-2pt}y_{1\hspace*{-1pt}:t\hspace*{-1pt}-\hspace*{-1pt}1}^i\hspace*{-1pt},\hspace*{-2pt}z_{1\hspace*{-1pt}:t\hspace*{-1pt}-\hspace*{-1pt}1}\hspace*{-1pt},\hspace*{-2pt}a_{1\hspace*{-1pt}:t\hspace*{-1pt}-\hspace*{-1pt}1}^i\hspace{-1pt}\}\\
=\\
\hspace*{-3pt}\sum_{a_{t-1}^{-i}}\hspace{-3pt}\Big[\hspace*{-1pt}\mathbb{P}^{\sigma^*}_\psi\hspace{-2pt}\{\hspace{-1pt}y_{t}^i\hspace*{-1pt},\hspace*{-2pt}z_{t}\hspace*{-1pt}|x_{1\hspace*{-1pt}:t}\hspace*{-1pt},\hspace*{-2pt}\pi_{1\hspace*{-1pt}:t}\hspace*{-1pt},\hspace*{-2pt}s_{1\hspace*{-1pt}:t}\hspace*{-1pt},\hspace*{-2pt}y_{1\hspace*{-1pt}:t\hspace*{-1pt}-\hspace*{-1pt}1}^i\hspace*{-1pt},\hspace*{-2pt}z_{1\hspace*{-1pt}:t\hspace*{-1pt}-\hspace*{-1pt}1}\hspace*{-1pt},\hspace*{-2pt}a_{1\hspace*{-1pt}:t\hspace*{-1pt}-\hspace*{-1pt}1}^i\hspace*{-1pt},\hspace*{-2pt}a_{t\hspace*{-1pt}-\hspace*{-1pt}1}^{-i}\hspace{-1pt}\}\sigma^{*-i}_t\hspace*{-1pt}(\hspace*{-1pt}\pi_{t\hspace*{-1pt}-\hspace*{-1pt}1}\hspace*{-1pt},\hspace*{-2pt}s_{t\hspace*{-1pt}-\hspace*{-1pt}1}\hspace*{-1pt})(\hspace*{-1pt}a_{t\hspace*{-1pt}-\hspace*{-1pt}1}^{-i}\hspace*{-1pt})\hspace*{-1pt}\Big]\\
\stackrel{\text{by system dynamics (\ref{eq:systemdynamic2})}}{=}\\
\sum_{a_{t\hspace*{-1pt}-\hspace*{-1pt}1}^{-i}}\hspace{-1pt}\mathbb{P}\{\hspace{-1pt}y_{t}^i\hspace*{-1pt},\hspace*{-2pt}z_{t}\hspace*{-1pt}|x_{t}\hspace*{-1pt},\hspace*{-2pt}a_{t\hspace*{-1pt}-\hspace*{-1pt}1}^i\hspace*{-1pt},\hspace*{-2pt}a_{t\hspace*{-1pt}-\hspace*{-1pt}1}^{-i}\hspace{-1pt}\}\sigma^{*-i}_t\hspace*{-1pt}(\hspace*{-1pt}\pi_{t\hspace*{-1pt}-\hspace*{-1pt}1}\hspace*{-1pt},\hspace*{-2pt}s_{t\hspace*{-1pt}-\hspace*{-1pt}1}\hspace*{-1pt})(\hspace*{-1pt}a_{t\hspace*{-1pt}-\hspace*{-1pt}1}^{-i}\hspace*{-1pt})\\
=\\
\mathbb{P}^{\sigma^{*-i}}\hspace*{-1pt}\{\hspace*{-1pt}y_{t}^i\hspace*{-1pt},\hspace*{-2pt}z_{t}\hspace*{-1pt}|x_{t}\hspace*{-1pt},\hspace*{-1pt}\pi_{t\hspace*{-1pt}-\hspace*{-1pt}1}\hspace*{-1pt},\hspace*{-2pt}s_{t\hspace*{-1pt}-\hspace*{-1pt}1}\hspace*{-1pt},\hspace*{-2pt}a_{t\hspace*{-1pt}-\hspace*{-1pt}1}^i\hspace*{-1pt}\}
=
\mathbb{P}^{\sigma^*}_\psi\hspace*{-1pt}\{\tilde{y}_{t}^i\hspace*{-1pt}|\tilde{x}_{t}\hspace*{-1pt},\hspace*{-2pt}a_{t\hspace*{-1pt}-\hspace*{-1pt}1}^i\hspace*{-1pt}\}.
\end{gather*}	

Property (c) is true because
\begin{align*}
&\mathbb{E}^{\sigma^*}_\psi\hspace{-4pt}\left\{\hspace*{-1pt}u_t^i\hspace*{-1pt}(\hspace*{-1pt}X_t\hspace*{-1pt},\hspace*{-2pt}A_t\hspace*{-1pt})\hspace*{-1pt}|\tilde{x}_{1\hspace*{-1pt}:t}\hspace*{-1pt},\hspace*{-2pt}a_{1\hspace*{-1pt}:t}^i\hspace*{-1pt},\hspace*{-2pt}\tilde{y}_{1\hspace*{-1pt}:t}^i\hspace*{-1pt}\right\}\\
&=
\hspace*{-2pt}\mathbb{E}^{\sigma^*}_\psi\hspace{-4pt}\left\{\hspace*{-1pt}u_t^i\hspace*{-1pt}(\hspace*{-1pt}X_t\hspace*{-1pt},\hspace*{-2pt}A_t\hspace*{-1pt})\hspace*{-1pt}|x_t\hspace*{-1pt},\hspace*{-2pt}\pi_t\hspace*{-1pt},\hspace*{-2pt}s_t\hspace*{-1pt},\hspace*{-2pt}\tilde{x}_{1\hspace*{-1pt}:t-1}\hspace*{-1pt},\hspace*{-2pt}a_{1\hspace*{-1pt}:t}^i\hspace*{-1pt},\hspace*{-2pt}\tilde{y}_{1\hspace*{-1pt}:t}^i\hspace*{-1pt}\right\}\\
&=
\hspace*{-2pt}\mathbb{E}^{\sigma^{*-i}}_\psi\hspace{-4pt}\left\{\hspace*{-1pt}u_t^i\hspace*{-1pt}(\hspace*{-1pt}X_t\hspace*{-1pt},\hspace*{-2pt}(a_t^i\hspace*{-1pt},\hspace{-2pt}\sigma^{*-i}\hspace*{-1pt}(\hspace*{-1pt}\pi_t\hspace*{-1pt},\hspace{-2pt}s_t^{-i})\hspace*{-1pt})\hspace*{-1pt})\hspace*{-1pt}|x_t\hspace*{-1pt},\hspace{-2pt}\pi_t\hspace*{-1pt},\hspace{-2pt}s_t\hspace*{-1pt},\hspace{-2pt}\tilde{x}_{1\hspace*{-1pt}:t\hspace*{-1pt}-\hspace*{-1pt}1}\hspace*{-1pt},\hspace{-2pt}a_{1\hspace*{-1pt}:t}^i\hspace*{-1pt},\hspace{-2pt}\tilde{y}_{1\hspace*{-1pt}:t}^i\hspace{-1pt}\right\}\\
&=
u_t^i\hspace*{-1pt}(\hspace*{-1pt}x_t\hspace*{-1pt},\hspace*{-2pt}(\hspace*{-1pt}a_t^i\hspace*{-1pt},\hspace*{-2pt}\sigma^{*-i}\hspace*{-1pt}(\hspace*{-1pt}\pi_t\hspace*{-1pt},\hspace*{-2pt}s_t^{-i}\hspace*{-1pt})\hspace*{-1pt})\hspace*{-1pt})\hspace*{-2pt}:=\hspace*{-2pt}\tilde{u}_t(\hspace*{-1pt}\tilde{x}_t\hspace*{-1pt},\hspace*{-1pt}a_t^i\hspace*{-1pt})
\end{align*}

To establish step (ii), that is, to show that $\{\hspace*{-1pt}\Pi_t,\hspace*{-1pt}S_t^i\}$ is an information state for agent $i$ when he interacts with the superdynamic system defined above based on SIB assessment $(\sigma^*\hspace*{-1pt},\gamma^\psi)$, we need to prove: (1) it can be updated recursively at $t$, \textit{i.e.} it can be determined using $\{\Pi_{t-1}\hspace*{-1pt},\hspace*{-1pt}S_{t-1}^i\hspace*{-1pt}\}$ and $\{\hspace*{-1pt}\tilde{Y}_{t}^i\hspace*{-1pt},\hspace*{-1pt}A_t^i\hspace*{-1pt}\}\hspace*{-2pt}=\hspace*{-2pt}\{\hspace*{-1pt}Y_t^i\hspace*{-1pt},\hspace*{-1pt}Z_t,\hspace*{-1pt}A_t^i\hspace*{-1pt}\}$; (2) agent $i$'s belief about $\{\hspace*{-1pt}\Pi_{t+1}\hspace*{-1pt},\hspace*{-1pt}S_{t+1}^i\hspace*{-1pt}\}$ conditioned on $\{\hspace*{-1pt}\Pi_t,\hspace*{-1pt}S_t^i\hspace*{-1pt},\hspace*{-1pt}A_t^i\}$ is independent of  $H_t^i$ and his actual strategy $g^i$; and (3) it is sufficient to evaluate the agent $i$'s instantaneous utility at $t$ for every action $a_t^i\hspace*{-2pt}\in\hspace*{-2pt}\mathcal{A}_t^i$, for all $t\hspace*{-2pt}\in\hspace*{-2pt}\mathcal{T}$.

Condition (1) is satisfied since $\Pi_t\hspace*{-1pt}=\hspace*{-1pt}\psi_t(\Pi_{t-1},\hspace*{-1pt}Z_t)$ and $S_t^i\hspace*{-1pt}=\hspace*{-1pt}\phi_t(S_{t-1}^i,\hspace*{-1pt}\{Y_t^i,\hspace*{-1pt}Z_t,\hspace*{-1pt}A_t^i\})\hspace*{-1pt}$ for $t\hspace*{-1pt}\in\hspace*{-1pt}\mathcal{T}\backslash\{1\}$; see part (i) of Definition \ref{def:sufficient} and (\ref{eq:CIBupdaterule}).

To prove condition (2), let 
\begin{align}g^{*j}_t(h_t^j)=\sigma^{*j}_t(l(h_t^j),\gamma_t^\psi(c_t)), \label{eq:thm2-5}\end{align} for all $j\in\mathcal{N}$ and $t\in\mathcal{T}$. Then condition (2) is satisfied since
\begin{gather*}
\mathbb{P}^{\sigma^*}_\psi\hspace*{-2pt}\{\hspace*{-1pt}s_{t+1}^i\hspace*{-1pt},\hspace*{-2pt}\pi_{t\hspace*{-1pt}+\hspace*{-1pt}1}\hspace*{-1pt}|h_t^i\hspace*{-1pt},\hspace*{-2pt}a_t^i\hspace*{-1pt}\}
\hspace*{-2pt}=\hspace*{-7pt}
\sum_{h_t^{-i}\hspace*{-2pt},a_t^{-i}}\hspace*{-6pt}\mathbb{P}^{\sigma^{*}}_\psi\hspace*{-3pt}\{\hspace*{-1pt}s_{t\hspace*{-1pt}+\hspace*{-1pt}1}^i\hspace*{-1pt},\hspace*{-2pt}\pi_{t\hspace*{-1pt}+\hspace*{-1pt}1}\hspace*{-1pt},\hspace*{-2pt}h_t^{-i}\hspace*{-2pt},\hspace*{-2pt}a_{t}^{-i}\hspace*{-1pt}|h_t^i\hspace*{-1pt},\hspace*{-2pt}a_t^i\hspace*{-1pt}\}\\
\stackrel{\text{by Theorem \ref{thm:beliefindependence} and (\ref{eq:thm2-5})}}{=}\\
\sum_{h_t^{-i}\hspace*{-2pt},a_t^{-i}}\hspace*{-6pt}\mathbb{P}^{\sigma^*}_\psi\hspace*{-3pt}\{\hspace*{-1pt}s_{t\hspace*{-1pt}+\hspace*{-1pt}1}^i\hspace*{-1pt},\hspace*{-2pt}\pi_{t\hspace*{-1pt}+\hspace*{-1pt}1}\hspace*{-1pt}|h_t\hspace*{-1pt},\hspace*{-2pt}a_t\hspace*{-1pt}\}\mathbb{P}^{g^{*\hspace*{-1pt}-i}}\hspace*{-3pt}\{\hspace*{-1pt}h_t^{-i}\hspace*{-1pt}|h_t^i\hspace*{-1pt}\}g^{*-i}_t\hspace*{-1pt}(\hspace*{-1pt}h_t^{-i}\hspace*{-1pt})(\hspace*{-1pt}a_t^{-i}\hspace*{-1pt})\\
=\\
\sum_{h_t^{-i}\hspace*{-2pt},a_t^{-i}\hspace*{-2pt},s_{t\hspace*{-1pt}+\hspace*{-1pt}1}^{-i}}\hspace{-12pt}\mathbb{P}^{\sigma^*}_\psi\hspace*{-3pt}\{\hspace*{-1pt}s_{t\hspace*{-1pt}+\hspace*{-1pt}1}^i\hspace*{-1pt},\hspace*{-2pt}\pi_{t\hspace*{-1pt}+\hspace*{-1pt}1}\hspace*{-1pt},\hspace*{-2pt}s_{t\hspace*{-1pt}+\hspace*{-1pt}1}^{-i}\hspace*{-1pt}|h_t\hspace*{-1pt},\hspace*{-2pt}a_t\hspace*{-1pt}\}\hspace{-1pt}\mathbb{P}^{g^{\hspace*{-1pt}*-i}}\hspace*{-1pt}\{\hspace*{-1pt}h_t^{-i}\hspace*{-1pt}|h_t^i\hspace*{-1pt}\}g^{*-i}_t\hspace*{-1pt}(\hspace*{-1pt}h_t^{-i}\hspace*{-1pt})(\hspace*{-1pt}a_t^{-i}\hspace*{-1pt})\\
\stackrel{\text{by Lemma \ref{lemma:lem} and $s_t^{-i}\hspace*{-2pt}=\hspace*{-1pt}\zeta_t^{-i}\hspace*{-1pt}(h_t^{-i};g^*_{1:t-1})$ (see Definition \ref{def:sufficient})}}{=}\\
\sum_{h_t^{-i}\hspace*{-2pt},a_t^{-i}\hspace*{-2pt},s_{t\hspace*{-1pt}+\hspace*{-1pt}1}^{-i}}\hspace{-12pt}\mathbb{P}^{\sigma^*}_\psi\hspace*{-3pt}\{\hspace*{-1pt}s_{t\hspace*{-1pt}+\hspace*{-1pt}1}\hspace*{-1pt},\hspace*{-2pt}\pi_{t\hspace*{-1pt}+\hspace*{-1pt}1}\hspace*{-1pt}|s_t\hspace*{-1pt},\hspace*{-2pt}\pi_t\hspace*{-1pt},\hspace*{-2pt}a_t\hspace*{-1pt}\}\hspace{-1pt}\mathbb{P}^{g^{\hspace*{-1pt}*-i}}\hspace*{-1pt}\{\hspace*{-1pt}s_t^{-i}\hspace*{-1pt}|h_t^i\hspace*{-1pt}\}g^{\hspace*{-1pt}*-i}_t\hspace*{-1pt}(\hspace*{-1pt}h_t^{-i}\hspace*{-1pt})(\hspace*{-1pt}a_t^{-i}\hspace*{-1pt})\\
\stackrel{\text{by part (ii) of Lemma \ref{lemma:equivalence}}}{=}\\
\sum_{s_t^{-i}\hspace*{-2pt},a_t^{-i}\hspace*{-2pt},s_{t\hspace*{-1pt}+\hspace*{-1pt}1}^{-i}}\hspace{-12pt}\mathbb{P}^{\sigma^*}_\psi\hspace*{-3pt}\{\hspace*{-1pt}s_{t\hspace*{-1pt}+\hspace*{-1pt}1}\hspace*{-1pt},\hspace*{-2pt}\pi_{t\hspace*{-1pt}+\hspace*{-1pt}1}\hspace*{-1pt}|s_t\hspace*{-1pt},\hspace*{-2pt}\pi_t\hspace*{-1pt},\hspace*{-2pt}a_t\hspace*{-1pt}\}\hspace{-1pt}\mathbb{P}\{\hspace*{-1pt}s_t^{-i}\hspace*{-1pt}|s_t^i\hspace*{-1pt},\hspace*{-2pt}\pi_t\}g^{\hspace*{-1pt}*-i}_t\hspace*{-1pt}(\hspace*{-1pt}h_t^{-i}\hspace*{-1pt})(\hspace*{-1pt}a_t^{-i}\hspace*{-1pt})\\
\stackrel{\text{by (\ref{eq:thm2-5})}}{=}\\
\sum_{h_t^{-i}\hspace*{-2pt},a_t^{-i}\hspace*{-2pt},s_{t\hspace*{-1pt}+\hspace*{-1pt}1}^{-i}}\hspace{-12pt}\mathbb{P}^{\sigma^*}_\psi\hspace*{-3pt}\{\hspace*{-1pt}s_{t\hspace*{-1pt}+\hspace*{-1pt}1}^i\hspace*{-1pt},\hspace*{-2pt}\pi_{t\hspace*{-1pt}+\hspace*{-1pt}1}\hspace*{-1pt},\hspace*{-2pt}s_{t\hspace*{-1pt}+\hspace*{-1pt}1}^{-i}\hspace*{-1pt}|h_t\hspace*{-1pt},\hspace*{-2pt}a_t\hspace*{-1pt}\}\hspace{-1pt}\mathbb{P}\{\hspace*{-1pt}s_t^{-i}\hspace*{-1pt}|s_t^i,\hspace*{-2pt}\pi_t\hspace*{-1pt}\}\sigma^{\hspace*{-1pt}*-i}_t\hspace*{-1pt}(\hspace*{-1pt}\pi_t\hspace*{-1pt},\hspace*{-2pt}s_t^{-i}\hspace*{-1pt})(\hspace*{-1pt}a_t^{-i}\hspace*{-1pt})\\
=\\
\mathbb{P}^{\sigma^{*}}_\psi\hspace*{-3pt}\{\hspace*{-1pt}s_{t+1}^i\hspace*{-1pt},\hspace*{-1pt}\pi_{t+1}\hspace*{-1pt}|s_t^i\hspace*{-1pt},\hspace*{-1pt}\pi_t\hspace*{-1pt},\hspace*{-1pt}a_t^i\hspace*{-1pt}\}.
\end{gather*}.

To prove condition (3), we need to show that for all $a_t^i\hspace*{-1pt}\in\hspace*{-1pt}\mathcal{A}_t^i$,
\begin{align}
\mathbb{E}^{g^{*-i}}\hspace*{-2pt}\{\hspace*{-1pt}u_i^t(\hspace*{-1pt}X_t\hspace*{-1pt},\hspace*{-2pt}A_t^{-i}\hspace*{-2pt},\hspace*{-2pt}a_t^i)\hspace*{-1pt}|h_t^i\hspace*{-1pt}\}\hspace*{-2pt}=\hspace*{-2pt}\mathbb{E}^{g^{*-i}}\hspace*{-2pt}\{\hspace*{-1pt}u_i^t(\hspace*{-1pt}X_t\hspace*{-1pt},\hspace*{-2pt}A_t^{-i}\hspace*{-2pt},\hspace*{-2pt}a_t^i)\hspace*{-1pt}|\hspace*{-1pt}\pi_t\hspace*{-1pt},\hspace*{-2pt}s_t^i\},
\end{align}
for all $h_t^i,\pi_t,s_t^i$, $t\hspace*{-1pt}\in\hspace*{-1pt}\mathcal{T}$.

By Lemma \ref{lemma:equivalence},
\begin{align}
\mathbb{P}^{g^{*-i}}\{s_t^{-i}|h_t^i\}=\mathbb{P}\{s_t^{-i}|\pi_t,s_t^i\}. \label{eq:proof-thm2-1}
\end{align}
Setting $\pi_t\hspace*{-1pt}=\hspace*{-1pt}\gamma^\psi(c_t)$, and $A_t^{-i}\hspace*{-1pt}=\hspace*{-1pt}\sigma^{*-i}(\pi_t,S_t^{-i})$,
\begin{align}
\mathbb{E}^{\sigma^*}_\psi\hspace*{-2pt}\{\hspace*{-1pt}u_t^i\hspace*{-1pt}(\hspace*{-1pt}X_t\hspace*{-1pt},\hspace*{-2pt}A_t^{-i}\hspace*{-1pt},\hspace*{-2pt}a_t^i\hspace*{-1pt})\hspace*{-1pt}|h_t^i\hspace*{-1pt}\}&\hspace*{-2pt}
=\hspace*{-2pt}
\mathbb{E}^{\sigma^*}_\psi\hspace*{-2pt}\{\hspace*{-1pt}u_t^i\hspace*{-1pt}(\hspace*{-1pt}X_t\hspace*{-1pt},\hspace*{-2pt}\sigma_t^{*-i}\hspace*{-1pt}(\hspace*{-1pt}\pi_t\hspace*{-1pt},\hspace*{-2pt}S_t^{-i}\hspace*{-1pt})\hspace*{-1pt},\hspace*{-2pt}a_t^i\hspace*{-1pt})\hspace*{-1pt}|h_t^i\hspace*{-1pt}\}\nonumber\\
&\hspace*{-75pt}=\hspace*{-1pt}
\mathbb{E}^{\sigma^*}_\psi\hspace*{-3pt}\left\{\hspace*{-1pt}\mathbb{E}^{\sigma^{*-i}}\hspace*{-2pt}\{\hspace*{-1pt}u_t^i\hspace*{-1pt}(\hspace*{-1pt}X_t\hspace*{-1pt},\hspace*{-2pt}\sigma_t^{*-i}\hspace*{-1pt}(\hspace*{-1pt}\pi_t\hspace*{-1pt},\hspace*{-2pt}S_t^{-i}\hspace*{-1pt})\hspace*{-1pt},\hspace*{-2pt}a_t^i\hspace*{-1pt})\hspace*{-1pt}|S_t^{-i}\hspace*{-2pt},\hspace*{-2pt}\pi_t\hspace*{-1pt},\hspace*{-2pt}s_i^t\hspace*{-1pt},\hspace*{-2pt}h_t^i\hspace*{-1pt}\}\hspace*{-1pt}\Big|h_t^i\hspace*{-1pt}\right\}\nonumber\\
&\hspace*{-75pt}=\hspace*{-1pt}
\mathbb{E}^{\sigma^*}_\psi\hspace*{-3pt}\left\{\hspace*{-1pt}\mathbb{E}^{\sigma^{*-i}}\hspace*{-2pt}\{\hspace*{-1pt}u_t^i\hspace*{-1pt}(\hspace*{-1pt}X_t\hspace*{-1pt},\hspace*{-2pt}\sigma_t^{*-i}\hspace*{-1pt}(\hspace*{-1pt}\pi_t\hspace*{-1pt},\hspace*{-2pt}S_t^{-i}\hspace*{-1pt})\hspace*{-1pt},\hspace*{-2pt}a_t^i\hspace*{-1pt})\hspace*{-1pt}|S_t^{-i}\hspace*{-2pt},\hspace*{-2pt}\pi_t\hspace*{-1pt},\hspace*{-2pt}s_t^i\hspace*{-1pt},\hspace*{-2pt}c_t\hspace*{-1pt}\}\hspace*{-1pt}\Big|h_t^i\hspace*{-1pt}\right\}\nonumber\\
&\hspace*{-75pt}=\hspace*{-1pt}
\mathbb{E}^{\sigma^{*-i}}_\psi\hspace*{-3pt}\left\{\hspace*{-1pt}\mathbb{E}^{\sigma^{*-i}}\hspace*{-2pt}\{\hspace*{-1pt}u_t^i\hspace*{-1pt}(\hspace*{-1pt}X_t\hspace*{-1pt},\hspace*{-2pt}\sigma_t^{*-i}\hspace*{-1pt}(\hspace*{-1pt}\pi_t\hspace*{-1pt},\hspace*{-2pt}S_t^{-i}\hspace*{-1pt})\hspace*{-1pt},\hspace*{-2pt}a_t^i\hspace*{-1pt})\hspace*{-1pt}|S_t^{-i}\hspace*{-2pt},\hspace*{-2pt}\pi_t\hspace*{-1pt},\hspace*{-2pt}s_t^i\hspace*{-1pt},\hspace*{-2pt}c_t\hspace*{-1pt}\}\hspace*{-1pt}\Big|h_t^i\hspace*{-1pt}\right\}\nonumber\\
&\hspace*{-75pt}=\hspace*{-1pt}
\mathbb{E}^{\sigma^{*-i}}_\psi\hspace*{-3pt}\left\{\hspace*{-1pt}\mathbb{E}^{\sigma^{*-i}}\hspace*{-2pt}\{\hspace*{-1pt}u_t^i\hspace*{-1pt}(\hspace*{-1pt}X_t\hspace*{-1pt},\hspace*{-2pt}\sigma_t^{*-i}\hspace*{-1pt}(\hspace*{-1pt}\pi_t\hspace*{-1pt},\hspace*{-1pt}S_t^{-i}\hspace*{-1pt})\hspace*{-1pt},\hspace*{-2pt}a_t^i)\hspace*{-1pt}|S_t^{-i}\hspace*{-2pt},\hspace*{-2pt}\pi_t\hspace*{-1pt},\hspace*{-2pt}s_t^i\hspace*{-1pt}\}\hspace*{-1pt}\Big|h_t^i\hspace*{-1pt}\right\}\nonumber\\
&\hspace*{-75pt}=\hspace*{-1pt}
\mathbb{E}^{\sigma^{*-i}}_\psi\hspace*{-3pt}\{\hspace*{-1pt}u_t^i\hspace*{-1pt}(\hspace*{-1pt}X_t\hspace*{-1pt},\hspace*{-2pt}\sigma_t^{*-i}\hspace*{-1pt}(\hspace*{-1pt}\pi_t\hspace*{-1pt},\hspace*{-2pt}S_t^{-i}\hspace*{-1pt})\hspace*{-1pt},\hspace*{-2pt}a_t^i)\hspace*{-1pt}|\pi_t\hspace*{-1pt},\hspace*{-2pt}s_t^i\hspace*{-1pt}\}.\label{eq:proof-thm2-2}
\end{align}

The first equality above is by substituting $A_t^{-i}\hspace*{-2pt}=\hspace*{-2pt}\sigma_t^{*-i}\hspace*{-1pt}(\hspace*{-1pt}\pi_t\hspace*{-1pt},\hspace*{-1pt}S_t^{-i}\hspace*{-1pt})$. The second equality follows from the smoothing property of conditional expectation. The third equality holds by condition (iii) of Definition \ref{def:sufficient}. The fourth equality follows from Theorem \ref{thm:beliefindependence} since $S_t^{-i}$ is a function of $H_t^{-i}$. The fifth equality holds since for every $x_t,s_t,\pi_t,c_t$,
\begin{align*}
\mathbb{P}\{\hspace*{-1pt}x_t\hspace*{-1pt}|s_t\hspace*{-1pt},\hspace*{-2pt}\pi_t\hspace*{-1pt},\hspace*{-2pt}c_t\hspace*{-1pt}\}\hspace*{-2pt}=\hspace*{-1pt}\frac{\mathbb{P}\{\hspace*{-1pt}x_t\hspace*{-1pt},\hspace*{-2pt}s_t\hspace*{-1pt}|\pi_t\hspace*{-1pt},\hspace*{-2pt}c_t\hspace*{-1pt}\}}{\mathbb{P}\{\hspace*{-1pt}s_t\hspace*{-1pt},\hspace*{-2pt}\pi_t\hspace*{-1pt},\hspace*{-2pt}c_t\hspace*{-1pt}\}}\hspace*{-2pt}=\hspace*{-1pt}\frac{\pi_t\hspace*{-1pt}(\hspace*{-1pt}x_t\hspace*{-1pt},\hspace*{-2pt}s_t\hspace*{-1pt})}{\sum_{\hat{x}_t}\hspace*{-2pt}\pi_t\hspace*{-1pt}(\hspace*{-1pt}\hat{x}_t\hspace*{-1pt},\hspace*{-2pt}s_t\hspace*{-1pt})}=\mathbb{P}\{\hspace*{-1pt}x_t\hspace*{-1pt}|s_t\hspace*{-1pt},\hspace*{-2pt}\pi_t\hspace*{-1pt}\}.
\end{align*}  
The last equality is true by (\ref{eq:proof-thm2-1}). By (\ref{eq:proof-thm2-2}) we prove condition (3) for $\{\hspace*{-1pt}\Pi_t\hspace*{-1pt},\hspace*{-2pt}S_t^i\hspace*{-1pt}\}$ to be an information state, and thus establish the result of Theorem \ref{thm:closeness}.

\end{proof} 

\begin{proof}[Proof of Theorem \ref{thm:decomposition}]
	Let $(\sigma^*\hspace*{-2pt},\hspace*{-1pt}\psi)$ denote a solution of the dynamic program. We note that the SIB update rule $\psi$ is consistent with $\sigma^*$ by requirement (\ref{eq:decomposition-update}). Therefore, we only need to show that the SIB assessment $(\sigma^*\hspace*{-2pt},\hspace*{-1pt}\psi)$ is sequentially rational. To prove it, we use the \textit{one-shot deviation principle} for dynamic games with asymmetric information \cite{hendon1996one}. To state the one-shot deviation, we need the following definitions. 
	
	\begin{definition}[One-shot deviation]
		We say $\tilde{g}^i$ is a one-shot deviation from $g^{*i}$ if there exists a unique $h_t^i\hspace*{-2pt}\in\hspace*{-1pt}\mathcal{H}^i$ such that $\tilde{g}_t^i\hspace*{-1pt}(\hspace*{-1pt}h_t^i\hspace*{-1pt})\hspace*{-2pt}\neq\hspace*{-1pt} g^{*i}_t\hspace*{-1pt}(\hspace*{-1pt}h_t^i\hspace*{-1pt})$, and $\tilde{g}_\tau^i\hspace*{-1pt}(\hspace*{-1pt}h_\tau^i\hspace*{-1pt})\hspace*{-2pt}\neq \hspace*{-2pt}g^{*i}_\tau\hspace*{-1pt}(\hspace*{-1pt}h_\tau^i\hspace*{-1pt})$ for all $h_\tau^i\hspace*{-2pt}\neq \hspace*{-2pt}h_t^i$, $h_\tau^i\hspace*{-2pt}\in\hspace*{-1pt}\mathcal{H}^i$.
	\end{definition}
	
	\begin{definition}[Profitable one-shot deviation]
	Consider an assessment $(g^*\hspace*{-2pt},\hspace*{-1pt}\mu)$. We say $\tilde{g}^i$ is a profitable one-shot deviation for agent $i$ if $\tilde{g}^i$ is a one-shot deviation from $g^{*i}$ at $h_t^i$ such that $\tilde{g}_t^i\hspace*{-1pt}(h_t^i)\hspace*{-2pt}\neq\hspace*{-2pt} g^{*i}_t\hspace*{-1pt}(h_t^i)$, and 
	\begin{align*}
	\mathbb{E}_\mu^{(g^{*-i}\hspace*{-2pt},\tilde{g}^i)}\hspace*{-3pt}\left\{\hspace*{-2pt}\sum_{\tau=t}^{T}\hspace*{-1pt}u_\tau^i\hspace*{-1pt}(\hspace*{-1pt}X_\tau\hspace*{-1pt},\hspace*{-2pt}A_\tau\hspace*{-1pt})\hspace*{-1pt}\Big|h_t^i\hspace*{-1pt}\right\}\hspace*{-2pt}>\hspace*{-2pt}\mathbb{E}_\mu^{(g^{*-i}\hspace*{-2pt},g^{*i})}\hspace*{-3pt}\left\{\hspace*{-1pt}\sum_{\tau=t}^{T}\hspace*{-1pt}u_\tau^i\hspace*{-1pt}(\hspace*{-1pt}X_\tau\hspace*{-1pt},\hspace*{-2pt}A_\tau\hspace*{-1pt})\hspace*{-1pt}\Big|h_t^i\hspace*{-1pt}\right\}
	\end{align*}	
	\end{definition}
	
	\textbf{\textit{One-shot deviation principle \cite{hendon1996one}:}} A consistent assessment $(g^*\hspace*{-2pt},\hspace*{-1pt}\mu)$ is a PBE if and only if there exists no agent that has a profitable one-shot deviation.
	
	Below, we show that the consistent SIB assessment $(\sigma^*\hspace*{-2pt},\hspace*{-1pt}\psi)$ satisfies the sequential rationality condition using the one-shot deviation principle.
	
	Consider an arbitrary agent $i\hspace*{-2pt}\in\hspace*{-1pt}\mathcal{N}$, time $t\hspace*{-2pt}\in\hspace*{-1pt}\mathcal{T}\hspace*{-1pt}$, and history realization $h_t^i\hspace*{-2pt}\in\hspace*{-1pt}\mathcal{H}_t^i$. Agent $i$ has a profitable one-shot deviation at $h_t^i$ only if 
	\begin{align*}
	\sigma^{*i}_t\hspace*{-1pt}(\hspace*{-1pt}\pi_t\hspace*{-1pt},\hspace*{-2pt}s_t^i)\hspace*{-2pt}\notin \hspace*{-3pt}\argmax_{\tilde{g}_t^i\hspace*{-1pt}(h_t^i)\in\Delta(\hspace*{-1pt}\mathcal{A}_t^i)} \hspace*{-3pt}\mathbb{E}^{\sigma^*}_\pi\hspace*{-3pt}\left\{\hspace*{-1pt}\bar{U}_t^i\hspace*{-1pt}(\hspace*{-1pt}(\sigma^{*}\hspace*{-1pt}(\hspace*{-1pt}\pi_t\hspace*{-1pt},\hspace*{-2pt}S_t)\hspace*{-1pt})\hspace*{-1pt},\hspace*{-2pt}S_t\hspace*{-1pt},\hspace*{-2pt}\pi_t\hspace*{-1pt},\hspace*{-2pt}V_{t\hspace*{-1pt}+\hspace*{-1pt}1}\hspace*{-1pt},\hspace*{-2pt}\psi_{t\hspace*{-1pt}+\hspace*{-1pt}1}\hspace*{-1pt})\hspace*{-1pt}\Big|h_t^i\hspace*{-1pt}\right\}\hspace*{-1pt}.
	\end{align*}
	 
	 Given $(\pi_t\hspace*{-1pt},\hspace*{-1pt}V_{t+1}\hspace*{-1pt},\hspace*{-1pt}\psi_{t+1}\hspace*{-1pt},\hspace*{-1pt}\sigma^*_t)$, the expected value of the function $\bar{U}_t^i$ conditioned on $h_t^i$ is only a function of $s_t^i$, agent $i$'s belief about $S_t^{-i}$, as well as agent $i$'s strategy $\tilde{g}^i_t(h_t^i)$. Agent $i$'s belief about $S_t^{-i}$ given $h_t^i$ is only a function of $s_t^i$ and $\pi_t$ (see (\ref{eq:lemma1})). Therefore, any solution to the maximization problem above can be written as a function of $\pi_t$ and $s_t^i$, that is, it is a SIB strategy $\tilde{\sigma}^i_t(\pi_t,\hspace*{-1pt}s_t^i)$ for agent $i$. Consequently, agent $i$ has a profitable one-shot deviation only if
	 	\begin{align*}
	 	\sigma^{\hspace*{-1pt}*i}_t\hspace*{-1pt}(\hspace*{-1pt}\pi_t\hspace*{-1pt},\hspace*{-2pt}s_t^i\hspace*{-1pt})\hspace*{-2pt}\notin\hspace*{-7pt} \argmax_{\tilde{\sigma}_t^i\hspace*{-1pt}(\hspace*{-1pt}\pi_t,s_t^i)\in\Delta(\hspace*{-1pt}\mathcal{A}_t^i)} \hspace*{-8pt}\mathbb{E}^{\sigma^{\hspace*{-1pt}*}}_\pi\hspace*{-4pt}\left\{\hspace*{-2pt}\bar{U}_t^i\hspace*{-1pt}(\hspace*{-1pt}(\hspace*{-1pt}\sigma^{*}\hspace*{-1pt}(\hspace*{-1pt}\pi_t\hspace*{-1pt},\hspace*{-2pt}S_t\hspace*{-1pt})\hspace*{-1pt})\hspace*{-1pt},\hspace*{-2pt}S_t\hspace*{-1pt},\hspace*{-2pt}\pi_t\hspace*{-1pt},\hspace*{-2pt}V_{t\hspace*{-1pt}+\hspace*{-1pt}1}\hspace*{-1pt},\hspace*{-2pt}\psi_{t\hspace*{-1pt}+\hspace*{-1pt}1}\hspace*{-1pt})\hspace*{-1pt}\Big|\hspace*{-1pt}\pi_t\hspace*{-1pt},\hspace*{-2pt}s_t^i\hspace*{-2pt}\right\}\hspace*{-2pt}.
	 	\end{align*}
	 
	 By (\ref{eq:decomposition-BNE}), $\sigma^{*}_t$ is BNE of the stage game $\mathbf{G}_t(\pi_t,V_{t+1},\psi_{t+1})$ at $t$ (see also (\ref{eq:BNEcorrespondence})), \textit{i.e.}
	 	\begin{align*}
	 	\sigma^{\hspace*{-1pt}*i}_t\hspace*{-1pt}(\hspace*{-1pt}\pi_t\hspace*{-1pt},\hspace*{-2pt}s_t^i\hspace*{-1pt})\hspace*{-2pt}\in\hspace*{-7pt} \argmax_{\tilde{\sigma}_t^i\hspace*{-1pt}(\hspace*{-1pt}\pi_t,s_t^i)\in\Delta(\hspace*{-1pt}\mathcal{A}_t^i)} \hspace*{-8pt}\mathbb{E}^{\sigma^{\hspace*{-1pt}*}}_\pi\hspace*{-4pt}\left\{\hspace*{-2pt}\bar{U}_t^i\hspace*{-1pt}(\hspace*{-1pt}(\hspace*{-1pt}\sigma^{*}\hspace*{-1pt}(\hspace*{-1pt}\pi_t\hspace*{-1pt},\hspace*{-2pt}S_t\hspace*{-1pt})\hspace*{-1pt})\hspace*{-1pt},\hspace*{-2pt}S_t\hspace*{-1pt},\hspace*{-2pt}\pi_t\hspace*{-1pt},\hspace*{-2pt}V_{t\hspace*{-1pt}+\hspace*{-1pt}1}\hspace*{-1pt},\hspace*{-2pt}\psi_{t\hspace*{-1pt}+\hspace*{-1pt}1}\hspace*{-1pt})\hspace*{-1pt}\Big|\hspace*{-1pt}\pi_t\hspace*{-1pt},\hspace*{-2pt}s_t^i\hspace*{-2pt}\right\}\hspace*{-2pt}.
	 	\end{align*}
	 	Consequently, there exists no profitable deviation from $\sigma^{*i}_t\hspace*{-1pt}(\hspace*{-1pt}\pi_t,\hspace*{-1pt}s_t^i)$ at $h_t^i$. Therefore, there exists no agent that has a profitable one-shot deviation. Hence, by one-shot deviation principle, the consistent SIB assessment $(\sigma^*\hspace*{-2pt},\hspace*{-1pt}\psi)$ is sequentially rational, and thus, it is a SIB-PBE. 
\end{proof}

\begin{proof}[Proof of Lemma \ref{lemma:sufficient1}]
We prove below that if $V_{t+1}(\cdot,\hspace*{-1pt}s_{t+1})$ is continuous in $\pi_{t+1}$, then the dynamic program has a solution at stage $t$, $t\hspace*{-1pt}\in\hspace*{-1pt}\mathcal{T}$; that is, there exists at least one $\sigma^*_t$ such that $\sigma^*_t\hspace*{-1pt}\in \hspace*{-1pt}\text{\textbf{BNE}}_t(V_{t+1},\hspace*{-1pt}\psi_{t+1})$, where $\psi_{t+1}$ is consistent with $\sigma^*_t$.

For every $\pi_t$, define a perturbation of the stage game $\mathbf{G}_t(\hspace*{-1pt}\pi_t,\hspace*{-1pt}V_{t+1},\hspace*{-1pt}\psi_{t+1})$ by restricting the set of strategies of each agent to mixed strategies that assign probability of at least $\epsilon\hspace*{-1pt}>\hspace*{-1pt}0$ to every action $a_t^i\hspace*{-1pt}\in\hspace*{-1pt}\mathcal{A}_t^i$ of agent $i\hspace*{-1pt}\in\hspace*{-1pt}\mathcal{N}$; for every agent $i\hspace*{-1pt}\in\hspace*{-1pt}\mathcal{N}$
we denote this class of $\epsilon$-restricted strategies by $\Sigma_t^{i,\epsilon}$ and $\Sigma_t^{\epsilon}\hspace*{-1pt}:=\hspace*{-1pt}\Sigma_t^{1,\epsilon}\hspace*{-2pt}\times\hspace*{-1pt}\ldots\hspace*{-1pt}\times\hspace*{-1pt}\Sigma_t^{N,\epsilon}\hspace*{-1pt}$.
In the following we prove that, for every $\epsilon\hspace*{-1pt}>\hspace*{-1pt}0$, the corresponding perturbed stage game has an equilibrium $\sigma^{*,\epsilon}_t$ along with a consistent update rule $\psi_{t+1}^\epsilon$.

 We note that when the agents' equilibrium strategies are perfectly mixed strategies, then the update rule $\psi_{t+1}^\epsilon$ is completely determined via Bayes' rule. Therefore, for every strategy profile $\sigma^{*,\epsilon}_t\hspace*{-2pt}\in\hspace*{-1pt}\Sigma_t^\epsilon$ we can write $\psi_{t+1}^\epsilon\hspace*{-2pt}:=\hspace*{-1pt}\beta_{t+1}\hspace*{-1pt}(\sigma^{*,\epsilon})$, where  $\beta_{t+1}\hspace*{-1pt}(\sigma^{*,\epsilon})$ is Bayes' rule where $\sigma^{*,\epsilon}$ is utilized (see (\ref{eq:CIBconsistency-on}))

For every agent $i\hspace*{-1pt}\in\hspace*{-1pt}\mathcal{N}$, define a best response correspondence $\text{\textbf{BR}}_t^{i,\epsilon}\hspace*{-2pt}:\hspace*{-1pt}\Sigma_t^\epsilon\hspace*{-1pt}\rightrightarrows \hspace*{-1pt}\Sigma_t^{i,\epsilon}$ as
\begin{align}
&\hspace*{-4pt}\text{\textbf{BR}}_t^{i,\epsilon}\hspace*{-1pt}(\sigma^{*,\epsilon}_t)\hspace{-2pt}:=\hspace{-2pt}
\Bigg\{\hspace{-2pt}\sigma_t^i\hspace{-2pt}\in\hspace{-2pt}\Sigma_t^{i,\epsilon}\hspace{-3pt}:\hspace{-2pt}\sigma_t^i\hspace*{-1pt}(\hspace*{-1pt}\pi_t\hspace*{-1pt},\hspace{-2pt}s_t^i\hspace*{-1pt})\hspace{-2pt}\in\hspace{-2pt}\argmax_{\sigma_t^{i,\epsilon}\in\Sigma_t^{i,\epsilon}}\hspace*{-1pt}\mathbb{E}^{\sigma^{*\hspace*{-1pt}-\hspace*{-1pt}i\hspace*{-1pt},\hspace*{-1pt}\epsilon}_t\hspace{-2pt},\sigma_t^{\hspace*{-1pt}i\hspace*{-1pt},\hspace*{-0.5pt}\epsilon}}\hspace*{-3pt}\{\hspace*{-1pt}\bar{U}_t^i\hspace*{-1pt}(\hspace{-1pt}A_t\hspace*{-1pt},\hspace{-2pt}S_t\hspace*{-1pt},\hspace{-2pt}\pi_t\hspace*{-1pt},\hspace{-3pt}V_{\hspace*{-1pt}t\hspace*{-0.5pt}+\hspace*{-0.5pt}1}\hspace*{-1pt},\hspace{-2pt}\beta_{t\hspace*{-0.5pt}+\hspace*{-0.5pt}1}\hspace*{-1pt}(\hspace*{-1pt}\sigma^{*\hspace*{-1pt},\hspace*{-1pt}\epsilon}_t\hspace*{-1pt})\hspace{-1pt})\hspace*{-1pt}|\hspace*{-1pt}s_t^i\hspace*{-1pt},\hspace{-2pt}\pi_t\hspace{-1pt}\}\hspace*{-1pt},\hspace{-2pt}\forall\hspace*{-1pt} \pi_t,\hspace{-2pt}s_t^i\hspace{-3pt}\Bigg\}\hspace{-1pt},\hspace*{-6pt}
\end{align} 
which determines the set of all agent $i$'s best responses within the class of $\epsilon$-restricted strategies assuming that agents $-i$ are playing $\sigma^{*-i,\epsilon}_t$ and the update rule $\psi_{t+1}^\epsilon\hspace*{-2pt}=\hspace*{-1pt}\beta_{t+1}\hspace*{-1pt}(\sigma^{*,\epsilon}_t)$.  

For every $i\hspace*{-2pt}\in\hspace*{-2pt}\mathcal{N}$ and $\sigma^{\hspace*{-1pt}*,\epsilon}\hspace*{-3pt}\in\hspace*{-2pt}\Sigma_t^\epsilon$, we prove below that $\text{\textbf{BR}}_t^{i,\epsilon}\hspace*{-1pt}(\hspace*{-1pt}\sigma^{*,\epsilon}_t)$ is non-empty, convex, closed, and upper hemicontinuous.
 
We note that 
\begin{align*}
&\mathbb{E}^{\sigma^{*,\epsilon}_t\hspace*{-2pt},\sigma_t^{i,\epsilon}}\hspace*{-1pt}\{\hspace*{-1pt}\bar{U}_t^i\hspace*{-1pt}(\hspace*{-1pt}A_t\hspace*{-1pt},\hspace*{-2pt}S_t\hspace*{-1pt},\hspace*{-2pt}\pi_t\hspace*{-1pt},\hspace*{-2pt}V_{t+1}\hspace*{-1pt},\hspace*{-2pt}\beta_{t+1}\hspace*{-1pt}(\hspace*{-1pt}\sigma^{*,\epsilon}_t\hspace*{-1pt})\hspace*{-1pt})\hspace*{-1pt}|s_t^i\hspace*{-1pt},\hspace*{-2pt}\pi_t\hspace*{-1pt}\}\\
\hspace*{-2pt}=&\sum_{a_t^i}\hspace*{-1pt}\mathbb{E}^{\sigma^{*\hspace*{-1pt}-\hspace*{-1pt}i\hspace*{-1pt},\hspace*{-1pt}\epsilon\hspace*{-1pt}}_t\hspace*{-1pt},\hspace*{-0.5pt}A_t^i\hspace*{-0.5pt}=a_t^i}\hspace*{-1pt}\{\hspace*{-1pt}\bar{U}_t^i\hspace*{-1pt}(\hspace*{-1pt}A_t\hspace*{-1pt},\hspace*{-2pt}S_t\hspace*{-1pt},\hspace*{-2pt}\pi_t\hspace*{-1pt},\hspace*{-3pt}V_{\hspace*{-1pt}t\hspace*{-0.5pt}+\hspace*{-0.5pt}1}\hspace*{-1pt},\hspace*{-2pt}\beta_{t\hspace*{-0.5pt}+\hspace*{-0.5pt}1}\hspace*{-1.5pt}(\hspace*{-1pt}\sigma^{*\hspace*{-1pt},\hspace*{-1pt}\epsilon}_t\hspace*{-1pt})\hspace*{-1.5pt})\hspace*{-1pt}|\hspace*{-1pt}s_t^i\hspace*{-1pt},\hspace*{-2pt}\pi_t\hspace*{-1pt}\}\sigma_t^{\hspace*{-1pt}i\hspace*{-1pt},\hspace*{-0.5pt}\epsilon}\hspace*{-1.5pt}(\hspace*{-1.5pt}\pi_t\hspace*{-1pt},\hspace*{-2pt}s_t^i\hspace*{-1pt})\hspace*{-1pt}(\hspace*{-1.5pt}a_t^i\hspace*{-1pt})\\
\hspace*{-2pt}=&\sum_{a_t^i}\tilde{U}_{\pi_t\hspace*{-1pt},s_t^i}^{*\hspace*{-1pt},\epsilon\hspace*{-1pt}}\hspace*{-1pt}(\hspace*{-1pt}a_t^i\hspace*{-1pt})\sigma_t^{\hspace*{-1pt}i\hspace*{-0.5pt},\epsilon}(\hspace*{-1pt}\pi_t\hspace*{-1pt},\hspace*{-2pt}s_t^i\hspace*{-1pt})\hspace*{-1pt}(a_t^i),
\end{align*}
where
\begin{align*} \tilde{U}_{\pi_t\hspace*{-1pt},s_t^i}^{\sigma^{*\hspace*{-1pt},\epsilon\hspace*{-1pt}}}\hspace*{-1pt}(\hspace*{-1pt}a_t^i\hspace*{-1pt})\hspace*{-2pt}:=\hspace*{-2pt}\mathbb{E}^{\sigma^{{*\hspace*{-1pt}-i\hspace*{-1pt},\epsilon\hspace*{-1pt}}}_t\hspace*{-1pt},A_t^i=a_t^i}\hspace*{-1pt}\{\hspace*{-1pt}\bar{U}_t^i\hspace*{-1pt}(\hspace*{-1pt}A_t\hspace*{-1pt},\hspace*{-2pt}S_t\hspace*{-1pt},\hspace*{-2pt}\pi_t\hspace*{-1pt},\hspace*{-2pt}V_{t+1}\hspace*{-1pt},\hspace*{-2pt}\beta_{t+1}\hspace*{-1pt}(\hspace*{-1pt}\sigma^{{*\hspace*{-1pt},\epsilon}}_t)\hspace*{-1pt})\hspace*{-1pt}|s_t^i\hspace*{-1pt},\hspace*{-2pt}\pi_t\hspace*{-1pt}\}.
\end{align*}

Therefore, for every $\pi_t,\hspace*{-1pt}s_t^i$, we have 
\begin{align*}
\sigma_t^{i,\epsilon}(\pi_t,\hspace*{-1pt}s_t^i)\hspace*{-1pt}\in\hspace*{-1pt} \argmax_{\alpha\in\Delta(\mathcal{A}_t^i):\alpha(a_t^i)\geq\epsilon, \forall a_t^i} \sum_{a_t^i}\tilde{U}_{\pi_t,s_t^i}^{\sigma^{*,\epsilon}}(a_t^i)\alpha(a_t^i).
\end{align*} 
We note that $\max_{\alpha\in\Delta(\mathcal{A}_t^i):\alpha(a_t^i)\geq\epsilon, \forall a_t^i} \sum_{a_t^i}\tilde{U}_{\pi_t,s_t^i}^{\sigma^{*,\epsilon}}(a_t^i)\alpha(a_t^i)$ is a linear program, thus, by Theorem 16 of \cite{nedic}, the set of agent $i$'s best responses $\text{\textbf{BR}}_t^{i,\epsilon}\hspace*{-1pt}(\sigma^{*,\epsilon}_t)$ is closed and convex. If  $V_{t+1}$ is continuous in $\pi_{t+1}$ then $V_{t+1}$ is continuous in agent $i$'s strategy $\sigma_t^i$. Moreover, the instantaneous utility $u_t^i$ is continuous in agent $i$'s strategy $\sigma_t^i$. Therefore, $\bar{U}_t^i$, given by (\ref{eq:stageutility}), is continuous in agent $i$' strategy $\sigma_t^i$. Therefore, by the maximum theorem \cite{ok2007real} the set of $i$'s best responses in upper hemicontinuous in $\sigma^{*,\epsilon}_t$ and non-empty. 

Consequently, we establish that for every $i\hspace*{-1pt}\in\hspace*{-1pt}\mathcal{N}$, $\text{\textbf{BR}}_t^{i,\epsilon}\hspace*{-1pt}(\sigma^{*,\epsilon}_t)$ is closed, convex, upper hemicontinuous, and non-empty for every $\sigma^{*,\epsilon}_t\hspace*{-2pt}\in\hspace*{-2pt}\Sigma_t^\epsilon$. Define $\text{\textbf{BR}}_t^\epsilon\hspace*{-2pt}:=\hspace*{-2pt}\bigtimes_{i\in\mathcal{N}}\text{\textbf{BR}}_t^{i,\epsilon}$ where $\bigtimes$ denotes the Cartesian product. The correspondence $\text{\textbf{BR}}_t^{\epsilon}(\sigma^{*,\epsilon}_t)$ is closed, convex, upper hemicontinuous, and non-empty for every $\sigma^{*,\epsilon}_t\hspace*{-2pt}\in\hspace*{-2pt}\Sigma_t^\epsilon$ since $\text{\textbf{BR}}_t^{i,\epsilon}(\sigma^{*,\epsilon}_t)$ is closed, convex, upper hemicontinuous, and non-empty for every $\sigma^{*i,\epsilon}_t\hspace*{-2pt}\in\hspace*{-2pt}\Sigma_t^{i,\epsilon}$ for all $i\hspace*{-2pt}\in\hspace*{-2pt}\mathcal{N}$. Therefore, by Kakutani's fixed-point theorem \cite[Corollary 15.3]{border1989fixed}, the correspondence  $\text{\textbf{BR}}_t^{\epsilon}$ has a fixed point. Therefore, every perturbed stage game has an equilibrium $\sigma^{*,\epsilon}_t$ along with a consistent update rule $\psi_{t+1}\hspace*{-1pt}=\hspace*{-1pt}\beta_{t+1}(\sigma^{*,\epsilon}_t)$.

Now consider the sequence of these perturbed games when $\epsilon\hspace*{-1pt}\rightarrow \hspace*{-1pt}0$. Since the set of agents' strategies is compact, there exists a subsequence of these perturbed games whose equilibrium strategies converge, say to $\sigma^*_t$. Similarly let $\psi_{t+1}^*$ denote the convergence point of $\beta_{t+1}(\sigma^{*,\epsilon}_t)$.  We note that $\psi_{t+1}^*$ is consistent with $\sigma^*_t$ since $\beta_{t+1}\hspace*{-1pt}(\sigma^{*,\epsilon}_t)$ (\textit{i.e.} Bayes' rule) is continuous in $\sigma^{*,\epsilon}_t$. We show below that for every agent $i\hspace*{-1pt}\in\hspace*{-1pt}\mathcal{N}$, \hspace*{-1pt}$\sigma^{*i}_t$ is a best response for him given $V_{t+1},\hspace*{-1pt}\psi_{t+1}^*$ when he chooses his strategy from the unconstrained class of SIB strategies.

As we proved above, the set of agent $i$'s best responses $\text{\textbf{BR}}_t^{i}(\sigma^{*}_t)$ is upper hemicontinuous and closed given $\psi_{t+1}^*$. Therefore, $\sigma^*_t(\pi_t,\hspace*{-1pt}\cdot)$ is also a best response for agent $i$ in the stage game $\mathbf{G}_t(\pi_t\hspace*{-1pt},\hspace*{-2pt}V_{t+1}\hspace*{-1pt},\hspace*{-2pt}\psi_{t+1}^*)$.  Consequently, $\sigma^*_t\hspace*{-2pt}\in\hspace*{-1pt} \text{\textbf{BNE}}_t(V_{t+1},\psi_{t+1}^*)$ where $\psi_{t+1}^*$ is consistent with $\sigma_t^*$. 
\end{proof}

\begin{proof}[Proof of Theorem \ref{thm:zerosum}]
We have a Bayesian zero-sum game with finite state and action spaces. By  \cite[Theorem \hspace*{-2pt}1]{einy2012continuity} the equilibrium payoff is a continuous function of the agents' common prior/belief. Using this result, we prove, by backward induction, that every stage of the dynamic program described  by (\ref{eq:decomposition-BNE})-(\ref{eq:decomposition-value}), has a solution and $V_t$ is continuous in $\pi_t$ for all $t$.

For $t\hspace*{-2pt}=\hspace*{-1pt}T\hspace*{-2pt}+\hspace*{-2pt}1$ the dynamic program has a solution trivially since the agents have utility for time less than or equal to $T$. Moreover, $V_{T+1}(.,\hspace*{-1pt}.)\hspace*{-2pt}=\hspace*{-1pt}0$ is trivially continuous in $\pi_{T+1}$. 

For $t\hspace*{-2pt}\leq \hspace*{-2pt}T$, assume that $V_{t+1}$ is continuous in $\pi_{t+1}$. Then, by Lemma \ref{lemma:sufficient1} the dynamic program has a solution at $t$. We note that the continuation game from $t$ to $T$ is a dynamic zero-sum game with finite state and actions spaces. Therefore, as we argued above, by \cite[Theorem 1]{einy2012continuity} the agents' equilibrium payoff at $t$ (\textit{i.e.} $V_t$) is unique and is continuous in the agents' common prior given by $\pi_t$. 

Therefore, by induction we establish the assertion of Theorem \ref{thm:zerosum}. 
\end{proof}

\begin{proof}[Proof of Lemma \ref{lemma:sufficient2}]
	
	Assume that $\psi_{1:T}$ is independent of $\sigma$. Then, the evolution of $\Pi_t$ is independent of $\sigma^*$ and known a priori. As a result, we can ignore the consistency condition (\ref{eq:decomposition-update}) in the dynamic program. Given $\psi_{t+1}$, the stage game $\mathbf{G}_t(\hspace*{-1pt}\pi_t\hspace*{-1pt},\hspace*{-2pt}V_{t+1}\hspace*{-1pt},\hspace*{-2pt}\psi_{t+1}\hspace*{-1pt})$ is a static game of incomplete information with finite actions (given by $\mathcal{A}_t^{1:N}$) and finite types (given by $\mathcal{S}_t^{1:N}$) for every $\pi_t$. Therefore, by the standard existence results for finite games \cite[Theorem  1.1]{fudenberg1991game}, there exists an equilibrium for the stage game $\text{\textbf{BNE}}_t(\hspace*{-1pt}V_{t+1}\hspace*{-1pt},\hspace*{-2pt}\psi_{t+1}\hspace*{-1pt})$. Consequently, the correspondence $\text{\textbf{BNE}}_t(\hspace*{-1pt}V_{t+1}\hspace*{-1pt},\hspace*{-2pt}\psi_{t+1}\hspace*{-1pt})$ is non-empty for every $t\hspace*{-2pt}\in\hspace*{-1pt}\mathcal{T}$, thus, the dynamic programming given by (\ref{eq:decomposition-BNE}-\ref{eq:decomposition-value}) has a solution.   
\end{proof}